\documentclass{lmcs}
\pdfoutput=1
\usepackage[utf8]{inputenc}

\usepackage{lastpage}
\lmcsdoi{19}{2}{12}
\lmcsheading{}{\pageref{LastPage}}{}{}%
{Sep.~23,~2021}{May~25,~2023}{}

 \newif \ifmycolour \mycolourfalse
 \newcommand{\Comment}[1]{}

	\def \green{}
	\def \ltermcol{}
	\def \typecol{}
	
	\def \blue{}
	\def \black{}
	
	\def \Purple{}
	
	\def \lterm{}
	\def \neutral{}

\theoremstyle{plain}
\newenvironment{theorem}{\begin{thm}}{\end{thm}}
\newenvironment{corollary}{\begin{cor}}{\end{cor}}
\newenvironment{lemma}{\begin{lem}}{\end{lem}}
\newenvironment{proposition}{\begin{prop}}{\end{prop}}

\newenvironment{example}{\begin{exa}}{\end{exa}}
\newenvironment{definition}{\begin{defi}}{\end{defi}}

 \usepackage{qsymbols}
 
\makeatletter
 \def \resetqsymbol#1#2{\newqsymbol@{#1}{#2}}
 \def \ifnextchar{\@ifnextchar}
\makeatother

 \resetqsymbol{`,}{\mathord{{,}\dsk}}
 \resetqsymbol{`;}{\mathord{{;}\dsk}}
 \resetqsymbol{`:}{\mathord{:}}
 \resetqsymbol{`@}{\Appl }
 \resetqsymbol{`B}{{\green \perp}}
 \resetqsymbol{`G}{\myGamma} 
 \resetqsymbol{`L}{\Abstr}
 \resetqsymbol{`M}{\Mu }
 \resetqsymbol{`N}{\Nu}
 \resetqsymbol{`T}{{\green \top}}
 \resetqsymbol{`v}{{\typecol \varphi}}
 \resetqsymbol{`[}{\raise2\point\hbox{$\llceil$}}
 \resetqsymbol{`]}{\rrfloor}
 \resetqsymbol{`(}{(\hsk}
 \resetqsymbol{`)}{\hsk)}
 \resetqsymbol{`<}{{}\langle\hsk}
 \resetqsymbol{`>}{\hsk\rangle{}}
 \resetqsymbol{`.}{\mathop{\cdot}}
 \resetqsymbol{`S}{\hsk{\textbf{S}}}
 \resetqsymbol{`F}{\textsc{f}}

 \newenvironment{AuxilCases}[1]{\left \{ 
\begin{array}{#1} }{ \end{array} \right . }
 \renewenvironment{cases}{\ifnextchar\bgroup
 {\begin{AuxilCases}}{\begin{AuxilCases}{@{\,}lllll}}}{\end{AuxilCases}}

 \usepackage{derivation}
 \usepackage{relsize}
 \def \mrelsize#1#2{\hbox{\relsize{#1}{#2}}}

 \def \LC	{\mbox{the $\semcolour `l$-calculus}}
 \def \betamusubscript{
\kern-.25\point\raise -2.5\point\hbox{\scriptsize $`b\kern-.4\point`m$}}
 \def \same	{\mathrel{\equiv}}
 \def\LeftTop{\llceil}
 \def\RightBot{\rrfloor}
 \def \Sem#1{\mathord{\LeftTop{#1}\RightBot}}
 
 \def \longharpoon#1{\setbox66=\hbox{\tiny $\rightharpoonup$}%
\ifdim#1>\wd66\rule[1.5pt]{#1}{.4pt}\kern-.94\wd66\fi\box66}

\newlength\arglngth

\def\VectOne#1{
        \setbox155=\hbox{$#1$}%
\ifdim\wd155<4pt \setbox155=\hbox{$\hspace*{2pt}#1\hspace*{2pt}$}\fi
\arglngth=\wd155 \kern1pt
\ifdim\ht155>4pt \raise\ht155\hbox{\longharpoon{\arglngth}} %
\else   \raise4pt\hbox{\longharpoon{\arglngth}}\fi
\kern-\wd155 \kern-.5pt\box155}
\def\VectTwo[#1]#2{\leavevmode
        \setbox155=\hbox{$#2$}%
\ifdim\wd155<4pt
        \setbox155=\hbox{$\hspace*{2pt}#2\hspace*{2pt}$}\fi
\arglngth=\wd155
\{\ifdim\ht155>4pt
        \raise\ht155 \hbox{\longharpoon{\arglngth}} %
\else   \raise4pt\hbox{\longharpoon{\arglngth}}\fi \kern-\wd155 
\box155\}_{#1}}
\def\Vect{\ifnextchar[{\VectTwo}{\VectOne}}

 \def \nsubscr {\mbox{\relsize{-1}{\sc n}}}
 \def \vsubscr {\mbox{\relsize{-1}{\sc v}}}
 \setbox164=\hbox{\lterm {\textsf {C}}}
 \def \ContS{\copy164}
 \def \LCont{{\lceil}\sk}
 \def \RCont{\sk{\rfloor}}
 \def \EmptyCont{\LCont\RCont}
 \def \basicCont#1{\LCont{#1}\RCont}
 \def \ContNarg[#1]{\ltermcol {\ContS_{\nsubscr}\tqsk{\basicCont{#1}}}}
 \def \ContN{\ifnextchar[
 	{\ContNarg}{\ltermcol {\ContS_{\nsubscr}}}}
\def \ContVarg[#1]{\ltermcol {\ContS_{\vsubscr}\tqsk{\basicCont{#1}}}}
 \def \ContV{\ifnextchar[
 	{\ContVarg}{\ltermcol {\ContS_{\vsubscr}}}}
 \def \ContNoSubscBr[#1]
	{\ContS\tqsk{\basicCont{#1}}}
 \def \ContNoSubsc{\ifnextchar[
	{\ContNoSubscBr}{\ContS}}
 \def \ContSubscBr_#1[#2]{\ContS_{#1}\psk[{#2}]}
 \def \ContSubsc_#1{\ifnextchar[
	{\ContSubscBr_{#1}}{\ContS_{#1}}}
 \def \Cont{\ifnextchar_{\ContSubsc}{\ContNoSubsc}} 

 \def \refpr(#1){(#1)}
 \def \sk {\hspace*{1\point}\ifnextchar(
	{\refpr}{}}
 \def \dsk {\hspace*{2\point}}
 \def \psk {\hspace{.4\point}}
 \def \qsk {\hspace*{.25\point}}
 \def \hsk {\hspace{.5\point}}
 \def \tqsk {\hspace*{.75\point}}
 \def \nhsk {\hspace{-.5\point}}
 \def \skpr {\hspace*{1.5\point}}
 \def \dquad {\quad \quad}
 \def \tquad {\quad \quad \quad}
 \def \qquad {\quad \quad \quad \quad}

 \def \lambdapoint{\qsk{.}\qsk}

 \def \Sc#1{{\mrelsize{1}{\textsc{\Purple #1}}}}

 \def \cal {\mathcal}
 
 \def \ftsc#1{\raise-1pt\hbox{\scriptsize \sc #1}}
 \def \IL	{{\blue \Sc{il}}}
 \def \CL	{{\blue \Sc{cl}}}

 \def \Set#1{\setbox73=\hbox{$#1$}\ifdim\wd73
 	>30\point\{\, \box73 \,\}\else \{\qsk \box73 \qsk\}\fi}

 \def \GUnd_#1{\Gamma\kern-1.5\point_{#1}}
 \def \GCom,{\Gamma\kern-1\point,}
 \def \GaccUnd_#1{\Gamma'_{#1}}
 \def \Gacc'{\ifnextchar_{\GaccUnd}{\ifnextchar,{\Gamma'\kern-2.5\point}{\Gamma'}}}
 \def \myGamma{\ifnextchar_%
 	{\GUnd}{\ifnextchar,%
		{\GCom}{\ifnextchar'%
			{\Gacc}{\Gamma}}}}

\def \Abstr #1.{`l #1 \lambdapoint }
\def \NAbstr #1.{`n #1 \lambdapoint }
\def \ApplSBrBr [#1] (#2){ [\qsk{#1}\qsk] \skpr ({#2}) }
\def \ApplSBrNoBr [#1] { [\qsk{#1}\qsk] \skp }
\def \ApplSBr [#1] {\ifnextchar(
	{\ApplSBrBr [{#1}] }{\ApplSBrNoBr [{#1}] }}
\def \ApplBrBr (#1) (#2){ ({#1}) \skpr ({#2}) }
\def \ApplBrNoBr (#1) { ({#1}) \skp }
\def \ApplBr (#1) {\ifnextchar(
	{\ApplBrBr ({#1}) }{\ApplBrNoBr ({#1}) }}
\def \ApplNoBrBr #1 (#2){ {#1} \skp ( {#2} ) }
\def \ApplNoBrNoBr#1{ {#1} \hsk }
\def \ApplNoBr #1{\ifnextchar(
	{\ApplNoBrBr {#1} }{\ApplNoBrNoBr {#1}}}
\def \Appl {\ifnextchar[
	{\ApplSBr}{\ifnextchar(
		{\ApplBr}{\ApplNoBr}}}

 \def \PrednoB#1{
`(#1`)}
 \def \PredsqB[#1]{\PrednoB{#1}}
 \def \PredroB(#1){\PrednoB{#1}}
 \def \Pred{\ifnextchar[
 	{\PredsqB}{\ifnextchar(
	 	{\PredroB}{\PrednoB}}}
 \def \ExistsRb #1 (#2){\exists\,#1 \, \Pred[{#2}]}
 \def \ExistsSb #1 [#2]{\exists\,#1 \, \Pred[{#2}]}
 \def \Exists #1 {\ifnextchar(
 	{\ExistsRb #1 }{\ifnextchar[
		{\ExistsSb #1 }{\exists\,#1 ~ }}}
 \def \ForallRb #1 (#2){\forall\,#1 \, \Pred[{#2}]}
 \def \ForallSb #1 [#2]{\forall\,#1 \, \Pred[{#2}]}
 \def \Forall #1 {\ifnextchar(
 	{\ForallRb #1 }{\ifnextchar[
		{\ForallSb #1 }{\forall\,#1 ~ }}}
 \def \arrow	{\mathop{\rightarrow}}
 \def \arr	{\mathord{\rightarrow}}
 \def \Arrow	{\mathbin{\Rightarrow}}

 \newdimen \itemcorrection \itemcorrection .9\parindent 
 \def \Turn {\mathrel{\vdash}}
 \def \TurnNI {\mathrel{\vdash_{\kern-2\point\ftsc{ni}}}}
 \def \TurnLK {\mathrel{\blue \vdash_{\kern-2\point\ftsc{lk}}}}
 \def \Turnnlm {\mathrel{\vdash_{\kern-2\point`n`l`m}}}

 \def \derivThree #1 |- #2 | #3 {{\colourfortype #1} \Turn {\colourfortype #2} 
\Mid {\colourfortype #3} }
 \def \derivTwo #1 |- #2 {\setbox51=\hbox{$#2$}%
\ifdim\wd51>1pt 
	{\colourfortype #1 } \Turn {\colourfortype #2 }
\else
	{\colourfortype #1 } \Turn {} 
\fi}
 \def \deriv #1 |- #2 {\ifnextchar|{\derivThree #1 |- {#2} }{\derivTwo #1 |- {#2} }}
 \def \derlog {\deriv}

 \def \ele {\mathbin{\in}}
 \def \notele {\mathbin{\not\in}}
 \def \ImplyoneBr(#1){\mathrel{\hspace{1pt}\Rightarrow\,(#1)\hspace{1pt}}}
 \def \Implyone#1{\mathrel{\hspace{1pt}\Rightarrow\,(#1)\hspace{1pt}}}
 \def \Imply{\ifnextchar{\bgroup}{\Implyone}{\ifnextchar(
 	{\ImplyoneBr}{\mathrel{\hspace{1pt}\Rightarrow\hspace{1pt}}}}}
 \def \Implies{\Imply}
 \def \pow {\wp\,}
 \def \Prod {\mathbin{\times}}
 \newcommand{\indexot}[2]{\def \n{n} \def\m{m} 1\seq {#1} \seq {#2}}
 \def \iotn {\indexot{i}{\n}}
 \def \jotn {\indexot{j}{\n}}
 \def \jotm {\indexot{j}{\m}}
 
 \def \Union {\mathrel{\cup}}
 \def \union {\mathbin{\cup}}
 \def \Ax {{\blue \textsl{Ax}}}
 \def \Or {\mathbin{\vee}}
 \def \OrI	{{\vee}\textsl{I}}
 \def \OrE	{{\vee}\textsl{E}}
 \def \And {\mathbin{\wedge}}
 \def \arrI	{{\blue \arr \textsl{I}\sk}}
 \def \arrE	{{\blue \arr \textsl{E}\psk}}
 \def \PbC{{\blue \textsl{PbC}}}
 \def \AndI	{{\wedge}\textsl{I}}
 \def \AndE	{{\wedge}\textsl{E}}
 \def \AndL	{{\wedge}\textsl{L}}
 \def \AndR	{{\wedge}\textsl{R}}
 \def \Cut	{\textsl{cut}}
 \def \negE	{{\blue \neg \textsl{E}\psk}}
 \def \negI	{{\blue \neg \textsl{I}\psk}}
 \def \Weak	{\textsl{\blue Wk}}
 \def \Thin{\textsl{\blue Th}}
 \def \Pass	{{\blue \textsl{Pass}\psk}}
 \def \Act	{{\blue \textsl{Act}\psk}}
 \def \LEM{{\blue \textsl{LEM}}}

\setbox138=\hbox{$=$}
\setbox139=\hbox{{\scriptsize $\Delta$}}%
\setbox149=\hbox{\raise-1.75\point\copy138\kern-.5\wd138\kern-.5\wd139\raise4.25\point\hbox{\copy139}\kern-.5\wd139\kern.5\wd138}
 \def \ByDef {\mathrel{\black \sk \copy149\sk}}

 \def \CBV{\textsc{\blue cbv}}
 \def \CBN{\textsc{\blue cbn}}

 \def \FV#1{\textit {fv}\qsk({\ltermcolour #1})}
 \def \fv(#1){\FV{#1}}
 \def \fn(#1){\textit {\semcolour fn}(\lterm {#1} )}

 \def \FreshVariable {\textit {fresh}}
 \def \idS {\textit{Id}_S}
 \def \subs {(`v\mapsto C)\hspace{1mm}}
 \def \after {\hsk{`o}\hsk}
 \def \Ssymb{\textit{S}}
 \def \Supar_#1 (#2) {\Ssymb_{#1}\hsk (#2)}
 \def \Subr_#1 #2{\Ssymb_{#1}\skpr {#2}}
 \def \Su_#1{\ifnextchar\bgroup{\Subr_{#1} }
	{\ifnextchar(
		{\Supar_{#1} }{\Ssymb_{#1}}
}}
 \def \Sapar (#1){\Ssymb'\hsk (#1)}
 \def \Sabr #1{\Ssymb'\skpr #1}
 \def \Sa'{\ifnextchar\bgroup {\Sabr}
	{\ifnextchar(
		{\Sapar}{\Ssymb'}
}}
 \def \Spar (#1){\Ssymb \, ({#1})}
 \def \Sbr #1{\Ssymb \, {#1}}
 \def \typesubst{\ifnextchar_{\Su}
 	{\ifnextchar'{\Sa}%
 	{\ifnextchar\bgroup{\Sbr}
	{\ifnextchar(
		{\Spar}{\Ssymb}
}}}}

 \def \seq	{\mathbin{\leq}}

 \def \pointyPair<#1,#2>{`<{#1}\,,\,{#2}`>}
 \def \normalPair#1{({#1})}
 \def \Pair{\ifnextchar<{\pointyPair}{\normalPair}}

 \def \UCont{\textit {unifyC}}
 \def \UCThree#1#2#3{\UCont #1~#2~#3}
 \def \unifyContextsargs#1#2{\ifnextchar\bgroup{\UCThree{#1}{#2}}{\setbox31=\hbox{$#1$}\setbox32=\hbox{$#2$}
\ifdim\wd31<12\point\ifdim\wd32<12\point {\UCont}\ #1\ #2
	\else {\UCont}\ #1\ (#2) \fi
\else \ifdim\wd32<12\point {\UCont\ (#1)\ #2}
\else {\UCont}\ (#1)\ (#2) \fi \fi }}

 \def \unifyContexts{\ifnextchar\bgroup{\unifyContextsargs}{\UCont}}
 \def \unifyw{\textit {unify}}
 \def \unifytwo#1#2{\textit {unify}\ #1\ #2 }
 \def \unify{\ifnextchar{\bgroup}{\unifytwo}{\ifnextchar(
	{\unifytwo}{\textit {unify}}}}

 \def \lmu	{{\semcolour `l`m}}
 \def \Cmd	{\textsl {\rm C}} 
 \def \lmn	{{\semcolour `l`m`n}}
 \def \nlm	{`n`l`m}

 \def \betasubscript{_{\kern-1\point`b}}
 \def \reduces {\reduc}
  \def \reducarg(#1){\mathrel{\semcolour {\rightarrow}\psk(#1)}}
 \def \reducnor {\mathrel{\semcolour {\rightarrow}}}
 \def \reduc {\ifnextchar(
 	{\reducarg}{\reducnor}}
 \newcommand{\rtcredlmn}[1][*]{\mathrel{\semcolour {\rightarrow}^{#1}_{\Lsubscr}}}
 \def \redbeta{\arr \betasubscript}
 \newdimen\correction
  \def \paramsupscript#1{\setbox32=\hbox{\scriptsize $#1$}%
\hbox{\ifdim\ht32>4\point
	\raise4.5\point\copy32
\else	\raise4.5\point\copy32
\fi\kern-\wd32}}
\newcommand{\rtcredbeta}[1][\ast]{%
\setbox56=\hbox{$\betasubscript$}%
\setbox57=\hbox{{\scriptsize $#1$}}%
\ifdim \wd56>\wd57 \correction0pt \else \correction \wd57 \advance \correction -\wd56 \fi
 \mathrel{{\semcolour \arr\paramsupscript{#1}\betasubscript}\kern\correction}}

 \def\Vsubscr{\mbox{\relsize{-2}{\sc v}}}
 \def\Nsubscr{\mbox{\relsize{-2}{\sc n}}}
 \def \redname {\mathrel{\semcolour \arr_{\kern-.5pt\Nsubscr}}}
 \def \rtcredname {\mathrel{\semcolour \arr^{\ast}_{\kern-.5pt\Nsubscr}}}
 \def \redvalue {\mathrel{\semcolour \arr_{\kern-.5pt\Vsubscr}}}
 \def \rtcredvalue {\mathrel{\semcolour \arr^{\ast}_{\kern-.5pt\Vsubscr}}}

 \def \musub#1.#2 {{#1}{`.}{#2}}
 \def \testmusub#1{\ifnextchar.{\musub#1}{#1}}
 \def \For{{/}}
 \def \tsubstSq [#1/#2]{\tqsk\{\nhsk\testmusub#1 \For #2\nhsk\}}
 \def \tsubstP #1/#2 {\tqsk\{\nhsk\testmusub#1 \For #2\nhsk\}}
 \def \tsubst {\ifnextchar[
 	{\tsubstSq}{\tsubstP}}
 \def \ltsubstSq [#1.#2/#3]{\{ \nhsk #1{`.}#2 \For #3\nhsk \} \tqsk }
  \def \ltestmusub#1{\ifnextchar.{\musub#1}%
{{}#1\cdot}}
 \def \ltsubstP #1/#2 { \{ \nhsk \ltestmusub#1 \For #2\nhsk \} \tqsk }
 \def \ltsubst {\ifnextchar[
 	{\ltsubstSq}{\ltsubstP}}
 \def \ntsubst [#1/#2]{\tqsk\{#1 \!{/}\! #2\}}
 \def \redbmu {\mathrel{\semcolour {\rightarrow}\betamusubscript}}

 \def \comparr#1#2{%
\setbox33=\hbox{\scriptsize ${\semcolour #1}$}%
\setbox34=\hbox{${\semcolour #2}$}%
{\semcolour \rightarrow}\ifdim \wd33>\wd34 
	\copy34\kern-\wd34
	\hbox{\ifdim\ht33>4\point
		\raise4.5\point\copy33
	\else	\raise4.5\point\copy33
	\fi\kern-\wd32}
\else
	\hbox{\ifdim\ht33>4\point
		\raise4.5\point\copy33
	\else	\raise4.5\point\copy33
	\fi\kern-\wd33\copy34}
\fi}
 \newcommand{\rtcredbmu}[1][\ast]{\mathrel{\comparr{#1}\betamusubscript}}

 \def \Mid{\mathrel{\mid}}
 \def \TurnF {\mathrel{\vdash_{\kern-2\point\ftsc{f}}}}
 \def \Turnlmu {\mathrel{\vdash_{\kern-1\point`l\kern-1\point`m}}}
 \def \derlmuFour #1 |- #2 : #3 | #4 {%
	\setbox71=\hbox{$#1$}\setbox74=\hbox{$#4$}
	#1 \Turnlmu {#2} \colon {#3} \Mid 
	\ifdim\wd74<1pt {~} \else {#4} \hsk \fi 
	}
 \def \derlmuThree #1 |- #2 | #3 {#2 \colon #1 \Turnlmu #3 }
 \def \derlmuTwo #1 |- #2 {\ifnextchar:%
		{\derlmuFour {#1} |- {#2} }%
		{\derlmuThree {#1} |- {#2} }%
	}
 \def \derlmu {\ifnextchar|{\derlmuTwo ~ }{\derlmuTwo }}
 \def \dernlm #1 |- #2 : #3 {{#1} \Turnnlm {#2} : {#3} }
 \def \derlmnThree #1 |- #2 : #3 {{ \coltype #1 \Turnlmn {#2} : {#3} }}
 \def \derlmnTwo #1 |- #2 {%
 	\ifnextchar:%
		{ \derlmnThree {#1} |- {#2} }%
		{{\coltype {#1} \Turnlmn {#2} }}
}
 \def \derlmn {\ifnextchar|{\derlmnTwo {} }{\derlmnTwo }}
 \def \derLK #1 |- #2 { #1 \TurnLK #2 }
 \def \derNI #1 |- #2 { #1 \TurnNI #2 }
 \def \derF #1 |- #2 | #3 {{\colourfortype #1 } \TurnF {\colourfortype #2 } \Mid {\colourfortype #3 }\hsk}
 
 \def \Nu #1.{\mbox{\relsize{.5}{$`n$}} #1 \lambdapoint }
 \def \MuwNameNoBr #1.#2 {`m\qsk#1\lambdapoint[#2]\hsk}
 \def \MuwName #1.[#2]{`m\qsk#1\lambdapoint[#2]\hsk}
 \def \MuwCell #1.<#2|#3>{`m\qsk#1\lambdapoint \hsk \cell<#2|#3>}
 \def \Mu #1.{\ifnextchar[
	{\MuwName {#1}.}{\ifnextchar`%
		{\MuwNameNoBr {#1}.}{\ifnextchar<
			{\MuwCell {#1}.}{`m\hsk {#1}\lambdapoint }}}}
 \def \tp{\textrm{\sf tp}}
 \def \Lmu	{\semcolour {\Lambda`m}}

 \def \paramarrayqed[#1]{\par \vspace*{-21\point} \vspace*{#1} \hbox{ } \hfill \qed\leavevmode}
 \def \arrayqed{\ifnextchar[{\paramarrayqed}{\paramarrayqed[0\point]}}
 
 \def \QED{\qed}

        \def \Def{Definition}
        \def \Exm{Example}
        
        \def \Lmm{Lemma}
        \def \Prop{Proposition}
        \def \Sect{Section}

	\def \citeyear{\cite}
	\def \citeN{\cite}

\def \psline #1(#2)(#3){}
\def \psset #1{}

 \def \Erase {\theta}
 \def \Rename {`r}
 \def \Except	{\mathord{\setminus}}

	\def \coltype{}
	\def \colourfortype{}
	\def \ltermcolour{}
	\def \semcolour{}
	\def \mysf#1{\mbox{\relsize{-.5}{\sf #1}}}

	\def \WHERE {\hspace*{3mm} \textrm{where} }

	\def \lmn{{\cal L}}
	\def \Lsubscr {\mbox{\relsize{-3}{$\cal L$}}}
	\def \Turnlmn{\mathrel{{\Turn}\kern-2.5\point_{\Lsubscr}}}
	\def \redlmn {\mathrel{\semcolour {\rightarrow}\kern-2\point_{\Lsubscr}}}
	\def \eqredlmn {\mathrel{\semcolour {\rightarrow}\kern-1.5\point_{\Lsubscr}^=}}
	\def \ParConglmn {\mathrel{\semcolour {\approx}\kern-1\point_{\Lsubscr}}}
	\def \LTypes {{{\cal T}_{\kern-1\point{\Lsubscr}}}}

	\def\OL#1{\overline{#1}}
	\def\UL#1{\underline{#1}}
	\def \And {\mathbin{\wedge}}

	\def \Name{\textsl{N}}

	\def \Redlmn {\mathrel{\semcolour {\Rightarrow}\kern-1.5\point_{\Lsubscr}}}
	\def \lmndeR {\mathrel{\semcolour {\Leftarrow}\kern-1.5\point_{\Lsubscr}}}
	\def \SEarrow{\mathop{\kern-1.5\point\raise 7\point\hbox{\rotatebox{-60}{$\Arrow$}}}}
	\def \SWarrow{\mathop{\kern-1\point\raise 10\point\hbox{\rotatebox{-120}{$\Arrow$}}}}
	\def \oldConv#1#2#3{#1 \SEarrow #2 \SWarrow #3}
	\def \newConv #1 => #2 <= #3 {#1 \SEarrow #2 \SWarrow #3}
	\def \Conv{\ifnextchar\bgroup{\oldConv}{\newConv}}
	\def \rtcRedlmn {\mathrel{\semcolour {\Rightarrow}\kern-1.5\point_{\Lsubscr}^*}}

	\def\observe{\item}

\begin{document}

 \title{Adding Negation to Lambda Mu} 

 \author[S. van Bakel]{Steffen van Bakel\lmcsorcid{0000-0003-2077-011X}}

 \address{Department of Computing, Imperial College London, 180 Queen's Gate, London SW7 2BZ, UK}

 \email{s.vanbakel@imperial.ac.uk}

 \begin{abstract}
We present $\lmn$, an extension of Parigot's $\lmu$-calculus by adding negation as a type constructor, together with syntactic constructs that represent negation introduction and elimination.

We will define a notion of reduction that extends $\lmu$'s reduction system with two new reduction rules, and show that the system satisfies subject reduction.
Using Aczel's generalisation of Tait and Martin-L\"of's notion of parallel reduction, we show that this extended reduction is confluent.

Although the notion of type assignment has its limitations with respect to representation of proofs in natural deduction with implication and negation, $\TurnNI$, we will show that all propositions that can be shown in $\TurnNI$ have a witness in $\lmn$.

Using Girard's approach of reducibility candidates, we show that all typeable terms are strongly normalisable, and conclude the paper by showing that type assignment for $\lmn$ enjoys the principal typing property.

 \end{abstract}


 \maketitle

 \let\citeyear\cite
 \def \mycite{\cite}

 \section*{Introduction} 
Intuitionistic Logic (\IL)~\cite{Brouwer'07,Brouwer'08,Brouwer'75} plays an important role in Computer Science, given its strong relation with types in functional programming and the $`l$-calculus~\mycite{Curry'34,Barendregt'84} through the Curry-Howard isomorphism~\cite{Howard'80}, \emph{i.e.}~through the fact that typeable functions in a functional programming language correspond to proofs in \IL, and provable properties to inhabitable types.
Its importance is most prominent in the context of proof assistants, of which many are rooted in \IL.
Proof assistants or theorem provers can also be seen as programming languages for which the type system corresponds to a formal logic and ensure proof correctness by capitalising on the Curry-Howard correspondence through their type system.
Under this correspondence, checking that a term has a type is operationally equivalent to checking a proof of a proposition \cite{Wadler'15}.

There are currently many different proof assistants in use, that come in different shapes and forms, each with their own characteristic: Coq \cite{CoqManual} has a particular focus on the theorem proving aspect where proofs can be written with intuitive tactics, whereas Agda \cite{Norell'07} and Idris \cite{Brady'13} are more deeply connected to functional programming languages like Haskell.

The more widely used proof assistants that are based on the Curry-Howard isomorphism are all founded on \emph{intuitionistic type theory} \cite{Martin-Lof'84}. 
However, the use of {\IL} inescapably limits these languages to the fact that they are unable to prove a simple notion, which use is widespread in normal, everyday mathematics: each proposition is either true or false.
This is known as the law of the excluded middle (\Sc{lem}), and is the distinguishing feature of Classical Logic~(\CL)~\cite{Gentzen'35,Gentzen'69}.
There are theorem provers that use classical logic, like Trybulec's Mizar \cite{Trybulec-B'85}, but that is not founded in the Curry-Howard isomorphism; for an overview of proof assistants and their background, see \cite{Geuvers'09}.

It can be argued that {\IL} very rightly rejects this notion and there are many that do exactly that: they stress the value of \IL, where a proof of `$A$ or $B$' must be \emph{constructive}, \emph{i.e.}~constructed from a proof of either $A$ or from a proof of $B$, so stating `$A$ or not $A$', without justifying either first, is unacceptable. 
Likewise, a proof for the statement $ \Exists x\ele C [Q(x)] $ is only acceptable if first $Q(c)$ is shown, for some object $c \ele C$ (\emph{i.e.}~a \emph{witness} for $Q$ has been produced).
Therefore $ \neg \Forall x\ele C [P(x)] \Implies \Exists x\ele C [\neg P(x)] $ cannot be shown in \IL, since knowing that there has to be an element in $C$ for which $P$ does not hold is not the same as knowing which element that is.
Accepting {\IL} as the basis for mathematical reasoning, which for many is the only right thing to do for philosophical reasons, severely limits the collection of provable results, and is therefore not a popular choice amongst mathematicians.
Some theorem provers, perhaps begrudgingly, allow for the addition of the axiom `$A$ or not $A$', witnessed through a term constant; although it allows for provability of mathematical statements, this approach does not lend computational context to proofs, as theorem provers for {\IL} do, and does not really adhere to the Curry-Howard isomorphism. 

In fact, the popularity of \IL, constructive logic and constructive mathematics in computer science can be explained through its strong ties with computability through the Curry-Howard correspondence and the relation between \IL, the $`l$-calculus and functional programming.
That situation changed when Griffin~\citeyear{Griffin'90} observed that the $C$-operator of Felleisen's $`lC$-calculus~\citeyear{Felleisen-Hieb'92}, similar to the \mysf{call/cc} function in \mysf{Scheme}, can be typed with $\neg\neg A\arrow A$ (or rather $((A\arrow `B)\arrow `B)\arrow A$), \emph{double negation elimination}, another property that only holds in \CL, thus highlighting the first link between {\CL} and sequential control in computer science. 
This soon led to the definition of $\lmu$ by Parigot~\citeyear{Parigot'92,Parigot-Brno'93}, a calculus that represents minimal classical logic~\mycite{AHS'07}, followed by an impressive body of work in the area of {\CL} and computer science, with many contributions from various authors.

\def \Candid{\mysf{Candid}}
\newcommand{\ECCk}{\ensuremath{\textrm{ECC}_{\textrm{\scriptsize K}}}}
Looking to investigate the possibility and suitability of developing theorem provers for {\CL} based on $\lmu$, in \cite{Davies-etal'21} the case was made that in terms of implementability, expressiveness, and elegance, proof assistants based on {\CL} have much to contribute.
It presented \Candid, a theorem prover based on $\lmu$, but enriched with dependent types, as an extension of \ECCk\ \cite{Miquey'20} adding co-products and dependent algebraic data types.
It treated a system of classical natural deduction that uses the logical connectors \textit{implication}, \textit{negation}, \textit{conjunction}, and \textit{disjunction}.

\def \nef {\Sc{nef}}
As seen in that paper, the link between first order classical logic and computation is tricky. 
Theorem provers are based in dependent-types systems, but \cite{Herbelin'05} showed that by naively combining dependent types and $\lmu$'s control operators, all types have only one inhabitant. 
Fortunately, \cite{Herbelin'12} shows a way to restrict how dependent types and control operators interact, which regains a logically consistent type theory. 
An important notion to address this problem is the use of \emph{negative elimination free} (\nef) terms that cannot contain a negation elimination. 
Since in $\lmu$ negation elimination gets represented through application, as well as through \emph{naming} (see Example~\ref{dne in lmu}), this restriction is quite drastic.
Although it is unavoidable for a \nef\ term to not contain sub-terms of the form $`M`a.[`b]N$ as the subterm $[`b]N$ corresponds with an application of $(\neg E)$, it also cannot contain an application $`@ M N $, as this could correspond with $(\neg E)$ when $M$ has type $A \arr `B$ and $N$ has type $A$. 
Introducing separate syntax for negation, as done here, strongly expands the set of \nef-terms to those really not dealing with negation, and will strengthen the implementation of \Candid.

Another reason to deal with negation explicitly is the fact that $\lmu$ does not really represent \CL, in that tautologies are not necessarily represented by closed terms.
This is in part due to the fact the system only has implicit negation and `proof by contradiction' $(\PbC)$ to express dealing with conflict, so negation $ \neg A $ is expressed through $ A \arr `B $ (where $`B$ is not a type in the original presentation), negation introduction through abstraction and negation elimination through application.
For example, in Example~\ref{dne in lmu} we will show Parigot's proof for double negation elimination in $\lmu$; the witness $`Ly.`M`a . [`g] `@ y (`Lx. `M`d.[`a]x) $ contains a free name $`g$ of type $`B$.
It is needed because the subterm $ `@ y (`Lx. `M`d.[`a]x) $ has type $`B$, and the only way to deal with that in $\lmu$ is applying the rule for $(\PbC)$, which forces the prefix $`m`a.[`g]$ to the term.
We will see that, dealing explicitly with negation, this problem disappears.

One of the strengths of $\lmu$ is that Call-by-Name reduction (\CBN) is confluent, as shown by Py \cite{Py-PhD'98}.
Summers \cite{Summers-PhD'08} defines $\nlm$, a variant of $\lmu$ that deals with negation as well, but represents also Call-by-Value reduction (\CBV), rendering reduction non-confluent.
The choice we make here is to aim for confluence, so reduction in the calculus $\lmn$ we present here does not model {\CBV} reduction, as is the case for Parigot's $\lmu$. 
This paper presents $\lmn$ and shows all the necessary properties for it, like soundness, confluence, expressiveness, termination, and principal typing.

 \subsection*{Overview}
This paper introduces the calculus $\lmn$, which expands on $\lmu$ by adding negation.
We will start in Section~\ref{NDCL} with an overview of two of the common representations of \CL, where we will focus on natural deduction and proof contraction, and why double negation elimination poses a particular problem for the latter.
We will define $\TurnNI$, a restriction to natural deduction for {\CL} that uses negation and implication, and plays a central role in this paper.
We will revisit Parigot's $\lmu$ also through its underlying logic, and explain how it deals with negation, implicitly through assumptions stored in the co-context, and explicitly through $`. \arr `B$.
We in particular highlight that $\lmu$ is not fully equipped to deal with the latter kind of negation, as witnesses to tautologies not necessarily are closed terms.
We also revisit Summer's $\nlm$-calculus that fully represents $\TurnNI$, together with its non-confluent notion of reduction.

In Section~\ref{lmn} we define the calculus $\lmn$ as an extension of $\lmu$ by adding syntax and inference rules that express negation; it can also be seen as a restriction of $\nlm$.
This calculus comes with four elementary notions of reduction, and we will show soundness results for all of them.
This is followed in Section~\ref{confluence} by the proof that reduction is confluent, and in Section~\ref{representation} by the proof that, although a restriction of $\nlm$, $\lmn$ can still inhabit all provable judgements of $\TurnNI$.
Then in Section~\ref{strong normalisation}, we will show that reduction is strongly normalisable, and conclude in Section~\ref{pp} by showing that type assignment enjoys the principal typing property.

 \section{Natural Deduction for Classical Logic} \label{NDCL}

Natural Deduction for \CL, defined by Gentzen in \cite{Gentzen'35} is a way of describing the structure of formal proofs in mathematics that follow the intuitive, human, lines of reasoning as much as possible. 
It is defined through \emph{inference rules} that are generically of the shape 
 \[
\Inf	[\textit{Rule}]
	{ \textit{Premisses}
	}{ \textit{Conclusion} }
 \]
and describes a step allowed in this formal system, where, assuming that all the statements in the premisses hold, then after applying this step named $(\textit{Rule})$ we can accept that the conclusion holds as well.
A statement, also called a judgement, is of the shape $ \derlog `G |- A $, where $A$ is a formula and $`G$ is a \emph{context}, a collection of formulas that form the assumptions needed for $A$ to hold, and expresses that `if all formulas in the collection $`G$ hold, then so does $A$'.
A number of these can together form the premisses; there is only one judgement in the conclusion.\footnote{An different notation can be found in the literature, where inference rules express the relation between the inferred formulas, without stating the context, and the assumptions are the formulas occurring in the leaves of the derivation tree.
Assumptions can be cancelled through steps like `implication introduction', and are then placed between square brackets or struck through.
Since the latter is a non-local operation on the inference tree that is not easily defined or treated formally, here we prefer the `sequent' notation: it neatly collects in the derived statement the assumptions on which it depends in the context $`G$.}

Proofs are constructed by applying rules to each other, in the sense that the conclusion of one rule can be a premise of another.
The premises on the initial rules (that are not the conclusion of other rules) are called the assumptions of the proof; the (single) conclusion occurs at the bottom. 
Judgements that are considered to be proven are those that appear at the bottom of the derivation tree.

The inference rules of natural deduction systems almost all come in two varieties for each logical operator: introduction and elimination rules, each for any particular logical connective. 
For example, for the logical operators $\And$ (\textit{conjunction}) and $\Or$ (\textit{disjunction}), these rules look like:
 \[ \begin{array}{r@{\tquad}l} 
\Inf	[\AndI]
	{ \derlog `G |- A \quad \derlog `G |- B 
	}{ \derlog `G |- A\And B }
 &
\Inf	{ \derlog `G |- A\And B 
	}{ \derlog `G |- A }
 \quad 
\Inf	[\AndE]
	{ \derlog `G |- A\And B 
	}{ \derlog `G |- B }
 \\ [4mm]
\Inf	{ \derlog `G |- A 
	}{ \derlog `G |- A\Or B }
 \quad 
\Inf	[\OrI]
	{ \derlog `G |- B 
	}{ \derlog `G |- A\Or B }
 & 
\Inf	[\OrE]
	{ \derlog `G |- A\Or B \quad \derlog `G,A |- C \quad \derlog `G,B |- C 
	}{ \derlog `G |- C }
 \end{array} \]

To deal with conclusions that need no premisses since they hold by themselves, an \emph{axiom} rule is added; these form the assumptions of the proof and occur in the `leaves' of the proof tree.
 \[
\Inf	[\Ax]{ \derlog `G,A |- A }
 \]

In his paper, Gentzen also presents the Sequent Calculus, which differs from Natural Deduction in that it derives sequences of the shape 
 \[ A_1,\ldots,A_n \vdash B_1,\ldots,B_m \]
with the intended meaning `if all of the properties $A_1$, \ldots, $A_n$ hold, then at least one of the $ B_1$, \ldots, $B_m$ does as well.' 
For each connector, there is a \emph{left} and \emph{right} introduction rule, as in 
 \[ \def\TurnLK{\Turn} \begin{array}{c@{\quad}c@{\tquad}c}
\Inf	
	{ \derLK `G,A |- `D 
	}{ \derLK `G,A\And B |- `D }
& 
\Inf	[\AndL] 
	{ \derLK `G,B |- `D 
	}{ \derLK `G,A\And B |- `D }
& 
\Inf	[\AndR] 
	{ \derLK `G |- A,`D 
	\quad
	\derLK `G |- B,`D 
	}{ \derLK `G |- A\And B,`D }
 \end{array} \]
There are no elimination rules for connectors, just a generic $(\Cut)$-rule:
 \[ \def\TurnLK{\Turn} 
\Inf	[\Cut]
	{ \derLK `G |- C,`D 
	 \quad
	 \derLK `G,C |- `D 
	}{ \derLK `G |- `D }
 \]
where $C$ of course can be $A\And B$.

For the Sequent Calculus, Gentzen defines a notion of (proof) contraction that removes occurrences of $(\Cut)$, and shows that this is (weakly) normalising: for every proof that shows $ \derLK `G |- `D $, there exists a $(\Cut)$-free proof that shows the same result.
The proof follows a left-most, innermost reduction strategy; it is not shown that cut-elimination is strongly normalising.
He does not show a normalisation result in \cite{Gentzen'35} for Natural Deduction, which would eliminate all introduction-elimination pairs, but there is evidence that he did solve that later \cite{Plato'08}.\footnote{One particular difficulty with defining proof contractions on either the Sequent Calculus or Natural Deduction is that this notion is not \emph{confluent}, in that proof contraction not always leads to the same result.}
Prawitz \cite{Prawitz'65} presented an extensive study of proof contraction for Natural Deduction.

The main issue is that in the Sequent Calculus, all logical connectors come with a left and a right introduction rule, whereas in Natural Deduction, not all proof-constructions follow the introduction-elimination pattern of the inference rules. 
For those that do, proof contraction consists of the removal from a proof of an introduction step followed immediately by an elimination step for the same logical connector; for `$\And$' that looks like:

 \[ \begin{array}{ccc@{\qquad}cccc}
\Inf	[\AndE]
	{\Inf	[\AndI]
		{ \InfBox{D_1}{ \derlog `G |- A }
		 \quad 
		 \InfBox{D_2}{ \derlog `G |- B }
		}{ \derlog `G |- A\And B }
	}{ \derlog `G |- A }
&
\Implies
&
\InfBox{D_1}{ \derlog `G |- A }
&
\Inf	[\AndE]
	{\Inf	[\AndI]
		{ \InfBox{D_1}{ \derlog `G |- A }
		 \quad 
		 \InfBox{D_2}{ \derlog `G |- B }
		}{ \derlog `G |- A\And B }
	}{ \derlog `G |- B }
&
\Implies
&
\InfBox{D_2}{ \derlog `G |- B }
 \end{array} \]
or, for implication:
 \[ \begin{array}{c@{\qquad}cc@{\quad}c}
 \begin{array}{c}
\Inf	[\arrI]
	{ \derlog `G,A |- B
	}{ \derlog `G |- A\arrow B }
 \\ [7mm]
\Inf	[\arrE]
	{ \derlog `G |- A\arrow B 
	 \quad
	 \derlog `G |- A
	}{ \derlog `G |- B }
 \end{array}
&
\Inf	[\arrE]
	{ \Inf	[\arrI]
		{\Inf	{\Inf	[\Ax]{\derlog `G, A |- A }
			}{\InfBox{D_1}<46>{ \derlog `G,A |- B } }
		}{ \derlog `G |- A\arrow B } 
	 ~
	 \InfBox{D_2}{ \derlog `G |- A }
	}{ \derlog `G |- B }
&
\Implies 
&
\Inf	{ \InfBox{D_2}{ \derlog `G |- A } 
	}{ \InfBox{D_1'}<28>{ \derlog `G |- B } }
 \end{array} \]
Notice that, in the rule $(\arrI)$, the formula $A$ ceases to be an assumption, and that, in the composed proof on the right, $A$ is no longer an assumption needed to reach the conclusion, since it has been shown to hold by $D_2$.
As a result $`D_1$ and $`D_1'$ are not identical, since dealing with different contexts; however, the have the same structure in terms of rules applied.

This is not possible for all logical connectors: the way negation is dealt with is, for example, not straightforward.
Negation comes of course with introduction and elimination rules:
 \[ \begin{array}{c@{\qquad}cc@{\quad}c}
 \begin{array}{c}
\Inf	[\negI]
	{ \derlog `G,A |- `B
	}{ \derlog `G |- \neg A }
 \\ [7mm]
\Inf	[\negE]
	{ \derlog `G |- \neg A
	\quad
	\derlog `G |- A
	}{ \derlog `G |- `B }
 \end{array}
&
\Inf	[\negE]
	{\Inf	[\negI]
		{\Inf	{\Inf	[\Ax]{\derlog `G, A |- A }
			}{\InfBox{D_1}<46>{ \derlog `G,A |- `B } }
		}{ \derlog `G |- \neg A }
	 ~
	 \InfBox{D_2}{ \derlog `G |- A }
	}{ \derlog `G |- `B }
&
\Implies
&
\Inf	{ \InfBox{D_2}{ \derlog `G |- A } 
	}{ \InfBox{D_1}<28>{ \derlog `G |- `B } }
 \end{array} \]
but, in Classical Logic, negation plays a more intricate role, in that the \emph{law of excluded middle} `$ A \Or \neg A $ is true for all $A$' holds (or something similar, like `there is no distinction between the formulas $ \neg\neg A $ and $A$').
This is a property that cannot be shown, but has to forced onto the system, and can cause havoc for proof contraction.

There are many different rules that express this to a different degree, like:
 \[ \begin{array}{@{}c@{\dquad}c@{\dquad}c@{\dquad}c@{\dquad}c}
\Inf	
	{ \derlog `G,\neg A |- `B 
	}{ \derlog `G |- A } 
&
\Inf	
	{ \derlog `G |- \neg\neg A 
	}{ \derlog `G |- A } 
&
\Inf	
	{ \derlog {} |- A\Or \neg A } 
&
\Inf	
	{ \derlog {} |- ((A\arr B)\arr A)\arr A } 
&
\Inf	
	{ \derlog {} |- (\neg A\arr A )\arrow A }
 \end{array} \]
(called `\emph{proof by contradiction}', `\emph{double negation elimination}', `\emph{law of excluded middle}', `\emph{Peirce's law}', and `\emph{reductio ad absurdum}', respectively.)
These rules have different expressive power, and adding one rather than another can change the set of derivable properties (see \cite{Ariola-Herbelin'03}).

\subsection{Classical Natural Deduction with Implication and Negation}
The variant of Classical Natural Deduction we will consider in this paper uses the logical connectors $\neg$ (\textit{negation}) and $\arr$ (\textit{implication}).

 \begin{definition}[Natural deduction with negation and implication] \label{CLimplneg}
The formulas we use for our system of natural deduction with negation and implication are: 
 \[ \begin{array}{rcl}
A,B &::=& `v \mid A\arrow B \mid \neg A 
 \end{array} \]
where `$\arr$' associates to the right and `$\neg$' binds stronger than `$\arr$', and $`v$ is basic formula\footnote{We choose to use $`v$ as a primitive formula because we aim to use these as types, using the Curry-Howard correspondence, by `inhabiting' proofs with term information.}, of which there are infinitely countable many.
A context $`G$ is a set of formulas, where $`G,A = `G \union \Set{A} $ and the inference rules are:
 \[ \begin{array}{rl@{\dquad}rl@{\dquad}rl@{\dquad}rl}
(\Ax) : &
\Inf	{ \derlog `G,A |- A }
&
(\arrI) : &
\Inf	{ \derlog `G,A |- B
	}{ \derlog `G |- A\arrow B }
&
(\arrE) : &
\Inf	{ \derlog `G |- A\arrow B 
	\quad
	\derlog `G |- A
	}{ \derlog `G |- B }
 \\ [5mm]
(\PbC) : &
\Inf	{ \derlog `G,\neg A |- `B 
	}{ \derlog `G |- A }
&
(\negI) : &
\Inf	{ \derlog `G,A |- `B
	}{ \derlog `G |- \neg A }
&
(\negE) : &
\Inf	{ \derlog `G |- \neg A
	\quad
	\derlog `G |- A
	}{ \derlog `G |- `B }
 \end{array} \] 
We write $ \derNI `G |- A $ for judgements derivable in this system, and $\TurnNI$ as name for the system.
 \end{definition}
Notice that $`B$ is not a formula in $\TurnNI$, but is a place-holder, used only to represent conflict; it could be omitted from the system, by deriving $ \derlog `G |- {} $ as the conclusion of $(\negE)$. 
Also, the weakening rule 
 \[ \begin{array}{rl}
(\Weak): & 
\Inf	[`G' \subseteq `G]
	{ \derlog `G |- A 
	}{ \derlog `G' |- A }
 \end{array} \]
is admissible.

To compare $\TurnNI$ with the logic with focus $\TurnF$ we will see below, 
we show Peirce's law in $\TurnNI$ (where $ `G = (A\arrow B)\arrow A,\neg A $):
 \[ \def \TurnNI{\Turn} \begin{array}{c}
\Inf	[\arrI]
	{\Inf	[\PbC]
		{\Inf	[\arrE]
			{\Inf	[\Ax]{\derNI `G |- \neg A }
			 ~
			 \Inf	[\arrE]
				{\Inf	[ \Ax]{\derNI `G |- (A\arrow B)\arrow A }
\kern-16mm
				\Inf	[\arrI]
					{\Inf	[\PbC]
						{\Inf	[\negE]
							{\Inf	[ \Ax]{\derNI `G,A,\neg B |- \neg A }
							 \Inf	[ \Ax]{\derNI `G,A,\neg B |- A }
							}{ \derNI `G,A,\neg B |- `B }
						}{ \derNI `G,A |- B }
					}{ \derNI `G |- A\arrow B }
				}{ \derNI `G |- A }
			}{ \derNI (A\arrow B)\arrow A,\neg A |- `B }
		}{ \derNI (A\arrow B)\arrow A |- A }
	}{ \derNI {} |- ((A\arrow B)\arrow A)\arrow A }
 \end{array} \]
(see also \Exm~\ref{Peirce}).

As suggested above, in the presence of $(\PbC)$ defining proof contraction is not straightforward.
Assume we have used $(\PbC)$ to show $ \derlog `G |- A\arr B $, and also have a proof for $ \derlog `G |- A $; then applying $(\arrE)$ constructs a proof that appears to be contractable.
It is, however, not directly clear how to define that.
 
 \begin{example} \label{problem case}
Take the following proof in $\TurnNI$:
 \[ 
\Inf	[\arrE]
	{\Inf	[\PbC]
		{\Inf	[\negE]
			{\InfBox{D_1}{ \derlog `G,\neg(A\arr B) |- \neg\neg(A\arr B) } 
			 \quad
			 \Inf	[\Ax]{ \derlog `G,\neg(A\arr B) |- \neg(A\arr B) }
			}{ \derlog `G,\neg(A\arr B) |- `B }
		}{ \derlog `G |- A\arr B } 
~
	 \InfBox{D_2}{ \derlog `G |- A }
	}{ \derlog `G |- B } 
\] 
It is \emph{a priori} not clear how to contract this proof. 
We would like to use the sub-derivations to be the building stones for the proof for $ \derlog `G |- B $ without the $(\PbC)$-$(\arrE)$ pair, but there is no sub-derivation above the step $(\PbC)$ that has $A$ as an assumption (so does not contain $\derlog `G,A |- A $ as the result of rule $(\Ax)$), or that derives $ \derlog `G |- A\arr B $.
 \end{example}

There are many ways around this problem presented in the literature, but at this point we just want to highlight the problem.
There are, of course, circumstances in which we can remove the $(\PbC)$-$(\arrE)$ pair in a proof in $\TurnNI$.

 \begin{example} \label{remove PbC-arrE pair TurnNI}
Assume we have the following proof (where $`G' \subseteq `G$)\footnote{Notice that the context can only get smaller when applying inference rule, except when using $(\Weak)$.}:
 \[ \def \TurnNI{\Turn} \def \ContS{\copy69}
\Inf	[\arrE]
	{\Inf	[\PbC]
		{\Inf	[\negE]
			{\Inf	[\Ax]{\derNI `G',\neg(A\arr B),\neg C |- \neg C }
\kern -34mm
			 \Inf	{\Inf	[\PbC]
					{\Inf	[\negE]
						{\Inf	[\Ax]{\derNI `G,\neg(A\arr B),\neg C,\neg D |- \neg(A\arr B) }
						 ~
						 \Inf	[\Weak]
							{\InfBox{D_1}{ \derNI `G',\neg C |- A\arr B }
							}{ \derNI `G,\neg(A\arr B),\neg C,\neg D |- A\arr B }
						}{  \derNI `G,\neg(A\arr B),\neg C,\neg D |- `B }
					}{ \derNI `G,\neg(A\arr B),\neg C |- D }
				}{\InfBox{D_2}<108>{\derNI `G',\neg(A\arr B),\neg C |- C } }
			}{\derNI `G',\neg(A\arr B),\neg C |- `B }
		}{\derNI `G',\neg C |- A\arr B }
\kern-35mm
	 \InfBox{D_3}{ \derNI `G',\neg C |- A }
	}{ \derNI `G',\neg C |- B }
 \] 
then we can bring sub-proof $D_3$ to the right of sub-proof $D_1$, apply $(\arrE)$, and construct the proof
 \[ \def \TurnNI{\Turn} 
\Inf	[\PbC]
	{\Inf	[\negE]
		{\Inf	[\Ax]{\derNI `G',\neg B,\neg C |- \neg C }
\kern-21mm
		 \Inf	{\Inf	[\PbC]
			 	{\Inf	[\negE]	
					{\Inf	[\Ax]{\derNI `G,\neg B,\neg C,\neg D |- \neg B }
					 ~
					 \Inf	[\Weak]
						{\Inf	[\arrE]
							{\InfBox{D_1}{ \derNI `G',\neg C |- A\arr B }
						 \quad
						 	\InfBox{D_3}{ \derNI `G',\neg C |- A }
							}{  \derNI `G',\neg C |- B }
						}{  \derNI `G,\neg B,\neg C,\neg D |- B }
					}{ \derNI `G,\neg B,\neg C,\neg D |- `B }
				}{ \derNI `G,\neg B,\neg C |- D }
			}{\InfBox{D_2'}<76>{\derNI `G',\neg B,\neg C |- C } }
		}{\derNI `G',\neg B,\neg C |- `B }
	}{\derNI `G',\neg C |- B }
 \]
whereby removing the $(\PbC)$-$(\arrE)$ pair; the derivation $D_2'$ is in structure equal to $D_2$, in the sense that the same rules get applied in the same order.
Since we have removed the $\arrow$-type, we can argue that the complexity of the proof has decreased.
We will see below (\Exm~\ref{remove PbC-arrE pair TurnF} and \ref{logic struct sub lmu}) that this kind of proof contraction gets successfully modelled in $\lmu$. 
 \end{example}
 
We will see in Section~\ref{nlm} a term calculus that directly represents proofs in $\TurnNI$, and presents a notion of reduction that represents the above proof contraction, by presenting a different kind of term substitution.

To better be able to reason about the structure of proofs and the technicalities of proof contraction, we need to represent the structure of proofs via term information from an appropriate calculus, and inhabit the inference rules with terms, such that proof contractions will come to correspond to term reduction.
This employs the Curry-Howard principle, which expresses a correspondence between terms and their types on one side, and proofs for propositions on the other.
We will see below that associating a term calculus to an inference system unlocks the subtle differences between the variants of Classical Logic we consider here.

The natural way to inhabit $\TurnNI$ is using Summer's $\nlm$ \cite{Summers-PhD'08}; we will first present Parigot's calculus $\lmu$ \cite{Parigot'93}, as this historically came first, and gives a very elegant solution to the proof-contraction problem of Example~\ref{problem case}.

 \section{The foundation of \texorpdfstring{$\lmu$}{λµ}}

Parigot's $\lmu $-calculus is a proof-term syntax for classical logic, expressed in Natural Deduction, defined as an extension of the Curry type assignment system for {\LC}.
With $\lmu $ Parigot created a multi-conclusion typing system which corresponds to a classical logic with \emph{focus}; there derivable statements have the shape {\def\TurnF{\Turn}$ \derF `G |- A | `D $}, where $A$ is the main conclusion of the statement, expressed as the \emph{active} conclusion, $`G$ is the set of assumptions and $`D$ is the set of alternative conclusions, or have the shape {\def\TurnF{\Turn}$ \derF `G |- `B | `D $} if there is no formula under focus.

 \subsection{A classical logic with focus}

Before discussing $\lmu$, in order to better compare it with the other calculi we discuss in this paper, we first revise its underlying logic, which corresponds to the following system.

 \begin{definition}[A classical logic with focus]
The formulas for this system are: 
 \[ \begin{array}{rcl}
A,B &::=& `v \mid A\arrow B 
 \end{array} \]
Contexts $`G$ and co-contexts $`D$ are sets of formulas, and the inference rules are defined through: 
{ \def\TurnF {\Turn}
 \[ \begin{array}{rl@{\dquad}rl@{\dquad}rl}
(\Ax) : &
\Inf	{ \derF `G,A |- A | `D }
&
(\Act) : &
\Inf	{ \derF `G |- `B | A,`D
	}{ \derF `G |- A | `D }
&
(\Pass) : &
 \Inf	{ \derF `G |- A | A,`D
	}{ \derF `G |- `B | A,`D }
 \end{array} \] 
 \[ \begin{array}{rl@{\dquad}rl}
(\arrI) : &
\Inf	{ \derF `G,A |- B | `D
	}{ \derF `G |- A\arrow B | `D }
&
(\arrE) : &
\Inf	{ \derF `G |- A\arrow B | `D
	 \quad
	 \derF `G |- A | `D
	}{ \derF `G |- B | `D }
 \end{array} \] }
 
 \observe
We write $ \derF `G |- M : A | `D $ for judgements derivable in this system.
 \end{definition}
Notice that, as above, $`B$ is not a formula in $\TurnF$, but only represents conflict; also, the formulas in $`D$ are interpreted as \emph{negated}.
The rule $(\Act)$ (\emph{activation}) moves the focus onto the formula $A$, whereas the rule $(\Pass)$ (\emph{passivation}) removes the focus of a proof, since it is in contradiction with the co-context.
It will be clear that, given that $`B$ is not a formula, the rule $(\Pass)$ can only be followed by the rule $(\Act)$, which together change the focus of the proof.

 \begin{example} \label{Peirce} 
Notice that negation is not part of the type language, so does not occur in $`G$ nor in $`D$.
It is therefore not possible to show $(\Sc{dne})$ in this logic (not without first extending the syntax of formulas, see Example~\ref {dne in lmu}); however, it is possible to show Peirce's law:
 \[ \begin{array}{c}
\Inf	[\arrI]
	{\Inf	[\Act]
		{\Inf	[\Pass]
			{\Inf	[\arrE]
				{\Inf	[ \Ax]{\derF (A\arrow B)\arrow A |- (A\arrow B)\arrow A | A }
       \quad
				\Inf	[\arrI]
					{\Inf	[\Act]
						{\Inf	[\Pass]
							{\Inf	[ \Ax]{\derF (A\arrow B)\arrow A,A |- A | A,B }
							}{ \derF (A\arrow B)\arrow A,A |- `B | A,B }
						}{ \derF (A\arrow B)\arrow A,A |- B | A }
					}{ \derF (A\arrow B)\arrow A |- A\arrow B | A }
				}{ \derF (A\arrow B)\arrow A |- A | A }
			}{ \derF (A\arrow B)\arrow A |- `B | A }
		}{ \derF (A\arrow B)\arrow A |- A | {} }
	}{ \derF {} |- ((A\arrow B)\arrow A)\arrow A | {} }
 \end{array} \]
 \end{example}


The intention of this system is to express classical logic, and for this it encapsulates the rule $(\PbC)$.
Since the formulas in $`D$ are seen as negated, any statement $ \derF `G |- A | `D $ can be seen as $ \derNI `G,\neg `D |- A $ (where $\neg `D$ lists the negated versions of all types in $`D$).
With that view, the rules $(\Act)$ and $(\Pass)$ corresponds to allowing the following variants of rules $(\PbC)$ and $(\negE)$
 \[ \begin{array}{c@{\qquad}c}
\Inf	[\PbC]
	{ \derlog `G,\neg `D,\neg A |- `B 
	}{ \derlog `G,\neg `D |- A }
&
\Inf	[\negE]
	{\Inf	[\Ax]
		{ \derlog `G,\neg `D,\neg A |- \neg A }
	 \quad 
	 \derlog `G,\neg `D,\neg A |- A 
	}{ \derlog `G,\neg `D,\neg A |- `B }
 \end{array} \]
but in a version of Natural Deduction where formulas have at most a negation at the front.
Note that it therefore avoides the problem of Example~\ref{problem case} by not allowing the rule $(\PbC)$ to be applied to assumptions on the right in $(\negE)$: $A$ cannot be a negated type, so the premises in the right-hand proof cannot occur swapped.

 \begin{example} \label{remove PbC-arrE pair TurnF}
Using the above observation, following from \Exm~\ref{remove PbC-arrE pair TurnNI}, we can create the proofs in $\TurnF$, and contract the left proof into the right one.
 \[ \def \TurnF{\Turn} \def \ContS{\copy69}
\Inf	[\arrE]
	{\Inf	[\Act]
		{\Inf	[\Pass]
			{\Inf	{\Inf	[\Act]
					{\Inf	[\Pass]
						{\Inf	[\Weak]
							{\InfBox{D_1}{ \derF `G' |- A\arr B | C,`D }
							}{ \derF `G |- A\arr B | A\arr B,C,D,`D }
						}{ \derF `G |- `B | A\arr B,C,D,`D }
					}{\derF `G |- D | A\arr B,C,`D }
				}{\InfBox{D_2}<104>{\derF `G' |- C | A\arr B,C,`D } }
			}{\derF `G' |- `B | A\arr B,C,`D }
		}{\derF `G' |- A\arr B | C,`D }
\kern-2mm
	 \InfBox{D_3}{ \derF `G' |- A | C,`D }
	}{ \derF `G' |- B | C,`D }
\quad
\Inf	[\Act]
	{\Inf	[\Pass]
		{\Inf	{\Inf	[\Act]
				{\Inf	[\Pass]
					{\Inf	[\Weak]
						{\Inf	[\arrE]
							{\InfBox{D_1}{ \derF `G' |- A\arr B | C,`D }
							 \dquad
							 \InfBox{D_3}{ \derF `G' |- A | C,`D }
							}{ \derF `G' |- B | C,`D }
						}{\derF `G |- B | B,C,D,`D }
					}{\derF `G |- `B | B,C,D,`D }
				}{\derF `G |- D | B,C,`D }
			}{\InfBox{D_2'}<80>{\derF `G |- C | B,C,`D } }
		}{\derF `G |- `B | B,C,`D }
	}{\derF `G |- B | C,`D }
 \] 
 \end{example}

We will see in \Exm~\ref{logic struct sub lmu} that this forms the basis of structural reduction in $\lmu$.

 \setbox73=\hbox{$\lmu$}
 \subsection{The \texorpdfstring{$\lmu$}{λµ}-calculus} \label{lmu}

We now present the variant of $\lmu $ we consider in this paper, as defined by Parigot in~\mycite{Parigot-Brno'93} and that gives a Curry-Howard interpretation to the inference rules of $\TurnF$.

 \begin{definition}[Syntax of $\lmu$] \label{lm-terms}
Let $x$ range over the infinite, countable set of \emph{term-variables}, and $`a,`b$ range over the infinite, countable set of \emph{names}.
The $\lmu $-\emph{terms} we consider are defined by the grammar:
 \[ \begin{array}{rcl@{\quad}l}
M,N &::=& V \mid `@ M N \mid `M `a. [`b]M \\
V &::=& x \mid `Lx.M & (\emph{values})
 \end{array} \]

Recognising both $`l$ and $`m$ as binders, the notion of free and bound names and variables of $M$, $\fv(M)$ and $\fn(M)$, respectively, is defined as usual, and we accept Barendregt's convention to keep free and bound names and variables distinct, using (silent) $`a$-conversion whenever necessary. 

We write $x \ele M$ ($`a \ele M$) if $x$ ($`a$) occurs in $M$, either free of bound, and call a term \emph{closed} if it has no free names or variables.
We will call the pseudo-terms of the shape $[`a]M$ \emph{commands}, and write $\Cmd$, and treat them as terms for reasons of brevity, whenever convenient.
 \end{definition}

\noindent
We will use these notations for other calculi as well in this paper.

As with Implicative Intuitionistic Logic, the reduction rules for the terms that represent the proofs correspond to proof contractions, but in $\TurnF$.
The reduction rules for {\LC} are the \emph{logical} reductions, \emph{i.e.}~deal with the removal of a introduction-elimination pair for implication and in addition to these, Parigot expresses also the \emph{structural} rules that change the focus of a proof, where elimination essentially deals with negation and takes place for a type constructor that appears in one of the alternative conclusions (the Greek variable is the name given to a subterm).
Parigot therefore needs to express that the focus of the derivation (proof) changes (see the rules in \Def~\ref{tas lmu}), and this is achieved by extending the syntax with two new constructs $[`a]M$ and $`M`a.M$\footnote{Notice that these constructs are \emph{pseudo} terms and that they always occur together in terms.} that act as witness to \emph{passivation} and \emph{activation} of $\TurnF$, which together move the focus of the derivation, and together are called a \emph{context switch}.

Parigot defines a notion of reduction on these terms, expressed via implicit substitution, and as usual, $M\tsubst[N/x]$ stands for the (instantaneous) substitution of all occurrences of $x$ in $M$ by $N$.
Two kinds of structural substitution are defined: the first is the standard one, where $M \tsubst[N.`g/`a]$ stands for the term obtained from $M$ in which every command of the form $[`a]P$ is replaced by $[`g] P N$ (here $`g $ is a fresh name).
This yields a reduction that is Call by Name (\CBN) in nature, and shown by Py \cite{Py-PhD'98} to be confluent.

The second will be of use for Call-by-Value (\CBV) reduction, where $\ltsubst[N.`g/`a] M$ stands for the term obtained from $M$ in which every $[`a]P$ is replaced by $[`g] N P$.
Although {\CBV} is not considered for the calculus $\lmn$ we define in \Sect~\ref{lmn}, we add its definition here, since it does form part of Summer's calculus that we discuss in \Sect~\ref{nlm}, and forms an intermediate stage between $\lmu$ and $\lmn$.  

They are formally defined by:

 \begin{definition}[Structural substitution] \label {struct sub lmu}
 \emph{Right-structural substitution}, $M \tsubst[ N.`g/`a ]$, and \emph{left-structural substitution}, $\ltsubst[ N.`g/`a ] M $, are defined inductively over (pseudo) terms.
The important cases are:
 \[ \begin{array}{rcl}
[`a]M \tsubst[ N.`g/`a ] & \ByDef & [`g](M \tsubst[ N.`g/`a ] N)
 \\ {}
[`b]M \tsubst[ N.`g/`a ] & \ByDef & [`b] ( M \tsubst[ N.`g/`a ] ) ~ (`b \not= `a)
 \end{array} 
\quad
 \begin{array}{lcl}
\ltsubst[ N.`g/`a ] [`a]M & \ByDef & [`g] N ( \ltsubst[ N.`g/`a ] M )
 \\
\ltsubst[ N.`g/`a ] [`b]M & \ByDef & [`b] \ltsubst[ N.`g/`a ] M ~ (`b \not= `a)
 \end{array} \]
 \end{definition}
\citeN{Parigot'92} only defines the first variant of these notions of structural substitutions (so does not use the prefix `right'); the two notions are defined together, but rather informally, using a notion of contexts in \mycite{Ong-Stewart'97}.

We have the following notions of reduction on $\lmu$.
For the fourth, call by value, different variants exists in the literature; we adopt the one from \cite{Ong-Stewart'97}.

 \begin{definition}[$\lmu $ reduction] \label{mu reduction definition}
 \begin{enumerate}
 \item
The reduction rules of $\lmu$ are:
 \[ \begin{array}{rrcl@{\hspace{6mm}}l}
\textit{logical } (`b): & `@ (`Lx.M ) N & \reduc & M \tsubst[N/x] \\
\textit{structural } (`m): & `@ (`M `a . \Cmd ) N & \reduc & `M `g . \Cmd \tsubst[N.`g/`a] & (`g\textit{ fresh}) \\
 \textit{erasing } (\Erase) : & `M `a . [`a] M & \reduc & M & (`a\notele M) \\
\textit{renaming } (\Rename) : & 
`M`a.[`b]{`M`g.[`d] M } & \reduc &
\multicolumn{2}{@{}l}{
 \begin{cases}
        `M`a . [`b] M\tsubst[`b/`g] & (`d = `g)\\
        `M`a . [`d] M\tsubst[`b/`g] & (`d \not= `g)
 \end{cases}}
 \end{array} \]

 \item
Evaluation contexts are defined as terms with a single hole $\EmptyCont$ by:
 \[ \begin{array}{rcl}
\Cont
	&::=& 
\EmptyCont \mid `@ {\Cont} M \mid `@ M {\Cont} \mid `Lx.{\Cont} \mid `M `a . [`b] \Cont 
 \end{array} \]
We write $ \Cont[M] $ for the term obtained by replacing the hole with the term $M$.

(Free, unconstrained) reduction $\redbmu$ on $\lmu$-terms is defined through
$ \Cont[M] \redname \Cont[N]$ if $M \reduces N$ using either the $`b$, $`m$, $\Erase$, or $\Rename$-reductions rule.

 \item \label{cbn def}
\emph{Call by Name (\CBN) evaluation contexts} are defined as:
 \[ \begin{array}{rcl}
\ContN
	&::=& 
\EmptyCont \mid `@ {\ContN} M \mid `M `a.[`b]\ContN
 \end{array} \]
 \emph{{\CBN}} reduction $\redname$ is defined through:
$ \ContN[M] \redname \ContN[N]$ if $M \reduces N$ using either the $`b$, $`m$, $\Erase$, or $\Rename$-reduction rule.

 \item
\emph{Call by Value (\CBV) evaluation contexts} are defined through:
 \[ \begin{array}{rcl}
\ContV 
	&::=& 
\EmptyCont \mid `@ {\ContV} M \mid `@ V {\ContV} \mid `M `a.[`b]\ContV
 \end{array} \]
 \emph{{\CBV}} reduction $\redvalue$ is defined through:
$ \ContV[M] \redvalue \ContV[N]$ if $M \reduc N$ using either $`m$, $\Erase$, $\Rename$, or:
 \[ \begin{array}[t]{rrcl@{\quad}l}
(`b_{\vsubscr}) : &
`@ ( `Lx.M ) V &\redvalue& M\tsubst[V/x] 
 \\
(`m_{\vsubscr}) : &
`@ V (`M `a.\Cmd) &\redvalue& `M `g . \ltsubst[V.`g/`a] \Cmd
	& ( `g\textit{ fresh})
 \end{array} \]

 \end{enumerate}
 \end{definition}

Remark that, for rule $(`m_{\vsubscr})$, $`M `a .[`b]N$ is not a value. 
Also, unlike for the $`l$-calculus, {\CBV} reduction is not a sub-reduction system of $\redbmu$: the rule $(`m_{\vsubscr})$ (and left-structural substitution) are not part of $\redbmu$.
Both {\CBN} and {\CBV} constitute \emph{reduction strategies} in that they pick exactly one $`b`m$-redex to contract; notice that a term might be in either {\CBN} or \CBV-normal form (\emph{i.e.}~reduction has stopped), but need not be that for $\redbmu$.

Type assignment for $\lmu$ is defined below through inhabiting the inference rules of $\TurnF$ with syntax; there is a \emph{main}, or \emph{active}, conclusion, labelled by a term, and the \emph{alternative} conclusions are labelled by names $`a$, $`b$, \textit{etc}.
Judgements in $\lmu$ are of the shape $\derlmu `G |- M : A | `D $, where $`D$ consists of pairs of Greek characters (the \emph{names}) and types; the left-hand context $`G$, as for the $`l$-calculus, contains pairs of Roman characters and types, and represents the types of the free term variables of $M$.

\begin{samepage} 
 \begin{definition}[Typing rules for $\lmu$] \label{tas lmu} \label{lmu rules}%
 \begin{enumerate}

 \item
Let $`v$ range over a countable (infinite) set of type-variables.
The set of types is defined by the grammar:
 \[ \begin{array}{rcl}
A,B &::=& `v \mid A\arrow B
 \end{array} \]

 \item
A \emph{context} (of term variables) $`G$ is a partial mapping from term variables to types, denoted as a finite set of \emph{statements} $x`:A$, such that the \emph{subjects} of the statements ($x$) are distinct.
We write $`G_1,`G_2$ for the \emph{compatible} union of $`G_1$ and $`G_2$ (if $x`:A_1 \ele `G_1$ and $x`:A_2 \ele `G_2$, then $A_1 = A_2$), and write $`G, x`:A$ for $`G, \Set{x`:A}$, $x \notele `G$ if there exists no $A$ such that $x`:A \ele `G$, and $`G\Except x$ for $`G\Except \Set{x`:A}$.

 \item
A \emph{context of names} $`D$ (or \emph{co-context}) is a partial mapping from \emph{names} to types, denoted as a finite set of \emph{statements} $`a`:A$, such that the \emph{subjects} of the statements ($`a$) are distinct.
Notions $`D_1,`D_2$, as well as $`D, `a`:A$ and $`a \notele `D$, $`D \Except `a $ are defined as for $`G$.

 \item
The type assignment rules for $\lmu$, adapted to our notation, are: 

{\def\Turnlmu{\Turn}
 \[ \kern-2mm
 \begin{array}{rl@{\quad}rl@{\quad}rl}
(\Ax) : &
\Inf	{ \derlmu `G,x`:A |- x : A | `D }
&
(`m) : &
\Inf	[`a\notele`D]
	{ \derlmu `G |- M : B | `a`:A,`b`:B,`D
	}{ \derlmu `G |- {`m `a.[`b]M} : A | `b`:B,`D }
 \quad
\Inf	[`a\notele`D]
	{ \derlmu `G |- M : A | `a`:A,`D
	}{ \derlmu `G |- {`m `a.[`a]M} : A | `D }
 \end{array} \]
 \[ \begin{array}{rl@{\quad}rl}
(\arrI) : &
\Inf	[x\notele`G]
	{ \derlmu `G,x`:A |- M : B | `D
	}{ \derlmu `G |- {`lx.M} : A\arrow B | `D }
&
(\arrE) : &
\Inf	{ \derlmu `G |- M : A\arrow B | `D
	 \quad
	 \derlmu `G |- N : A | `D
	}{ \derlmu `G |- { M N } : B | `D }
 \end{array} \] }

We will write $\derlmu `G |- M : A | `D $ for statements derivable in this system.

 \item
We extend Barendregt's convention on free and bound variables and names to judgements (for all the notions of type assignment we define here), so in $\derlmu `G,x`:A |- M : B | `a`:C,`D $, both $x$ and $`a$ cannot appear bound in $M$.

 \end{enumerate}
 \end{definition}
 \end{samepage}
We can think of $[`a]M$ as storing the type of $M$ amongst the alternative conclusions by giving it the name $`a $.
Notice that $`B$ is not used at all in $\Turnlmu$.

Notice that, if we erase all term information from the inference rules, we get the rules from $\TurnF$, but for the variants of $(`m)$; these we can infer, however, so they are admissible.
 \[ 
\Inf	[\Act]
	{\Inf	[\Pass]
		{ \InfBox{\derlog `G |- B | A,B,`D }
		}{ \derlog `G |- `B | A,B,`D }
	}{ \derlog `G |- A | B,`D }
\qquad
\Inf	[\Act]
	{\Inf	[\Pass]
		{ \InfBox{\derlog `G |- A | A,`D }
		}{ \derlog `G |- `B | A,`D }
	}{ \derlog `G |- A | `D }
 \] 

The following result is standard and of use in the proofs below.

 \begin{lemma} [Weakening and thinning for $\Turnlmu$] \label{lmn thinning lemma} \label{lmn weakening lemma}
The following rules for \emph{weakening} and \emph{thinning} are admissible for $\Turnlmu$:
 \[ \def \Turnlmu {\Turn} \begin{array}{rl@{\dquad}rl}
(\Weak) : &
\Inf	[`G\subseteq `G',`D\subseteq `D']
	{\derlmu `G |- M : A | `D
	}{\derlmu `G' |- M : A | `D' }
& 
(\Thin) : &
\Inf	[\ast]
	{\derlmu `G |- M : A | `D
	}{\derlmu `G' |- M : A | `D' }
 \end{array} \]
 \begin{description}
 \item [$(\ast)$] $ `G' = \Set{x`:B \ele `G \mid x \ele \fv(M)}$, $`D' = \Set{`a`:B \ele `D \mid `a \ele \fn(M)} $.
 \end{description}
 \end{lemma}
 \begin{proof}
Standard.\qed
 \end{proof}

The following soundness result holds.

 \begin{thmC}[\cite{Bakel-PPDP'19}] \label{soundness lmn}
If $M \redbmu N$, and $ \derlmu `G |- M : A | `D $, then $ \derlmu `G |- N : A | `D $.
 \end{thmC} 
This result is in that paper also shown for {\CBV} and \CBN-reduction.

 \begin{example} \label {logic struct sub lmu}
We can illustrate $`m$-reduction by the derivations for the reduction step
 \[ \begin{array}{rcl}
`@ (`M `a . [`b] \Cont[{`M `d . [`a] M }]) N & \redbmu & `M `g . [`b] \Cont[{`M `d . [`g] `@ M N }] 
 \end{array} \]
 \[ 
 \def \Turnlmu{\Turn} \def \ContS{\copy69}
\Inf	[\arrE]
	{\Inf	[`m]
		{\Inf	{\Inf	[`m]
				{\Inf	[\Weak]
					{\InfBox{D_1}{ \derlmu `G' |- M : A\arr B | `b`:C,`D }
					}{ \derlmu `G |- M : A\arr B | `a`:A\arr B,`d`:D,`b`:C,`D }
				}{\derlmu `G |- {`M `d . [`a] M } : D | `a`:A\arr B,`b`:C,`D }
			}{\InfBox{D_2}<186>{\derlmu `G |- \Cont[{`M `d . [`a] M }] : C | `a`:A\arr B,`b`:C,`D } }
		}{\derlmu `G |- {`M `a . [`b] \Cont[{`M `d . [`a] M }]} : A\arr B | `b`:C,`D }
~	 \InfBox{D_3}{ \derlmu `G |- N : A | `b`:C,`D }
	}{ \derlmu `G |- {`@ (`M `a . [`b] \Cont[{`M `d . [`a] M }]) N } : B | `b`:C,`D }
 \] \[
 \def \Turnlmu{\Turn} \def \ContS{\copy69}
\Inf	[`m]
	{\Inf	{\Inf	[`m]
			{\Inf	[\arrE]
				{\Inf	[\Weak]
					{\InfBox{D_1}{ \derlmu `G' |- M : A\arr B | `b`:C,`D }
					}{ \derlmu `G |- M : A\arr B | `d`:B,`g`:D,`b`:C,`D }
				 ~
				 \Inf	[\Weak]
				 	{\InfBox{D_3}{ \derlmu `G |- N : A | `b`:C,`D }
				 	}{\derlmu `G |- N : A | `d`:B,`g`:D,`b`:C,`D }
				}{\derlmu `G |- {`@ M N } : B | `d`:B,`g`:D,`b`:C,`D }
			}{\derlmu `G |- {`M `d . [`g] `@ M N } : D | `d`:B,`b`:C,`D }
		}{\InfBox{D_2'}<168>{\derlmu `G |- \Cont[{`M `d . [`g] `@ M N }] : C | `d`:B,`b`:C,`D } }
	}{\derlmu `G |- {`M `g . [`b] \Cont[{`M `d . [`g] `@ M N }]} : B | `b`:C,`D }
 \]
Notice that these are the `inhabited' version of the proofs in \Exm~\ref{remove PbC-arrE pair TurnF}; remember that a $(\Pass)$-$(\Act)$ pair collapses into $(`m)$.
The derivation $D_2'$ is in structure equal to $D_2$, since that is decided by the syntactic structure of the context $\Cont[`.]$ but contains $`M `d . [`g] `@ M N $ rather than $`M `d . [`a] M $.
 \end{example}

The intuition behind the structural rule is given by de Groote~\mycite{deGroote'94}: ``\emph{in a $\lmu$-term $`m `a.M$ of type $A\arrow B$, only the subterms named by $`a$ are \emph{really} of type $A\arrow B$ (\ldots); hence, when such a $`m$-abstraction is applied to an argument, this argument must be passed over to the sub-terms named by $`a$.}''
Remark that this is accurate, but hides the fact that the naming construction $[`a] M $ is actually a (hidden) instance of rule $(\negE)$, so `naming' is actually a kind of application.
The proof of Peirce's Law (\Exm~\ref{Peirce}) can be inhabited in $\lmu$ with $`lx.`m`a.[`a](x(`ly.`m`b.[`a]y))$.
In \cite{Parigot'92}, Parigot shows that `double negation elimination' can be represented in $\lmu$; as suggested above, $`B$ is added as a pseudo-type to express negation $\neg A$ through $A\arr `B$, as well as contradiction.

 \begin{example}[Double negation elimination in $\lmu$]  \label{dne in lmu}
Double negation elimination is shown in $\TurnNI$ by the proof on the left;
we can also show this in $\TurnF$, as in the proof on the right, but since $\TurnF$ has no rules for negation, we need to add $`B$ to express it, so write $C\arrow `B$ for $\neg C$.
 \[ \begin{array}{c@{\dquad}c}
\def \TurnNI {\Turn}
\Inf	[\arrI]
	{\Inf	[\PbC]
		{\Inf	[\negE]
			{\Inf	[\Ax]{ \derNI \neg\neg C,\neg C |- \neg\neg C }
			 \Inf	[\Ax]{ \derNI \neg\neg C,\neg C |- \neg C }
			}{ \derNI \neg\neg C,\neg C |- `B }
		}{ \derNI \neg\neg C |- C }
	}{ \derNI {} |- \neg\neg C\Implies C }
&
\def \TurnF {\Turn}
\Inf	[\arrI]
	{\Inf	[\Act]
		{\Inf	[\arrE]
			{
\multiput(30,0)(0,50){7}{.}
\raise3\RuleH \hbox to 10mm {\kern-20mm
			\Inf	[\Ax]
				{ \derF (C\arrow `B)\arrow `B |- (C\arrow `B)\arrow `B | C }
}
			 \Inf	[\arrI]
				{\Inf	[\Pass]
					{\Inf	[\Ax]{ \derF (C\arrow `B)\arrow `B,C |- C | C }
					}{ \derF (C\arrow `B)\arrow `B,C |- `B | C }
				}{ \derF (C\arrow `B)\arrow `B |- C\arrow `B | C }
			}{ \derF (C\arrow `B)\arrow `B |- `B | C }
		}{ \derF (C\arrow `B)\arrow `B |- C | {} }
	}{ \derF {} |- ((C\arrow `B)\arrow `B)\arrow C | `B }
 \end{array} \]
Notice that the rule $(\Pass)$ is not directly followed by $(\Act)$, while they always come together in $\lmu$, and that the assumption $ \derNI \neg\neg C,\neg C |- \neg C $ gets replaced by the proof for $ \derF (C\arrow `B)\arrow `B |- C\arrow `B | C $.
Moreover, $(\negE)$ is represented through $(\arrI)$ and $(\arrE)$.

Parigot shows that double negation elimination can be represented in $\lmu$~\cite{Parigot'92}.
 \[ \def\Turnlmu{\Turn}
\Inf	[\arrI]
	{\Inf	[`m]
		{\Inf	[\arrE]
			{\Inf	[\Ax]
				{ \derlmu y `: (C\arrow `B)\arrow `B |- y : (C\arrow `B)\arrow `B | {} }
			 \quad
			 \Inf	[\arrI]
				{\Inf	[`m]
					{\Inf	[\Ax]
					{ \derlmu x`:C |- x : C | `d`:`B,`a`:C,`g`:`B }
					}{ \derlmu x`:C |- { `M`d.[`a]x } : `B | `a`:C,`g`:`B }
				}{ \derlmu {} |- { `Lx. `M`d.[`a]x } : C\arrow `B | `a`:C,`g`:`B }
			}{ \derlmu y`:(C\arrow `B)\arrow `B |- { `@ y (`Lx. `M`d.[`a]x) } : `B | `a`:C,`g`:`B }
		}{ \derlmu y`:(C\arrow `B)\arrow `B |- { `M`a . [`g] `@ y (`Lx. `M`d.[`a]x) } : C | `g`:`B }
	}{ \derlmu {} |- { `Ly.`M`a . [`g] `@ y (`Lx. `M`d.[`a]x) } : ((C\arrow `B)\arrow `B)\arrow C | `g`:`B }
 \]
This corresponds to the proof in $\TurnF$ above, but for the fact that extra calls to $(\Pass)$ and $(\Act)$ are added inside the calls to $(`m)$, as well as additional names of type $`B$; notice that because of these extra rules this term is not closed as it has a free name $`g$.
The proof transformation we hinted at above translates to the following (where $`D' = `a`:C,`g`:`B,`D $):
 \[ \begin{array}{c@{\kern-5mm}c}
 \begin{array}[t]{c}
 \def\TurnNI{\Turn}
\Inf	[\PbC]
	{\Inf	{\Inf	[\negE]
			{\InfBox{ \derNI `G,\neg C |- \neg\neg C }
			 \quad
			 \Inf	[\Ax]{ \derNI `G,\neg C |- \neg C }
			}{\derNI `G,\neg C |- `B } 
		}{\InfBox<54>{ \derNI `G',\neg C |- `B } }
	}{\derNI `G' |- C } 
 \end{array}
& 
 \begin{array}[t]{c}
~ \\ [-8mm]
\def\Turnlmu{\Turn}
\Inf	[`m]
	{\Inf	{\Inf	[\arrE]
			{\InfBox{ \derlmu `G |- M : (C\arr`B)\arr `B | `D' }
			 \quad
			 \Inf	[\arrI]
				{\Inf	[`m]
					{\Inf	[\Ax]{ \derlmu `G,y`:C |- y : C | `d`:`B,`D' }
					}{ \derlmu `G,y`:C |- {`M`d . [`a] y } : `B | `D' }
				}{ \derlmu `G |- { `Ly. {`M`d . [`a] y }} : C\arr`B | `D' }
			}{ \derlmu `G |- {`@ M ( `Ly. {`M`d . [`a] y }) } : `B | `D' } 
		}{\InfBox<158>{ \derlmu `G' |- {\Cont[{`@ M ( `Ly. {`M`d . [`a] y }) }]} : `B | `D' } }
	}{\derlmu `G' |- { `M`a . [`g] \Cont[{`@ M ( `Ly. {`M`d . [`a] y }) }] } : C | `g`:`B,`D }
 \end{array}
 \end{array} \]
so Parigot essentially replaces here an instance of the $(\Ax)$ rule for $ \derNI `G,\neg C |- \neg C $ by a derivation for $ \derlmu `G |- { `Ly. {`M`d . [`a] y }} : C\arr`B | `a`:C,`D $.\footnote{Summers \cite{Summers-PhD'08} uses $ `Nz.`@ [y] (z N) $.}
It is this what allows for the successful encoding of $\TurnNI$ in $\lmu$.

We will see this kind of transformation play an important role in \Sect~\ref{representation}.
 \end{example}

It is important to point out that the use of $`g$ in the previous example creates an anomaly.
Although $ ((C\arrow `B)\arrow `B)\arrow C $ is a logical tautology, the $\lmu$-term that is its witness is \emph{not} a closed term so the proof has an uncanceled assumption.
Moreover, terms can have type $`B$ without being typed with the equivalent of rule $(\negI)$, but using $(\arrE)$.

Several attempts have been made to rectify this.
Parigot not only adds $`B$ to the language of types (in a side remark), but also allows for statements like $`g`:`B$ to be used without adding them explicitly to the co-context, so does not consider them `real' assumptions.
Ariola and Herbelin \cite{Ariola-Herbelin'03} define an extension of $\lmu$, adding a special syntax construct $[\tp]M$, where {\tp} acts as a `continuation constant' and represent the outermost context of the term.
In their system, the witness to $ ((C\arrow `B)\arrow `B)\arrow C $ is the term $ `Ly.`M`a . [\tp] `@ y (`Lx. `M`d.[`a]x) $, a closed term.

Another solution would be to detach, syntactically, passivation from activation, so to no longer insist that they strictly follow each other.
That is the approach in de Groote and Saurin's $\Lmu$-calculus \cite{deGroote'94,Saurin'08}; there the witness would be $ `Ly.`M `a . {`@ y (`Lx. [`a]x)} $ which directly inhabits the proof in $\TurnF$ above. 
That variant of $\lmu$ better expresses the logic of $\TurnF$, but one problem with $\Lmu$ is that is not clear if (denotational) semantics can be defined for it, which is possible for $\lmu$ \cite{Streicher-Reus'98,Bakel-Barbanera-Liguoro-LMCS'18}.
This is directly related to the fact that a $`m$-abstraction can now be applied to a term of type $`B$ that is an application, rather than a term typed (implicitly) using rule $(\negE)$.

 \subsection{The \texorpdfstring{$`n`l`m$}{νλµ}-calculus} \label{nlm}
In \cite{Summers-PhD'08}, Summers makes a strong case for inhabiting the rules of $\TurnNI$ directly and in full, and defines the calculus $\nlm$ by adding the rules for negation and their syntactic representation to a generalisation of $\lmu$.
He thereby extends the syntax with the construct $`@ [M] N $ which is used to represent negation elimination, not just when $M$ is a name, but also when the negated statement on the left is the result of a proof, and allows $(`m)$ to be applied to assumption used on the right in $(\arrE)$.
He also removes the distinction between names and variables, and brings all assumptions together in one context.

 \begin{definition}[Syntax of $\nlm$] \label{nlm-terms} 
The $\nlm$-\emph{terms} we consider are defined over variables (Roman characters) by the grammar:
 \[ \begin{array}{rcl@{\quad}l}
M,N &::=& x \mid `Lx.M \mid `@ M N \mid `N x.M \mid `@ [M] N \mid `M x . M 
 \end{array} \]
 \end{definition}
Type assignment (see Definition~\ref{tas nlm} below) will naturally allow $`m$-binding to terms of the shape $`@ [P] Q $, but since $`B$ is a type, variables and applications can have type $`B$, allowing the term $`m `a . {`@ y (`Lx. [`a]x)} $ to be typeable; since a term like $`Ly.P$ cannot be assigned the type $`B$, a term like $ `M x . {`L y . P} $ will not be typeable.

The reduction rules for $\nlm$ in \cite{Summers-PhD'08} are largely defined, as can be expected, through term substitution as far as the constructors $`l$ and $`n$ are concerned, but contracting a $`m$ redex now becomes more involved than in $\lmu$, for the reasons we discussed in Example~\ref{problem case}.

 \begin{definition}[Reduction in $\nlm$ \cite{Summers-PhD'08}] \label{nlm red}
 \begin{enumerate}
 \item 
The auxiliary notion of substitution \\ $\tsubst[z.N/x]$\footnote{\cite{Summers-PhD'08} uses a slightly different notation.} is defined inductively over the structure of terms, using the base cases
 \[ \begin{array}{rcl@{\quad}l}
x \tsubst[z.N/x] &=& `Nz.N \\
y \tsubst[z.N/x] &=& y & (y \not= x) \\
(`@ [x] M) \tsubst[z.N/x] &=& N \tsubst[{ M \tsubst[z.N/x] }/z]
 \end{array} \] 

 \item 
The reduction rules of $\nlm$ are:\vspace*{-2mm}
 \[ \begin{array}{rrcl}
(`l') : & `@ (`Lx.M ) N & \reduc & `My. {`@ [`Nx.{`@ [y] M }] N } \\
(`n) : & `@ [`Nx.M ] N & \reduc & M \tsubst[N/x] \\
(`m\arr_1) : & `@ (`Mx.M ) N & \reduc & `My.M\tsubst[z.{`@ [y] (zN) }/x] \\
(`m\arr_2) : & `@ N (`Mx.M ) & \reduc & `My.M\tsubst[z.{`@ [y] (Nz) }/x] \\
 \end{array} \quad \begin{array}{rrcl}
(`m\neg_1) : & `@ [`Mx.M ] N & \reduc & M\tsubst[z.{`@ [z] N }/x] \\
(`m\neg_2) : & `@ [N] `Mx.M & \reduc & M\tsubst[z.{`@ [N] z }/x] \\
(`m`n) : & `Ny.`Mx.M & \reduc & `Ny.M\tsubst[z.z/x] \\
(`m`m) : & `My.`Mx.M & \reduc & `My.M\tsubst[z.z/x] \quad \\
(`m`h) : & `Mx.[x] M & \reduc & M \hfill (x\notele M)
 \end{array} \]

Evaluation contexts are defined by:
 \[ \begin{array}{rcl}
 \Cont
	&::=& 
\EmptyCont \mid `Lx.{\Cont} \mid `@ {\Cont} M \mid `@ M {\Cont} \mid `Nx.\Cont\mid `@ [\Cont] M \mid `@ [M] {\Cont} \mid `M x . \Cont 
 \end{array} \]
(Free, unconstrained) reduction $\redbmu$ on $\lmn $-terms is defined through
$ \Cont[M] \redname \Cont[N]$ if $M \reduces N$ using either of the nine rules above.

 \end{enumerate}
 \end{definition}

It is clear that these reduction rules contain the \CBV-rules as well 
in $(`m\arr_2)$ and $(`m\neg_2)$. 
Thereby reduction is not confluent; we have a critical pair in the rules $(`m\arr_1)$ and $(`m\arr_2)$ and the term $ `@ (`Mx.M ) (`My.N ) $ is reducible using both, and these reduction steps will (normally) result in different outcomes.

 \begin{definition}[Type assignment for $\nlm $] \label{tas nlm} \label{nlm rules}%
 \begin{enumerate}

 \item
The set of types is defined by the grammar:
 \[ \begin{array}{rcl}
A,B &::=& `B \mid `v \mid A\arrow B \mid \neg A
 \end{array} \]
A \emph{context} (of term variables) $`G$ is defined as before.

 \item
The type assignment rules for $\nlm $ are:
 \[ \def\Turnnlm{\Turn}
 \begin{array}{rl@{\dquad}rl@{\dquad}rl}
(\Ax) : &
\Inf	{ \dernlm `G,x`:A |- x : A }
&
(\arrI) : &
\Inf	
	{ \dernlm `G,x`:A |- M : B 
	}{ \dernlm `G |- `lx.M : A\arrow B }
&
(\arrE) : &
\Inf	{ \dernlm `G |- M : A\arrow B 
	\quad
	\dernlm `G |- N : A 
	}{ \dernlm `G |- MN : B }
 \\ [5mm]
(`m): &
\Inf	{ \dernlm `G,x`:\neg A |- M : `B 
	}{ \dernlm `G |- `mx.M : A }
&
(\negI) : &
\Inf	{ \dernlm `G,x`:A |- M : `B
	}{ \dernlm `G |- `nx.M : \neg A }
&
(\negE) : &
\Inf	{ \dernlm `G |- M : \neg A 
	\quad
	\dernlm `G |- N : A 
	}{ \dernlm `G |- [M]N : `B }
 \end{array} \] 
We will write $\dernlm `G |- M : A $ for statements derivable in this system.

 \end{enumerate}
 \end{definition}
Notice that $`B$ \emph{is} a type in $\Turnnlm$.
Because the inference rules of $\TurnNI$ can be obtained from those of $\Turnnlm$ by removing all term information, it is immediately clear that all proofs in $\TurnNI$ have a term representation in $\nlm$. 

 \begin{example}
In this calculus, the witness for double negation elimination becomes:
 \[ \def\Turnnlm{\Turn}
\Inf	[\arrI]
	{\Inf	[`m]
		{\Inf	[\negE]
			{\Inf	[\Ax]
				{ \dernlm y`:\neg \neg C,x`:\neg C |- y : \neg \neg C }
			 \Inf	[\Ax]
				{ \dernlm y`:\neg \neg C,x`:\neg C |- x : \neg C }
			}{ \dernlm y`:\neg \neg C,x`:\neg C |- [y]x : `B }
		}{ \dernlm y`:\neg \neg C |- `Mx.[y]x : C }
	}{ \dernlm {} |- { `Ly . `Mx. [y]x } : (\neg \neg C)\arrow C }
 \]
 \end{example}

The presence of reduction rules $`m`n$ and $`m`m$ in Definition~\ref{nlm red} is remarkable, since they do not correspond to proof contractions in a proof system that uses $`B$ only to represent conflict.
Both rule $(`m)$ and $(\negI)$ are only applicable to a statement of the shape $ \dernlm `G |- M : `B $; the rule $(`m)$ above them implies an assumption of the shape $\neg`B$, which is allowed since in \cite{Summers-PhD'08}, $`B$ is a type, so the following are valid derivations.
 \[ \def\Turnnlm{\Turn}
\Inf	[\negI]
	{\Inf	[`m]
		{\InfBox{\dernlm `G,y`:A,x`:\neg `B |- M : `B }
		}{\dernlm `G,y`:A |- `Mx.M : `B }
	}{\dernlm `G |- `Ny.`Mx.M : \neg A }
\tquad
\Inf	[`m]
	{\Inf	[`m]
		{\InfBox{\dernlm `G,y`:\neg A,x`:\neg `B |- M : `B }
		}{\dernlm `G,y`:\neg A |- `Mx.M : `B }
	}{\dernlm `G |- `My.{`Mx.M} : A }
 \]
Moreover, treating $`B$ as a type gives that negation is represented in two different ways in $\nlm$.
In all, after the choices made by Summers, the calculus is rather too permissive and below, we will choose to not treat $`B$ as a type.

 \begin{example}
We can inhabit the proof for $ (A\arrow B)\arrow \neg B\arrow \neg A $ in both $\Turnnlm$ and $\Turnlmu$. 
Let $`G = x`:A\arr B, y`:\neg B,z`:A$, and $`G' = x`:A\arr B, y`:B\arr`B,z`:A$, then we can construct, respectively:
 \[ 
 \begin{array}{c}
 \def\Turnlmn{\Turn}
\Inf	[\arrI]
	{\Inf	[\arrI]
		{\Inf	[\negI]
			{\Inf	[\negE]
				{\Inf	[\Ax]{ \derlmn `G |- y : \neg B }
				 \Inf	[\arrE]
					{\Inf	[\Ax]{\derlmn `G |- x : A\arrow B }
					 \Inf	[\Ax]{\derlmn `G |- z : A }
					}{\derlmn `G |- xz : B }
				}{\derlmn `G |- {`@ [y] (xz) } : `B }
			}{ \derlmn `G{\setminus} z |- {`Nz.`@ [y] (xz) } : \neg A }
		}{ \derlmn x`:A\arrow B |- {`Ly.`Nz.`@ [y] (xz) } : \neg B\arrow \neg A }
	}{ \derlmn {} |- {`Lxy.`Nz.`@ [y] (xz) } : (A\arrow B)\arrow \neg B\arrow \neg A }	
 \end{array} 
 \] \[
 \begin{array}{c}
 \def\Turnlmu{\Turn}
\Inf	[\arrI]
	{\Inf	[\arrI]
		{\Inf	[\arrI]
			{\Inf	[\arrE]
				{\Inf	
					{ \derlmu `G' |- y : B\arr`B | {} }
				 \Inf	[\arrE]
					{\Inf 
						{\derlmu `G' |- x : A\arrow B | {} }
					 \Inf	
						{\derlmu `G' |- z : A | {} }
					}{\derlmu `G' |- xz : B | {} }
				}{\derlmu `G' |- y(xz) : `B | {} }
			}{ \derlmu `G'{\setminus} z |- `Lz.y(xz) : A\arr`B | {} }
		}{ \derlmu x`:A\arrow B |- `Lyz.y(xz) : (B\arr`B)\arrow A\arr`B | {} }
	}{ \derlmu {} |- `Lxyz.y(xz) : (A\arrow B)\arrow (B\arr`B)\arrow A\arr`B | {} }
 \end{array} 
 \]
Since this property holds in intuitionistic logic, it is no surprise that the co-context is not needed in $\Turnlmu$.

For $\TurnNI$, we can also show the counterpart, $ (\neg B\arrow \neg A)\arrow A\arrow B $. Let $`G = x`:\neg B\arrow \neg A,$ $y`:A,$ $z`:\neg B$.
 \[ \def\Turnlmn{\Turn}
\Inf	[\arrI]
	{\Inf	[\arrI]
		{\Inf	[`m]
			{\Inf	[\negE]
				{\Inf	[\arrE]
					{\Inf	[\Ax]{\derlmn `G |- x : \neg B\arrow \neg A }
					 \Inf	[\Ax]{\derlmn `G |- z : \neg B }
				 	}{\derlmn `G |- xz : \neg A }
				 \Inf	[\Ax]{\derlmn `G |- y : A }
				}{\derlmn `G |- {`@ [xz] y } : `B }
			}{\derlmn x`:\neg B\arrow \neg A,y`:A |- {`M z . {`@ [xz] y }} : B }
		}{\derlmn x`:\neg B\arrow \neg A |- {`Ly.`Mz.{`@ [xz] y }} : A\arrow B }
	}{\derlmn {} |- {`Lxy.`M z .{ `@ [xz] y }} : (\neg B\arrow \neg A)\arrow A\arrow B }
 \]
We can do a similar thing in $\Turnlmu$ (with $`G = x`:(B\arr`B)\arrow A\arr`B,y`:A$) :
 \[ \def\Turnlmu{\Turn}
\Inf	[\arrI]
	{\Inf	[\arrI]
		{\Inf	[`m]
			{\Inf	[\arrE]
				{\Inf	[\arrE]
					{
\multiput(30,0)(0,50){7}{.}
\raise3\RuleH \hbox to 10mm {\kern-20mm
					 \Inf	[\Ax]{\derlmu `G |- x : (B\arr`B)\arrow A\arr`B | {} }
}
					 \Inf	[\arrI]
						{\Inf	[`m]
							{\Inf	[\Ax]
							{ \derlmu `G,z`:B |- z : B | `d`:`B,`a`:B,`g`:`B }
							}{ \derlmu `G,z`:B |- { `M`d.[`a]z } : `B | `a`:B,`g`:`B }
						}{ \derlmu `G |- { `Lz. `M`d.[`a]z } : B\arrow `B | `a`:B,`g`:`B }
				 	}{\derlmu `G |- x({ `Lz. `M`d.[`a]z }) : A\arr`B | `a`:B,`g`:`B }
				 \Inf	[\Ax]{\derlmu `G |- y : A | `a`:B,`g`:`B }
				}{\derlmu `G |- { `@ {`@ x ( `Lz. `M`d.[`a]z ) } y } : `B | `a`:B,`g`:`B }
			}{\derlmu `G |- {`M `a . [`g] { `@ {`@ x ( `Lz. `M`d.[`a]z ) } y } } : B | `g`:`B }
		}{\derlmu `G{\setminus} y |- {`Ly.`M `a . [`g] { `@ {`@ x ( `Lz. `M`d.[`a]z ) } y } } : A\arrow B | `g`:`B }
	}{\derlmu {} |- {`Lxy.`M `a . [`g] { `@ {`@ x ( `Lz. `M`d.[`a]z ) } y } } : ((B\arr`B)\arrow A\arr`B)\arrow A\arrow B | `g`:`B }
 \]
Notice that the same kind of transformation has been applied to replace the negated assumption $ \derlmn `G |- z : \neg B $ on the right, and that again the witness for the property is not a closed term in $\lmu$ ($`g$ is free).
 \end{example}

 \section{The $\lmn$-calculus} \label{lmn}

We now present the calculus $\lmn $ we introduce in this paper, which can be seen as a variant of $\lmu$ that gives a Curry-Howard interpretation to the logical system below, which corresponds to $\TurnF$, extended with negation by treating it as a first-class citizen.
Our aim is to fully represent proofs in $\TurnNI$ in a natural way, but defining a calculus with a notion of reduction that is confluent.

We call this calculus $\lmn$ in honour of Mendelson's formal axiomatic theory L for the propositional calculus \cite{Mendelson'64}.
Adapted to our notation, L is defined through:

\begin{samepage}
 \begin{definition} [L cf.~\cite{Mendelson'64}] \label{theory L}
 \begin{enumerate}
 \item (Well formed) formulas are defined through the grammar:
 \[ \begin{array}{rcl}
A,B & ::=& p \mid \neg A \mid A\arr B
 \end{array} \]
where $p$ ranges over \emph{statement letters}.
 \item
If $A$, $B$ and $C$ are formulas of L, then the following are axioms of L:
 \begin{description}
 \item[A1] $ A\arr B\arr A $.
 \item[A2] $ (A\arr B\arr C)\arr (A\arr B )\arr A\arr C $.
 \item[A3] $ (\neg B\arr \neg A )\arr ( \neg B\arr A )\arr B $.
 \end{description}

 \item The only rule of inference of L is \emph{modus ponens} $(\arrE)$: $B$ is a direct consequence of $A$ and $A\arr B$. 

 \end{enumerate}
 \end{definition}
 \end{samepage}
 
The first two rules form the axiom-schemes for intuitionistic implicational logic; the third rule renders the system classical.
For example, using these three rules it is possible to show $ \neg\neg C\arr C $ (for details, see \cite{Mendelson'64}, Lemma 1.11(a)).

The attentive reader will recognise the types of the combinators $K$ and $S$ of Curry's Combinatory Logic \cite{Curry'34,Curry-Feys'58} in the first two axioms; this is the origin of the Curry-Howard isomorphism \cite{Curry'34}.
Of course here we follow Church's approach, by defining an extended $`l$-calculus.

We will base $\lmn$ on a variant of the system $\TurnF$ defined below; notice that, because we use negation explicitly, as in $\nlm$ we no longer have to separate the negated formulas from the non-negated ones.
 \[ \begin{array}{rl@{\dquad}rl@{\dquad}rl}
(\Ax) : &
\Inf	{ \derlog `G,A |- A }
&
(\arrI) : &
\Inf	{ \derlog `G,A |- B
	}{ \derlog `G |- A\arrow B }
&
(\arrE) : &
\Inf	{ \derlog `G |- A\arrow B
	 \quad
	 \derlog `G |- A
	}{ \derlog `G |- B }
 \end{array} \] 
 \[ \begin{array}{rl@{\dquad}rl@{\dquad}rl@{\dquad}rl}
(\negI) : &
\Inf	{ \derlog `G,A |- M 
	}{ \derlog `G |- \neg A }
&
(\negE) : &
\Inf	{ \derlog `G |- \neg A 
	\quad
	\derlog `G |- A 
	}{ \derlog `G |- `B }
&
(\Act) : &
\Inf	{ \derlog `G,\neg A |- `B
	}{ \derlog `G |- A }
&
(\Pass) : &
 \Inf	{ \derlog `G,\neg A |- A
	}{ \derlog `G,\neg A |- `B }
 \end{array} \] 
Notice that $A$ in rules $(\Act)$ and $(\Pass)$ can be a negated formula.
The rule $(\Pass)$ could be omitted, since, as before, we can derive:
 \[
\Inf	[\negE]
	{\Inf	[\Ax]
		{ \derlog `G,\neg A |- \neg A }
	 \quad
	 \InfBox{ \derlog `G,\neg A |- A }
	}{ \derlog `G,\neg A |- `B }
 \]
We keep the rule, however, since we want to preserve the fact that in rule $(\Act)$ we only cancel a negated assumption that was used on the left in $(\negE)$; notice that that characteristic is not expressed in the logic, but will be once we represent the structure of proofs through syntax.

 \begin{definition}[Syntax of $\lmn $] \label{lmn-terms}
The set of $\lmn $-\emph{terms} we consider is defined over variables and names by the grammar:
 \[ \begin{array}{rcl@{\quad}l}
M,N &::=& x \mid `lx.M \mid M N \mid `n x.M \mid [M] N \mid `m`a . M \mid [`a] N 
 \end{array} \]
 

 
 \end{definition}
Notice that $`a$ is not a term.
Since $`B$ is not a type, type assignment (see Definition~\ref{tas lmn} below) will only allow $`m$-binding to terms of the shape $[`a] Q $ or $`@ [P] Q $, so staying close to $\lmu$.

We will use $\lmn$ for the set of terms defined above, as well as for the system based on that, including the reduction and type assignment rules.
In $\lmn $, reduction of terms is expressed via three types of implicit substitution.
As usual, $M\tsubst[N/x]$ stands for the (instantaneous) substitution of all occurrences of $x$ in $M$ by $N$.
The definition of structural substitution for $\lmn$ is defined as for $\lmu$ (Definition~\ref{struct sub lmu}), but with small modifications.

 \begin{definition}[Structural substitution in $\lmn$] \label{struct sub nlm}
\emph{Structural substitution}, $M \tsubst[N.`g/`a]$ and \emph{insertion} $ M\tsubst[N/`a] $ are defined inductively over terms.
We give the main cases:
 \[ \begin{array}{rcl}
([`a]M) \tsubst[N.`g/`a] & \ByDef & [`g](M \tsubst[N.`g/`a] N) \\ {}
([`b]M) \tsubst[N.`g/`a] & \ByDef & [`b] ( M \tsubst[N.`g/`a] ) ~ (`b \not= `a) \\
 \end{array} 
 \quad
 \begin{array}{rcl}
([`a]M) \tsubst N/`a & \ByDef & `@ [M] N \\ {}
([`b]M) \tsubst N/`a & \ByDef & [`b] ( M\tsubst N/`a ) ~ (`b \not= `a) \\
 \end{array} \]
 \end{definition}

We have the following notion of reduction on $\lmn $.

 \begin{definition}[$\lmn $ reduction] \label{lmn reduction}
 \begin{enumerate}
 \item
The reduction rules of $\lmn $ are:
 \[ \begin{array}{rrcl}
(`b) : & `@ (`Lx.M ) N & \reduc & M \tsubst[N/x] \\
(`n) : & `@ [`Nx.M ] N & \reduc & M \tsubst[N/x] \\
(`m) : & `@ (`M`a . M ) N & \reduc & `M`g . M \tsubst[N.`g/`a] \quad (`g\textit{ fresh}) \\
 \end{array} \quad \begin{array}{rrcl}
(`d) : & `@ [`m`a.M ] N & \reduc & M \tsubst[N/`a] \\
(\Erase) : & `M`a . [`a] M & \reduc & M \quad (`a\notele M) \\
(\Rename) : & `@ [`b] `M`g.M & \reduc & M \tsubst[`b/`g]
 \end{array} \]

Evaluation contexts are defined by:
 \[ \begin{array}{rcl}
 \Cont
	&::=& 
\EmptyCont \mid `lx.{\Cont} \mid {\Cont} M \mid M {\Cont} \mid `n x . \Cont \mid [\Cont] M \mid [M] {\Cont} \mid `m `a . \Cont \mid [`a] {\Cont}
 \end{array} \]

Reduction $\redlmn$ on $\lmn $-terms is defined through
$ \Cont[M] \redlmn \Cont[N]$ if $M \reduces N$ using either the $`b$, $`n$, $`m$, $`d$, $\Erase$, or $\Rename$-reduction rule.
As usual, we will use $\eqredlmn$ for the reflexive closure, and $\rtcredlmn$ for the reflexive, transitive closure of $\redlmn$.

 \end{enumerate}
 \end{definition}

Since syntax and reduction rules for $\lmn$ are direct extensions of those for $\lmu$, we can show easily that reduction in $\lmn$ is a conservative extension of reduction in $\lmu$.

 \begin{theorem}
If $M$ and $N$ are $\lmu$ terms such that $M \rtcredbmu N$, then $M \rtcredlmn N$.
 \end{theorem}
 \begin{proof}
Straightforward. \QED
 \end{proof}
\noindent
A similar result cannot be shown for $\nlm$, nor for $\CBV$ reduction in $\lmu$.

Type assignment for $\lmn$ is defined through:

 \begin{definition}[Type assignment for $\lmn $] \label{tas lmn} \label{lmn rules}%
 \begin{enumerate}

 \item
The set of types $\LTypes$ is defined by the grammar:
 \[ \begin{array}{rcl}
A,B &::=& `v \mid A\arrow B \mid \neg A
 \end{array} \]
where `$\arr$' associates to the right and `$\neg$' binds stronger than `$\arr$'.
If $A = \neg B$, we call $A$ a \emph{negated} type, and if $A = \neg B$, but $B \not= \neg C$, we call $A$ a \emph{single negated} type.
If $A = \neg \neg B$, we call $A$ a \emph{double negated type}, where $B$ could be a negated type as well.

 \item
A \emph{context} $`G$ is defined as a partial mapping from term variables to types (which can be negated) and names to negated types, denoted as a finite set of \emph{statements} $x`:A$ and $`a`:\neg B$, such that the \emph{subjects} of the statements are distinct.

We define $\OL{`G}$ through:
 \[ \begin{array}{rcl}
\OL{`{}} &=& `{} \\
\OL{`G,x`:A} &=& \OL{`G},A \\ 
\OL{`G,`a`:\neg A} &=& \OL{`G},\neg A 
 \end{array} \] 

 \item
The type assignment rules for $\lmn $ are:
 \[ \begin{array}{c}
 \def\Turnlmn{\Turn} 
 \begin{array}{rl@{\quad}rl@{\quad}rl}
(\Ax) `: &
\Inf	{ \derlmn `G,x`:A |- x : A }
&
(`m) :  &
\Inf	{ \derlmn `G,`a`:\neg A |- M : `B 
	}{ \derlmn `G |- {`m `a . M } : A }
&
(\Name) `:  &
\Inf	{ \derlmn `G |- N : A 
	}{ \derlmn `G,`a`:\neg A |- [`a]N : `B }
 \end{array} 
 \\ [5mm]
 \def\Turnlmn{\Turn} \begin{array}{rl@{\quad}rl}
(\arrI) `: &
\Inf	
	{ \derlmn `G,x`:A |- M : B 
	}{ \derlmn `G |- {`lx.M} : A\arrow B }
&
(\arrE) `: &
\Inf	{ \derlmn `G |- M : A\arrow B 
	\quad
	\derlmn `G |- N : A 
	}{ \derlmn `G |- MN : B }
 \end{array} 
 \\ [5mm]
 \def\Turnlmn{\Turn} \begin{array}{rl@{\quad}rl}
(\negI) `:  &
\Inf	{ \derlmn `G,x`:A |- M : `B 
	}{ \derlmn `G |- {`n x.M } : \neg A }
&
(\negE) `:  &
\Inf	{ \derlmn `G |- M : \neg A 
	\quad
	\derlmn `G |- N : A 
	}{ \derlmn `G |- [M]N : `B }
 \end{array} 
 \end{array} \]

\noindent
We will write $\derlmn `G |- M : A $ for statements derivable in this system.

 \end{enumerate}
 \end{definition}
Notice that $`B$ is not a type.
Also, in rule $(\Name)$, $`a`:\neg A$ is added to the context; this allows for that statement to already occur there.
In all the rules where a variable or name is bound, by our variable convention it does not occur in the context in the conclusion.
The notation $\OL{`G}$ will be used in Corollary~\ref{from NI to lmn}.

 \begin{example}
In this calculus, $ `Ly . `M`a. {`@ [y] ( `Nx.[`a] x ) } $ is the witness for double negation elimination:
 \[ \def\Turnlmn{\Turn}
\Inf	[\arrI]
	{\Inf	[`m]
		{\Inf	[\negE]
			{\Inf	[\Ax]
				{ \derlmn y`:\neg \neg C,`a`:\neg C |- y : \neg \neg C }
			 \Inf	[\negI]
				{\Inf	[\Name]
					{\Inf	[\Ax]{ \derlmn y`:\neg \neg C,x`:C |- x : C }
					}{ \derlmn y`:\neg \neg C,x`:C,`a`:\neg C |- [`a]x : `B }
				}{ \derlmn y`:\neg \neg C,`a`:\neg C |- `Nx.[`a]x : \neg C }
			}{ \derlmn y`:\neg \neg C,`a`:\neg C |- { `@ [y] ( `Nx.[`a]x ) } : `B }
		}{ \derlmn y`:\neg \neg C |- `M`a.{ `@ [y] ( `Nx.[`a]x ) } : C }
	}{ \derlmn {} |- { `Ly . `M`a. {`@ [y] ( `Nx.[`a]x ) }} : (\neg \neg C)\arrow C }
 \]
Observe that $ `Ly . `M`a. {`@ [y] ( `Nx.[`a] x ) }$ is a closed term.

It is also straightforward to find untypeable terms. For example, we cannot type a term like $`@ [`lx.M] N $ since type assignment would require a negated type for $`lx.M$, nor $ `m`a.`lx.M $ since that would require $`B$ for $`lx.M$, nor $ `ny.`lx.M $, $ `m`a.MN $, $ `m`a.`nx.M $, etc.
 \end{example}

We can also how that type assignment in $\lmn$ is a conservative extension of that in $\lmu$.

 \begin{theorem}
If $\derlmu `G |- M : A | `D $, then $\derlmn `G,\neg `D |- M : A $.
 \end{theorem}
 \begin{proof}
Easy. \QED
 \end{proof}
 \noindent
Again, a similar result cannot be shown for $\nlm$.

It will be clear that, once allowing Greek characters for variables as well, the rule $(\Name)$ is admissible in $\Turnnlm$, as was the case above for rule $(\Pass)$.
Observe that, if $\derlmn `G,`a`:\neg A |- M : `B $, in order for the derivation for $\derlmn `G |- `m`a.M : A $ to be used as a subderivation, either $A = B\arr C$, or $A = \neg B$, for some $B$ and $C$.

We will now show that types are preserved under reduction.
For this we need a weakening result.

 \begin{lemma} [Weakening and thinning for $\Turnlmn$] \label{lmu thinning lemma} \label{lmu weakening lemma}
The following rules are admissible for $\Turnlmn$:
 \[ \def \Turnlmn {\Turn} \begin{array}{rl@{\dquad}rl}
(\Weak) : &
\Inf	[`G\subseteq `G']
	{\derlmn `G |- M : A
	}{\derlmn `G' |- M : A }
& 
(\Thin) : &
\Inf	[`G' = \Set{x`:B \ele `G \mid x \ele \fv(M)}]
	{\derlmn `G |- M : A
	}{\derlmn `G' |- M : A }
 \end{array} \]
 \end{lemma}
 \begin{proof}
Standard.\qed
 \end{proof}
Notice that, by our extension of Barendregt's convention in \Def~\ref{lmu rules}, $`G'$ cannot contain statements for the bound names and variables in $M$.

 \begin{example}
We illustrate the reduction rule $`d$:
 \[ \def \Turnlmn{\Turn}
\Inf	[\negE]
	{\Inf	[`m]
		{\Inf	{\Inf	[\Name]
				{\Inf	[\negI]
					{\InfBox{D}{ \derlmn `G,x`:A |- P : `B }
					}{ \derlmn `G |- `Nx.P : \neg A }
				}{ \derlmn `G,`a`:\neg\neg A |- {`@ [`a] `Nx.P } : `B }
			}{\InfBox<124>{ \derlmn `G,`a`:\neg\neg A |- M : `B } }
		}{ \derlmn `G |- `m`a.M : \neg A }
	 \quad
	 \InfBox { \derlmn `G |- N : A } 
	}{ \derlmn `G |- {`@ [`m`a.M] N } : `B }
\dquad
\Inf	{\Inf	[\negE]
		{\Inf	[\negI]
			{\InfBox{ \derlmn `G,x`:A |- P : `B }
			}{ \derlmn `G |- `Nx.P : \neg A }
		 \quad
		 \InfBox { \derlmn `G |- N : A } 
		}{ \derlmn `G |- {`@ [`Nx.P] N } : `B }
	}{\InfBox<87>{ \derlmn `G |- M\tsubst[N/`a] : `B } }
 \]
It might have been more natural, similar to the approach of \cite{Summers-PhD'08}, to define 
 \[ \begin{array}{rcl}
([`a]M)\tsubst N/`a & \ByDef & `@ [{`Nz.[z]N}] M 
 \end{array} \]
which would have created the subterm $ `@ [{`Nz.[z]N}] P $ in the derivation above; however, notice that $ `@ [{`Nz.[z]N}] P \redlmn `@ [P] N $ and that $\tsubst`Nz.[z]N/`a $ only ever gets applied `to the left'.

 \end{example}

We will now show that type assignment is closed under reduction.
First we show results for the three notions of term substitution.

 \begin{lemma}[Substitution lemma] 
 \label{term substitution lemma} \label{lem:substitution}
 \label{right structural substitution lemma} \label{left structural substitution lemma}
 \begin{enumerate}
 \item
If $\derlmn `G,x`:B |- M : A $ and $\derlmn `G |- L : B $, then \\ $\derlmn `G |- M\tsubst[L/x] : A $.

 \item 
If $\derlmn `G,`a`:\neg (B\arrow C) |- M : A $ and $ \derlmn `G |- L : B $, then $\derlmn `G,`g`: \neg C |- M\tsubst[L.`g/`a] : A $.

 \item 
If $\derlmn `G,`a`:\neg \neg B |- M : A $ and $ \derlmn `G |- L : B $, then $\derlmn `G |- M\tsubst[N/`a] : A $.

 \end{enumerate}
 \end{lemma}

 \begin {proof}

 \begin {enumerate} \itemsep 4\point

 \item 
Standard, by induction on the definition of term substitution.

 \item
By induction on the definition of structural substitution.
All cases follow straightforwardly, except for:

 \begin{description}

 \item [{$ ([`a]N) \tsubst[L.`g/`a] \ByDef `@ [`g] N\tsubst[L.`g/`a] $}]
Then we have $ \derlmn `G,`a`:\neg (B\arrow C) |- N : B\arrow C $ by rule $(\Name)$, and $A = `B$.
Then, by induction, we have 
$ \derlmn `G,`g`:\neg C |- N\tsubst[L.`g/`a] : B\arrow C $.
Since we know that $\derlmn `G |- L : B $, we can construct:
 \[ \def\Turnlmn{\Turn}
\Inf	[\Name]
	{\Inf	[\arrE]
		{\InfBox { \derlmn `G,`g`:\neg C |- N\tsubst[L.`g/`a] : B\arrow C }
		 \quad
		 \Inf	[\Weak]
		 	{\InfBox { \derlmn `G |- L : B }
			}{ \derlmn `G,`g`:\neg C |- L : B }
		}{ \derlmn `G,`g`:\neg C |- { `@ (N\tsubst[L.`g/`a]) L } : C }
	}{ \derlmn `G,`g`:\neg C |- { [`g] `@ (N\tsubst[L.`g/`a]) L } : `B }
 \]
 \end{description}

 \item
By induction on the definition of insertion.
All cases follow straightforwardly, except for:

 \begin{description}
 \item[{$ ([`a]M)\tsubst N/`a \ByDef `@ [M] N $}]
Then $ A = `B $, and the derivation is of the shape:
 \[ \def \Turnlmn{\Turn}
\Inf	[\Name]
	{\InfBox{ \derlmn `G,`a`:\neg\neg B |- P : \neg B }
	}{ \derlmn `G,`a`:\neg\neg A |- [`a]P : `B }
 \]
By induction we have $ \derlmn `G |- P\tsubst[N/`a] : \neg B $, and we can construct:
 \[ \def \Turnlmn{\Turn}
\Inf	[\negE]
	{\InfBox{ \derlmn `G |- P\tsubst[N/`a] : \neg B }
	 \quad
	 \InfBox { \derlmn `G |- N : B } 
	}{ \derlmn `G |- [{P\tsubst[N/`a]}] N : `B }
 \]
\arrayqed[-6\point]

 \end {description}
 \end {enumerate}
 \end {proof}

Notice that the structural substitution $\tsubst[N.`g/`a]$ gets performed by building an application with any subterm $P$ that is named $`a$, resulting in $[`g] P N $ of type $B$. 
Moreover, the insertion $\tsubst[N/`a]$ gets performed for typed terms towards a name that has a double negated type, which disappears.

We will now show that type assignment respects reduction:

 \begin{theorem}[Soundness] \label{soundness lmu}
If $ \derlmn `G |- M : A $, and $M \redlmn N$, then $ \derlmn `G |- N : A $.
 \end{theorem} 

 \begin{proof}
By induction on the definition of $\redlmn$, where we focus on the basic reduction rules.

 \begin{description} \itemsep 2\point

 \item[$(`b)$] Then $M \same `@ ( `lx.P) Q \redname P\tsubst[Q/x] \same N $.
The derivation for $ \derlmn `G |- {`@ (`l x . P ) Q } : A $ is shaped like
 \[ \def \Turnlmn{\Turn}
\Inf	[\arrE]
	{\Inf	[\arrI]
		{\InfBox{ \derlmn `G,x`:B |- P : A }
		}{ \derlmn `G |- `lx.P : B\arrow A }
	 \quad
	 \InfBox{ \derlmn `G |- Q : B }
	}{ \derlmn `G |- { `@ (`lx.P) Q } : A }
 \]
In particular, we have $\derlmn `G,x`:B |- P : A $ and $ \derlmn `G |- Q : B $.
Then we have $ \derlmn `G |- P\tsubst[Q/x] : A $ by \Lmm~\ref{term substitution lemma}.

 \item[$(`n)$] Then $M \same `@ [ `Nx . P ] Q \redname P\tsubst[Q/x] \same N $.
Then $A = `B$ and the derivation for $ \derlmn `G |- {`@ [`N x . P ] Q } : `B $ is shaped like
 \[ \def \Turnlmn{\Turn}
\Inf	[\negE]
	{\Inf	[\negI]
		{\InfBox{ \derlmn `G,x`:B |- P : `B }
		}{ \derlmn `G |- `nx.P : \neg B }
	 \quad
	 \InfBox{ \derlmn `G |- Q : B }
	}{ \derlmn `G |- { `@ [`nx . P ] Q } : `B }
 \]
In particular, we have $\derlmn `G,x`:B |- P : `B $ and $ \derlmn `G |- Q : B $.
Then we have $ \derlmn `G |- P\tsubst[Q/x] : `B $ by \Lmm~\ref{term substitution lemma}.

 \item[$(`m)$] Then $M \same `@ (`m`a.P) Q \redname `M`g . P\tsubst[Q.`g/`a] \same N $.
The derivation for $ `@ (`m`a.P) Q $ is shaped like
 \[ \def \Turnlmn{\Turn}
\Inf	[\arrE]
	{\Inf	[`m]
		{\InfBox{ \derlmn `G,`a`:\neg(B\arrow A) |- P : `B }
		}{ \derlmn `G |- `m`a.P : B\arrow A }
	 \quad
	 \InfBox{ \derlmn `G |- Q : B }
	}{ \derlmn `G |- {`@ (`m`a.P) Q } : A }
 \]
In particular, we have $\derlmn `G,`a`:\neg(B\arrow A) |- P : `B $ and $ \derlmn `G |- Q : B $.
Then by \Lmm~\ref{term substitution lemma}, we have $ \derlmn `G,`g`:\neg A |- P\tsubst[Q.`g/`a] : `B $, and applying rule $(`m)$ gives the result.

 \item[$(`d)$] Then $M \same `@ [ `M`a . P ] Q \redname P\tsubst[Q/`a] \same N $.
Then $A = `B$ and the derivation for $ \derlmn `G |- {`@ [`M `a . P ] Q } : `B $ is shaped like
 \[ \def \Turnlmn{\Turn}
\Inf	[\negE]
	{\Inf	[\negI]
		{\InfBox{ \derlmn `G,`a`:\neg\neg B |- P : `B }
		}{ \derlmn `G |- `m`a.P : \neg B }
	 \quad
	 \InfBox{ \derlmn `G |- Q : B }
	}{ \derlmn `G |- { `@ [`m`a . P ] Q } : `B }
 \]
In particular, we have $\derlmn `G,`a`:\neg\neg B |- P : `B $ and $ \derlmn `G |- Q : B $.
Then, by \Lmm~\ref{term substitution lemma}, we have $ \derlmn `G |- P\tsubst[Q/`a] : `B $.

 \item[$(\Erase)$] Then $M \same `M`a.[`a]P \redname P \same N $ with $`a \notele M$.
The derivation for $ `M`a.[`a]P $ is shaped like
 \[ \def \Turnlmn{\Turn}
\Inf	[`m]
	{\Inf	[\Name]
		{\InfBox{ \derlmn `G |- P : A }
		}{ \derlmn `G,`a`:\neg A |- [`a]P : `B } 
	}{ \derlmn `G |- `M`a.[`a]P : A }
 \]
We have $ \derlmn `G |- P : A $ through a sub-derivation.

 \item[$(\Rename)$] Then $M \same [`b]`M`g.M \redname M \tsubst[`b/`g] \same N $.
The derivation for $ [`b]`m`g.P $ is shaped like
 \[ \def \Turnlmn{\Turn}
\Inf	[\Name]
	{\Inf	[`m]
		{\InfBox{ \derlmn `G,`b`:\neg B,`g`:\neg B |- P : `B }
		}{ \derlmn `G,`b`:\neg B |- `m`g.P : B }
	}{ \derlmn `G,`b`:\neg B |- [`b]`m`g.P : `B }
 \]
So in particular, replacing all 
$`g$ by $`b$, we obtain $ \derlmn `G,`b`:\neg B |- P\tsubst[`b/`g] : `B $.
\qed

 \end{description}
 \end{proof}

 \section{Confluence} \label{confluence}
\def \Rel{\mathbin{R}}
 
In this section we will show that reduction in $\lmn$ satisfies the Church-Rosser property, \emph{i.e.}~is confluent.
This property is defined as follows:

 \begin{definition}[Diamond and Church-Rosser Properties \cite{Barendregt'84}] \label{CR def}
Let $\Rel$ be binary relation on a set $V$. 
 \begin{enumerate}
 \item
$\Rel$ satisfies the \emph{diamond property} if for all $t, u, v \ele V $, if $ t \Rel u $ and $ t \Rel v $, then there exists $ w \ele V $ such that $ u \Rel w $ and $ v \Rel w $.

 \item
$\Rel$ satisfies the \emph{Church-Rosser property} (is \emph{confluent}) if its reflexive, transitive closure $\Rel^*$ satisfies the diamond property. 

 \end{enumerate}
 \end{definition}
 
 \noindent
This immediately implies that if a relation is confluent, then so is its transitive closure. 

The standard approach to showing confluence is that of Tait and Martin-L\"of (see \cite{Barendregt'84,Pfenning'92}) by defining a notion of \emph{parallel reduction} that is based on the standard reduction, which is a reflexive relation defined (in the case of $`b$-reduction) through the rules:
 \[ 
 \begin{array}{c@{\dquad}c@{\dquad}c@{\dquad}c}
\Inf	{x \Arrow x}
 & 
\Inf	{M \Arrow M' 
	}{`Lx.M \Arrow `Lx.M' } 
 &
\Inf	{M \Arrow M' \quad N \Arrow N' 
	}{`@ M N \Arrow `@ M' N' } 
 &
\Inf	{ M \Arrow M' \quad N \Arrow N' 
	}{`@ (`lx.M) N \Arrow M' \tsubst[N'/x] } 
 \end{array} \] 

By the last rule, `$\Arrow$' encompasses `$\redbeta$'; also, if $N$ reduces to $N'$, then $ `@ (`lx.M) N $ reduces to $ M \tsubst[N'/x] $, so all contractions in the various copies of $N$ inside $ M \tsubst[N/x] $ are contracted simultaneously when contracting the redex $ `@ (`lx.M) N $; we are even allowed to contract a redex in $M$, and contracting all these together is considered a \emph{single step} in `$\Arrow$'.
The proof of confluence for $`b$-reduction then contains of showing that $\Arrow$ satisfies the diamond property, and that $ {\Arrow} = {\rtcredbeta} $.

Using this technique, confluence has been claimed for $\lmu$ in \cite{Parigot'92}, but, as noticed in \cite{Py-PhD'98,Baba-etal'01}, that proof was not complete.
The main reason is that the proof overlooks the fact that, perhaps unexpectedly, contraction of one redex can remove another.

 \begin{example}\label{Parigot wrong}
Take the term $ `@(`m`a.[`a]`m`b.[`a]M) N $; observe that contracting the outermost $`m$-redex $ `@(`m`a.\ldots) N $ destroys the innermost $\Rename$-redex $ `m`a.[`a]`m`b.[`a] M $.
The latter is contractable because the sub-term $ `m`b.[`a] M $ is a $`m$-abstraction:
 \[ \begin{array}{rcl}
`@(`m`a.[`a]`m`b.[`a]M) N & \Arrow (\Rename) & `@(`M`a.[`a]M) N 
 \end{array} \]
This is no longer true after the contraction of the outermost redex:
 \[ \begin{array}{rcl}
`@(`m`a.[`a]`m`b.[`a]M) N & \Arrow(`m) & `m`g.[`g]`@(`m`b.[`g]MN) N 
 \end{array} \]
where $N$ gets (also) placed as an argument to $ `m`b.[`g]MN$, creating the application\linebreak$`@ (`m`b.[`g]MN) N $ which means that the result is no longer a $`m$-abstraction, thus destroying the $\Rename$-redex.
The resulting terms can be joined, but not through a single parallel reduction step, as would be required.
So the diamond property does not hold for the standard notion of parallel reduction.
 \end{example}

This problem was successfully addressed by Py \cite{Py-PhD'98} and later in \cite{Baba-etal'01} using a slightly different approach.
We will follow the solution of the first here, using the modification of the definition of `$\Arrow$' as suggested by Aczel \cite{Aczel'78}.

As in \cite{Baba-etal'01,Bakel-ToCL'18,Bakel-Barbanera-Liguoro-LMCS'18},
we will not consider the extensional erasure reduction rule 
 \[ \begin{array}{rrcl@{\quad}l}
\Erase: & `M`a . [`a] M & \reduc & M & (`a\notele M) \\
 \end{array} \]

Below we will need the property that we can change the order in which the four implicit substitution operations are performed.
Notice that we can consider the substitution a binding operation for the variable or name involved, so for example the variable $x$ in $ M\tsubst[N/x]\tsubst[P/y] $ can be assumed to not occur in $P$.

 \begin{proposition} \label{term subst lemma}
 \begin{enumerate}
 \item 
	\begin{enumerate}
	\item $ M\tsubst[N/x]\tsubst[P/y] = M\tsubst [P/y]\tsubst[N{\tsubst [P/y]}/x] $.
	\item $ M\tsubst[N/x]\tsubst[P.`d/`a] = M\tsubst [P.`d/`a]\tsubst[N{\tsubst [P.`d/`a]}/x] $.
	\item $ M\tsubst[N/x]\tsubst[P/`a] = M\tsubst [P/`a]\tsubst[N{\tsubst [P/`a]}/x] $.
	\item $ M\tsubst[N/x]\tsubst[`d/`a] = M\tsubst [`d/`a]\tsubst[N{\tsubst [`d/`a]}/x] $.
	\end{enumerate}

 \item 
	\begin{enumerate}
	\item $ M\tsubst[N.`g/`b]\tsubst[P/y] = M\tsubst [P/y]\tsubst[{N\tsubst [P/y]}.`g/`b] $.
	\item $ M\tsubst[N.`g/`b]\tsubst[P.`d/`a] = M\tsubst [P.`d/`a]\tsubst[{N{\tsubst [P.`d/`a]} }.`d/`b] $.
	\item $ M\tsubst[N.`g/`b]\tsubst[P/`a] = M\tsubst [P/`a]\tsubst[{N\tsubst [P/`a]}.`g/`b] $.
	\item $ M\tsubst[N.`g/`b]\tsubst[`d/`a] = M\tsubst [`d/`a]\tsubst[{N\tsubst [`d/`a]}.`g/`b] $.
	\end{enumerate}

 \item 
	\begin{enumerate}
	\item $ M\tsubst[N/`b]\tsubst[P/y] = M\tsubst [P/y]\tsubst[N{\tsubst [P/y]}/`b] $.
	\item $ M\tsubst[N/`b]\tsubst[P.`d/`a] = M\tsubst [P.`d/`a]\tsubst[N{\tsubst [P.`d/`a]}/`b] $.
	\item $ M\tsubst[N/`b]\tsubst[P/`a] = M\tsubst [P/`a]\tsubst[N{\tsubst [P/`a]}/`b] $.
	\item $ M\tsubst[N/`b]\tsubst[`d/`a] = M\tsubst [`d/`a]\tsubst[N{\tsubst [`d/`a]}/`b] $.
	\end{enumerate}
 \item 
	\begin{enumerate}
	\item $ M\tsubst[`g/`b]\tsubst[P/y] = M\tsubst [P/y]\tsubst[`g/`b] $.
	\item $ M\tsubst[`g/`b]\tsubst[P.`d/`a] = M\tsubst [P.`d/`a]\tsubst[`g/`b]\tsubst[P.`d/`a] $.

	\item $ M\tsubst[`g/`b]\tsubst[P/`a] = M\tsubst [P/`a]\tsubst[`g/`b]\tsubst[P/`a] $.
	\item $ M\tsubst[`g/`b]\tsubst[`d/`a] = M\tsubst [`d/`a]\tsubst[`g/`b]\tsubst[`d/`a] $.
	\end{enumerate}

 \end{enumerate}
 \end{proposition}
 \begin{proof}
Straightforward by induction on the definition of the four substitutions.
\qed
 \end{proof} 

We now define a notion of parallel reduction for $\lmn$.

 \begin{definition}[Generalised Parallel Reduction for $\lmn$ (cf.~\cite{Aczel'78,Py-PhD'98})] 
We define \emph{parallel reduction} on terms in $\lmn$ inductively by the rules:

\newcounter{rulecounter} \setcounter{rulecounter}{0} 
\def \ruleitem#1 {\refstepcounter{rulecounter} (\therulecounter)}

\kern-5mm 
 \[ \def\arraystretch{2.25} \def \Redlmn {\Rightarrow}
 \begin{array}{rl}
\ruleitem{variable} & 
\Inf	{x \Redlmn x}
 \\ 
\ruleitem{lambda abstraction} & 
\Inf	{M \Redlmn M' 
	}{`Lx.M \Redlmn `Lx.M' } 
 \\
\ruleitem{mu abstraction} & 
\Inf	{M \Redlmn M' 
	}{`m`a.M \Redlmn `m`a.M' }
 \\
\ruleitem{nu abstraction} &
\Inf	{M \Redlmn M' 
	}{`Nx.M \Redlmn `Nx.M' }
 \\
 \end{array}
 \quad
 \begin{array}{rl}
\ruleitem{application} & 
\Inf	{M \Redlmn M' \quad N \Redlmn N' 
	}{`@ M N \Redlmn `@ M' N' } 
 \\
\ruleitem{negated application} &
\Inf	{M \Redlmn M' \quad N \Redlmn N' 
	}{ [M] N \Redlmn [M'] N' }
 \\
\ruleitem{naming} & 
\Inf	{M \Redlmn M' 
	}{[`a] M \Redlmn [`a] M' } 
 \\
\ruleitem{renaming} & 
\Inf	{M \Redlmn `M`a . M' 
	}{[`b] M \Redlmn M' \tsubst[`b/`a] }
 \end{array}
 \quad
 \begin{array}{rl}
\ruleitem{beta reduction} & 
\Inf	{M \Redlmn `lx.M' \quad N \Redlmn N' 
	}{`@ M N \Redlmn M' \tsubst[N'/x] } 
 \\
\ruleitem{mu reduction} & 
\Inf	[`g\textit{ fresh}]
	{M \Redlmn `M`a . M' \quad N \Redlmn N' 
	}{`@ M N \Redlmn `M`g . M' \tsubst[{N'}.`g/`a] }
 \\
\ruleitem{nu reduction} &
\Inf	{M \Redlmn `Nx.M' \quad N \Redlmn N' 
	}{`@ [M] N \Redlmn M' \tsubst[N'/x] }
 \\
\ruleitem{delta reduction} &
\Inf	{M \Redlmn `M`a . M' \quad N \Redlmn N' 
	}{ [M] N \Redlmn M \tsubst[N'/`a] } 
 \end{array} \] 
We write $ M \Redlmn N $ if the statement $ M \Rightarrow N $ is derivable using these rules.


 \end{definition}
It is easy to check that a term parallel reduces to itself, and under parallel reduction a term is considered to be in normal form if it only reduces to itself. 

Notice, in particular, the change in the rule based on $`b$-reduction, which changes from 
 \[
 \begin{array}{c@{\quad \textrm{ to }\quad }c}
\Inf	{ M \Arrow M' \quad N \Arrow N' 
	}{`@ (`lx.M) N \Arrow M' \tsubst[N'/x] } 
& \def \Redlmn {\,\Arrow\,}
\Inf	{M \Redlmn `lx.M' \quad N \Redlmn N' 
	}{`@ M N \Redlmn M' \tsubst[N'/x] } 
 \end{array} \]
It is this change that solves the problem mentioned.

 \begin{example}
The problem signalled in \Exm~\ref{Parigot wrong} does not occur, since the diverging reduction steps
 \[ \def \Redlmn {\,\Arrow\,}
\Inf	[10]
	{\InfBox{ `m`a.[`a]`m`b.P \Redlmn `m`a.[`a]`m`b.P }
	\dquad
	\InfBox{ N \Redlmn N }
	}{ `@ (`m`a.[`a]`m`b.P) N \Redlmn `m`g.[`g](`m`b.P \tsubst[N.`g/`a])N }
\quad
\Inf	[5]
	{\Inf	[3]
		{\Inf	[7]
			{\InfBox{ `m`b.P \Redlmn `m`b.P }
			}{ [`a]`m`b.P \Redlmn P \tsubst[`a/`b] }
		}{ `m`a.[`a]`m`b.P \Redlmn `m`a.P \tsubst[`a/`b] }
	 \dquad
	 \InfBox{ N \Redlmn N }
	}{ `@ (`m`a.[`a]`m`b.P) N \Redlmn `@ (`m`a.P \tsubst[`a/`b]) N }
 \]
can be joined:
 \[ \def \Redlmn {\,\Arrow\,}
\Inf	[3]
	{\Inf	[7]
		{\Inf	[10]
			{\InfBox{ `m`b.P \tsubst[N.`g/`a] \Redlmn `m`b.P \tsubst[N.`g/`a] }
			 \dquad
			 \InfBox{ N \Redlmn N }
			}{ (`m`b.P \tsubst[N.`g/`a])N \Redlmn `m`d.P \tsubst[N.`g/`a] \tsubst[N.`d/`b] }	
		}{ [`g](`m`b.P \tsubst[N.`g/`a])N \Redlmn P \tsubst[N.`g/`a] \tsubst[N.`d/`b] \tsubst[`g/`d] }	
	}{ `m`g.[`g](`m`b.P \tsubst[N.`g/`a])N \Redlmn `m`g.P \tsubst[N.`g/`a] \tsubst[N.`d/`b] \tsubst[`g/`d] }
 \] \[ \def \Redlmn {\,\Arrow\,}
\Inf	[10]
	{\InfBox{ `m`a.P \tsubst[`a/`b] \Redlmn `m`a.P \tsubst[`a/`b] }
	 \dquad
	 \InfBox{ N \Redlmn N }
	}{ `@ (`m`a.P \tsubst[`a/`b]) N \Redlmn `m`g.P \tsubst[`a/`b] \tsubst[N.`g/`a] }
 \]
(notice that $ `m`g.P \tsubst[`a/`b] \tsubst[N.`g/`a] = `m`g.P \tsubst[N.`g/`a] \tsubst[N.`d/`b] \tsubst[`g/`d] $).
 \end{example}

It is straightforward to show that $\rtcredlmn$ is the transitive closure of $\Redlmn$.

 \begin{lemma} \label{Bar 3.2.7}
$ {\rtcredlmn} = {\rtcRedlmn} $.
 \end{lemma}
 \begin{proof}
First, since $ M \Redlmn M $, for all $M$, by the presence of rules $(9)$, $(10)$, $(11)$, and $(12)$, we have $ {\eqredlmn} \subseteq {\Redlmn} $, so also $ {\rtcredlmn} \subseteq {\rtcRedlmn} $.
Since in $\Redlmn$ we essentially contract any number of $\redlmn$-redexes in parallel (including zero or just one) we also have that $ {\Redlmn} \subseteq {\rtcredlmn} $.
So in particular, $ \Redlmn$ is a subset of a relation that is transitive, so its transitive closure is that as well, so $ {\rtcRedlmn} \subseteq {\rtcredlmn} $.
So $ {\rtcredlmn} = {\rtcRedlmn} $.
\qed
 \end{proof}

The following property expresses that the four kinds of substitution are respected by $\Redlmn$. 

 \begin{lemma}[Substitution Lemma] \label{running inside substitution} 
If $ P \Redlmn P' $ and $ Q \Redlmn Q' $, then:
 \begin{flatenumerate}
 \item $ P\tsubst[Q/z] \Redlmn P'\tsubst[Q'/z] $,
 \item $ P\tsubst[Q.`g /z] \Redlmn P'\tsubst[{Q'}.`g /z] $,
 \item $ P\tsubst[Q/`a] \Redlmn P'\tsubst[Q'/`a] $, and 
 \item $ P\tsubst[`b/`a] \Redlmn P'\tsubst[`b/`a] $.
 \end{flatenumerate}
 \end{lemma}
 \begin{proof}
 \begin{enumerate}
 \item 
By induction on the definition of $\Redlmn$, where we focus on the first parallel reduction.
 \begin{description} \itemsep 3\point

 \item [$(1)$] 
 $ \begin{array}[t]{lll}
 P\tsubst[Q/z] = z \tsubst[Q/z] = Q \Redlmn Q' = z \tsubst[Q'/z] = P'\tsubst[Q'/z], \textrm{ and} 
 \\
 P\tsubst[Q/z] = y \tsubst[Q/z] = y \Redlmn y = y\tsubst[Q'/z] = P'\tsubst[Q'/z] \textrm{ if } y \not= z .
 \end{array} $

 \item [$(9)$] 
 Then $ MN \Redlmn M'\tsubst[N'/x] $ follows from $M \Redlmn `lx.M'$ and $N \Redlmn N'$. 
By induction, we have $M \tsubst[Q/z] \Redlmn (`lx.M') \tsubst[Q'/z]$ and $N \tsubst[Q/z] \Redlmn N' \tsubst[Q'/z] $.
So we can infer:
 \[
\Inf	[9]
	{M \tsubst[Q/z] \Redlmn (`lx.M') \tsubst[Q'/z]
	 \dquad
	 N \tsubst[Q/z] \Redlmn N' \tsubst[Q'/z] 
	}{ `@ M\tsubst[Q/z] N\tsubst[Q/z] \Redlmn M'\tsubst[Q'/z] \tsubst[{N'\tsubst[Q'/z]}/x] }
 \]
By Lemma~\ref{term subst lemma}, we have $ M'\tsubst[Q'/z] \tsubst[{N'\tsubst[Q'/z]}/x] = M' \tsubst[N'/x] \tsubst[Q'/z] $.

 \item [$(10)$] 
 Then $ MN \Redlmn `m`g.M'\tsubst[{N'}.`g/`a] $ follows from $M \Redlmn `m`a.M'$ and $N \Redlmn N'$. 
By induction, $M \tsubst[Q/z] \Redlmn (`m`a.M') \tsubst[Q'/z]$ and $N \tsubst[Q/z] \Redlmn N' \tsubst[Q'/z] $.
So we can infer:
 \[
\Inf	[10]
	{M \tsubst[Q/z] \Redlmn (`m`a.M') \tsubst[Q'/z]
	 \dquad
	 N \tsubst[Q/z] \Redlmn N' \tsubst[Q'/z] 
	}{ `@ M\tsubst[Q/z] N\tsubst[Q/z] \Redlmn `m`g.M'\tsubst[Q'/z] \tsubst[{N'\tsubst[Q'/z]}.`g/`a] }
 \]
We have $ `m`g.M'\tsubst[Q'/z] \tsubst[{N'\tsubst[Q'/z]}.`g/`a] = `m`g.M' \tsubst[{N'}.`g/`a] \tsubst[Q'/z] $ by Lemma~\ref{term subst lemma}.

 \item [$(11)$] 
 Then $ [M]N \Redlmn M'\tsubst[N'/x] $ follows from $M \Redlmn `nx.M'$ and $N \Redlmn N'$. 
By induction, $M \tsubst[Q/z] \Redlmn (`nx.M') \tsubst[Q'/z]$ and $N \tsubst[Q/z] \Redlmn N' \tsubst[Q'/z] $.
So we can infer:
 \[
\Inf	[9]
	{M \tsubst[Q/z] \Redlmn (`nx.M') \tsubst[Q'/z]
	 \dquad
	 N \tsubst[Q/z] \Redlmn N' \tsubst[Q'/z] 
	}{ `@ [{M\tsubst[Q/z]}] N\tsubst[Q/z] \Redlmn M'\tsubst[Q'/z] \tsubst[{N'\tsubst[Q'/z]}/x] }
 \]
By Lemma~\ref{term subst lemma}, we have $ M'\tsubst[Q'/z] \tsubst[{N'\tsubst[Q'/z]}/x] = M' \tsubst[N'/x] \tsubst[Q'/z] $.

 \item [$(12)$] 
 Then $ [M]N \Redlmn M'\tsubst[N'/`a] $ follows from $M \Redlmn `m`a.M'$ and $N \Redlmn N'$. 
By induction, $M \tsubst[Q/z] \Redlmn (`m`a.M') \tsubst[Q'/z]$ and $N \tsubst[Q/z] \Redlmn N' \tsubst[Q'/z] $.
So we can infer:
 \[
\Inf	[12]
	{M \tsubst[Q/z] \Redlmn (`m`a.M') \tsubst[Q'/z]
	 \dquad
	 N \tsubst[Q/z] \Redlmn N' \tsubst[Q'/z] 
	}{ `@ [{M\tsubst[Q/z]}] N\tsubst[Q/z] \Redlmn M'\tsubst[Q'/z] \tsubst[{N'\tsubst[Q'/z]}/`a] }
 \]
By Lemma~\ref{term subst lemma}, we have $ M'\tsubst[Q'/z] \tsubst[{N'\tsubst[Q'/z]}/`a] = M' \tsubst[N'/`a] \tsubst[Q'/z] $.

 \end{description}
The other cases all follow by induction.

\item ~\kern-3mm, $(\textit{3})$ and $(\textit{4})$ Very similar. \qed

 \end{enumerate}
 \end{proof}

The following property expresses the interaction between the syntactic structure of terms and $\Redlmn$.

 \begin{proposition} \label{the shape of things} 
 \begin{enumerate}
 \item \label {lambda abstraction case} If $ `Lx.M \Redlmn L $, then $ L \same `lx.N $ and $ M \Redlmn N $. 
 \item \label {mu abstraction case} If $ `m`a.M \Redlmn L $, then 
	$ L \same `M`a.N $ and $ M \Redlmn N $. 
 \item \label {nu abstraction case} If $ `Nx.M \Redlmn L $, then $ L \same `Nx.N $ and $ M \Redlmn N $. 
 \item \label{application case} If $ `@ M N \Redlmn L $, then either:
	\begin{enumerate}
	\item $ L \same `@ P Q $ with $ M \Redlmn P $ and $ N \Redlmn Q $, or 
	\item $ M \Redlmn `lx.P $, and $ L = P\tsubst[Q/x] $ with $ N \Redlmn Q $, or 
	\item $ M \Redlmn `m`a.P $, and $ L = `m`g.P\tsubst[Q.`g/`a] $ with $ N \Redlmn Q $. 
 \end{enumerate}

 \item \label{negation application case} If $ `@ [M] N \Redlmn L $, then either:
	\begin{enumerate}
	\item $ L \same `@ [P] Q $ with $ M \Redlmn P $ and $ N \Redlmn Q $, or 
	\item $ M \Redlmn `Nx.P $, and $ L = P\tsubst[Q/x] $ with $ N \Redlmn Q $, or 
	\item $ M \Redlmn `m`a.P $, and $ L = P\tsubst[Q/`a] $ with $ N \Redlmn Q $. 
	\end{enumerate}

 \item \label{naming case} If $ [`a] M \Redlmn L $, then either:
	\begin{enumerate}
	\item $ L \same [`a] P $ with $ M \Redlmn P $, or 
	\item $ L \same P\tsubst[`a/`b] $ with $ M \Redlmn \Mu `b . P $. 
	\end{enumerate}

 \end{enumerate}
 \end{proposition}
 \begin{proof}
Straightforward by the definition of $\Redlmn$.
\qed
 \end{proof} 

We now show that $\Redlmn$ satisfies the diamond property.
We will write `$\Conv P_1 => P_3 <= P_2 $' for `$ P_1 \Redlmn P_3 $ and $ P_2 \Redlmn P_3 $'.

 \begin{theorem} \label{Bar 3.2.6}
If $ P_0 \Redlmn P_1 $ and $ P_0 \Redlmn P_2 $ then there exists a $P_3$ such that $ \Conv P_1 => P_3 <= P_2 $.
 \end{theorem}
 \begin{proof}
	\def \Redlmn {\mathrel{\semcolour {\Rightarrow}}}
By induction on the definition of $\Redlmn$, where we focus on the first parallel reduction.
We only show the interesting cases.

 \begin{description} \itemsep3\point

 \item [$(1)$] 
Then $ P_0 \same x \Redlmn x \same P_1 $, and $ P_2 = x $; take $ P_3 = x $ as well. 

 \item [$(2), (3), (4)$]
By induction.

 \item [$(5)$] 
Then $ P_0 \same `@ M_0 N_0 \Redlmn `@ M_1 N_1 \same P_1 $ because $ M_0 \Redlmn M_1 $ and $ N_0 \Redlmn N_1 $.
By Proposition~\ref{the shape of things}\sk(\ref{application case}), either: 

	\begin{description}

	\item [$ P_2 = M_2 N_2 $, with $ M_0 \Redlmn M_2 $ and $ N_0 \Redlmn N_2 $]
By induction there exists $M_3$, $N_3$ such that $ \Conv M_1 => M_3 <= M_2 $ and $ \Conv N_1 => N_3 <= N_2 $.
Take $ P_3 = `@ M_3 N_3 $.

	\item[$ P_2 \same M_2\tsubst N_2/x $ with $ M_0 \Redlmn `lx.M_2 $ and $ N_0 \Redlmn N_2 $]
By induction there exists $M_3$, $N_3$ such that $ \Conv{M_1}{M_3}{`lx.M_2}$, and $ \Conv N_1 => N_3 <= N_2 $; by Proposition~\ref{the shape of things}\sk(\ref{lambda abstraction case}), $M_3 = `lx.M_3'$, and $M_2 \Redlmn M_3'$.
By Rule~9, we have $ `@ M_1 N_1 \Redlmn M_3' \tsubst[N_3/x] $, and by Lemma~\ref{running inside substitution}, we have $ M_2\tsubst[N_2/x]\Redlmn M_3'\tsubst N_3/x $.

	\item[$P_2 \same `m`g.M_2\tsubst N_2.`g/`a $ with $ M_0 \Redlmn `m`a.M_2 $ and $ N_0 \Redlmn N_2 $]
By induction there exists \\ $M_3$, $N_3$ such that $ \Conv{M_1}{M_3}{`m`a.M_2}$, and $ \Conv N_1 => N_3 <= N_2 $; by Proposition~\ref{the shape of things}\sk(\ref{mu abstraction case}), $M_3 = `m`a.M_3'$, and $M_2 \Redlmn M_3'$.
By Rule~10, we have $ `@ M_1 N_1 \Redlmn `m`g.M_3' \tsubst[{N_3}.`g/`a] $, and by Lemma~\ref{running inside substitution}, we have $ `m`g.M_2\tsubst{N_2}.`g/`a \Redlmn `m`g.M_3'\tsubst{N_3}.`g/`a $. 

	\end{description}

 \item [$(6)$] 
Then $ P_0 \same `@ [M_0] N_0 \Redlmn `@ [M_1] N_1 \same P_1 $ because $ M_0 \Redlmn M_1 $ and $ N_0 \Redlmn N_1 $.
By Proposition~\ref{the shape of things}\sk(\ref{negation application case}), either: 

	\begin{description}

	\item [{$ P_2 \same [M_2]N_2 $ with $ M_0 \Redlmn M_2 $ and $ N_0 \Redlmn N_2 $}]
By induction there exists $M_3$, $N_3$ such that $ \Conv M_1 => M_3 <= M_2 $ and $ \Conv N_1 => N_3 <= N_2 $.
Take $ P_3 = `@ [M_3] N_3 $.

	\item[$ P_2 \same M_2\tsubst N_2/x $ with $ M_0 \Redlmn `nx.M_2 $ and $ N_0 \Redlmn N_2 $]
By induction there exists $M_3$, $N_3$ such that $ \Conv{M_1}{M_3}{`nx.M_2}$, and $ \Conv N_1 => N_3 <= N_2 $; by Proposition~\ref{the shape of things}\sk(\ref{lambda abstraction case}), $M_3 = `nx.M_3'$, and $M_2 \Redlmn M_3'$.
By Rule~(11), we have $ `@ [M_1] N_1 \Redlmn M_3' \tsubst[N_3/x] $ and by Lemma~\ref{running inside substitution}, we have $ M_2\tsubst[N_2/x]\Redlmn M_3'\tsubst N_3/x $.

	\item[$P_2 = M_2\tsubst N_2/`a $ with $ M_0 \Redlmn `m`a.M_2 $ and $ N_0 \Redlmn N_2 $]
By induction there exists $M_3$, $N_3$ such that $ \Conv{M_1}{M_3}{`m`a.M_2}$, and $ \Conv N_1 => N_3 <= N_2 $; by Proposition~\ref{the shape of things}\sk(\ref{mu abstraction case}), $M_3 = `m`a.M_3'$, and $M_2 \Redlmn M_3'$.
By Rule~(12), we have $ `@ [M_1] N_1 \Redlmn M_3' \tsubst[N_3/`a] $.
By Lemma~\ref{running inside substitution}, we have $ M_2 \tsubst[N_2/`a] \Redlmn M_3'\tsubst{N_3}.`g/`a $. 

	\end{description}

 \item [$(7)$] 
Then $ P_0 \same [`b] M_0 \Redlmn [`b]M_1 \same P_1 $ because $ M_0 \Redlmn M_1 $.
By Proposition~\ref{the shape of things}\sk(\ref{naming case}), either: 

	\begin{description}

	\item [{$ P_2 \same [`b] M_2 $ with $ M_0 \Redlmn M_2 $}] 
By induction, there exists $M_3$ such that $\Conv M_1 => M_3 <= M_2 $; then by Rule $(7)$ also 
$ \Conv [`b]M_1 => [`b]M_3 <= [`b]M_2 $. 

	\item [$ P_2 \same M_2\tsubst`b/`a $ with $ M_0 \Redlmn \Mu `a . M_2 $]
By induction, there exists $M_3$ such that 
\linebreak$ M_1 \SEarrow M_3 \SWarrow $ $ M_2 $; then by Proposition~\ref{the shape of things}\sk(\ref{mu abstraction case}), $ M_3 \same `m`a.M_3' $, $ M_1 \same `m`a.M_1' $. 
By Lemma~\ref{running inside substitution} we have that $ M_2 \tsubst[`b/`a] \Redlmn M_3' \tsubst[`b/`a] $. 
Since $ `m`a.M_1' \Redlmn `m`a.M_3' $, by Rule $(7)$ also 
$ [`b]`m`a.M_1' \Redlmn M_3' \tsubst[`b/`a] $. 

	\end{description}

 \item [$(8)$] 
Then $ P_0 \same [`b] M_0 \Redlmn M_1 \tsubst[`b/`a] \same P_1 $ because $ M_0 \Redlmn `m`a.M_1 $.
By Proposition~\ref{the shape of things}\sk(\ref{naming case}), either: 

	\begin{description}

	\item [{$ P_2 \same [`b] M_2 $ with $ M_0 \Redlmn M_2 $}] 
By induction, there exists $M_3$ such that 
$ `m`a.M_1 \SEarrow M_3 $ $ \SWarrow M_2 $; then by Proposition~\ref{the shape of things}\sk(\ref{mu abstraction case}), $ M_3 \same `m`a.M_3' $ and $ M_2 \same `m`a.M_2' $. 
By Lemma~\ref{running inside substitution} we have that $ M_1 \tsubst[`b/`a] \Redlmn M_3' \tsubst[`b/`a] $. 
Since $ `m`a.M_2' \Redlmn `m`a.M_3' $, by Rule $(7)$ also 
$ [`b]`m`a.M_2' \Redlmn M_3' \tsubst[`b/`a] $. 

	\item [$ P_2 \same M_2\tsubst`b/`a $ with $ M_0 \Redlmn \Mu `a . M_2 $]
By induction, there exists $M_3$ such that 
\linebreak$ `m`a.M_1 \SEarrow M_3 $ $ \SWarrow `m`a.M_2 $; then by Proposition~\ref{the shape of things}\sk(\ref{mu abstraction case}), $ M_3 \same `m`a.M_3' $, and\linebreak$ \Conv M_1 => M_3' <= M_2 $.
Then by Lemma~\ref{running inside substitution}, also $ \Conv{ M_1 \tsubst[`b/`a] }{ M_3' \tsubst[`b/`a] }{ M_2 \tsubst[`b/`a] } $.

	\end{description}

 \item [$(9)$] 
Then $ P_0 \same `@ M_0 N_0 \Redlmn M_1 \tsubst[N_1/x] \same P_1 $ because $ M_0 \Redlmn `lx.M_1 $ and $ N_0 \Redlmn N_1 $.
By Proposition~\ref{the shape of things}\sk(\ref{application case}), either: 

	\begin{description}

	\item [$ P_2 \same M_2 N_2 $ with $ M_0 \Redlmn M_2 $ and $ N_0 \Redlmn N_2 $] 
By induction there exists $M_3$, $N_3$ such that $ \Conv `lx.M_1 => M_3 <= M_2 $, and $ \Conv N_1 => N_3 <= N_2 $; then by Proposition~\ref{the shape of things}\sk(\ref{lambda abstraction case}), $M_2 \same `lx.M_2'$ and $M_3 \same `lx.M_3'$ and $ \Conv M_1 => M_3' <= M_2' $.
Since $ M_2 \Redlmn `lx.M_3' $ and $ N_2 \Redlmn N_3 $, by Rule $(9)$, $ M_2N_2 \Redlmn M_3'\tsubst[N_3/x] $. 
We have $ M_1\tsubst[N_1/x] \Redlmn M_3'\tsubst[N_3/x] $ by Lemma~\ref{running inside substitution}. 

	\item[$ P_2 = M_2\tsubst N_2/x $ with $ M_0 \Redlmn `lx.M_2 $ and $ N_0 \Redlmn N_2 $]
By induction there exists $M_3$, $N_3$ such that $ \Conv{`lx.M_1}{M_3}{`lx.M_2} $, and $ \Conv N_1 => N_3 <= N_2 $; then by Proposition~\ref{the shape of things}\sk(\ref{lambda abstraction case}), $M_3 = `lx.M_3'$, and $\Conv M_1 => M_3' <= M_2 $.
Then $ \Conv M_1\tsubst[N_1/x] => M_3'\tsubst[N_3/x] <= M_2\tsubst[N_2/x] $ follows by Lemma~\ref{running inside substitution}. 

	\item[$ P_2 = `m`g.M_2\tsubst {N_2}.`g/`a $ with $ M_0 \Redlmn `m`a.M_2 $ and $ N_0 \Redlmn N_2 $]
By induction there exists \\ $M_3$ such that $ \Conv `lx.M_1 => M_3 <= `m`a.M_2 $; this is impossible.

	\end{description}

 \item [$(10)$] 
Then $ P_0 \same `@ M_0 N_0 \Redlmn `m`g.M_1 \tsubst[{N_1}.`g/`a] \same P_1 $ because $ M_0 \Redlmn `m`a.M_1 $ and $ N_0 \Redlmn N_1 $.
By Proposition~\ref{the shape of things}\sk(\ref{application case}), either: 

	\begin{description}

	\item [$ P_2 \same M_2 N_2 $ with $ M_0 \Redlmn M_2 $ and $ N_0 \Redlmn N_2 $] 
By induction there exist $M_3$, $N_3$ such that \\ $ \Conv `m`a.M_1 => M_3 <= M_2 $, and $ \Conv N_1 => N_3 <= N_2 $; then by Proposition~\ref{the shape of things}\sk(\ref{mu abstraction case}), $M_2 \same `m`a.M_2'$ and $M_3 \same `m`a.M_3'$ and $ \Conv M_1 => M_3' <= M_2' $.
Since $ M_2 \Redlmn `m`a.M_3' $ and $ N_2 \Redlmn N_3 $, by Rule $(10)$, $ M_2N_2 \Redlmn `m`g.M_3'\tsubst[{N_3}.`g/`a] $, and $ `m`g.M_1\tsubst[{N_1}.`g/`a] \Redlmn `m`g.M_3'\tsubst[{N_3}.`g/`a] $ follows by Lemma~\ref{running inside substitution}. 

	\item[$ P_2 = M_2\tsubst N_2/x $ with $ M_0 \Redlmn `lx.M_2 $ and $ N_0 \Redlmn N_2 $]
By induction $M_3$ exists such that \\ $ \Conv{`m`a.M_1}{M_3}{`lx.M_2} $; this is impossible.

	\item[$ P_2 = `m`g.M_2\tsubst {N_2}.`g/`a $ with $ M_0 \Redlmn `m`a.M_2 $ and $ N_0 \Redlmn N_2 $]
By induction there are $M_3$, $N_3$ such that $ \Conv `m`a.M_1 => M_3 <= `m`a.M_2 $, and $ \Conv N_1 => N_3 <= N_2 $; then $M_3 = `m`a.M_3'$, and 
$ M_1 \SEarrow $ $ M_3' \SWarrow M_2 $ 
by Proposition~\ref{the shape of things}\sk(\ref{lambda abstraction case}).
Then 
$ M_1\tsubst[{N_1}.`g/`a] $ $\SEarrow $ $ M_3'\tsubst[{N_3}.`g/`a] $ $ \SWarrow $ $ M_2\tsubst[{N_3}.`g/`a] $
follows by Lemma~\ref{running inside substitution}. 

	\end{description}

 \item [$(11)$] 
Then $ P_0 \same `@ [M_0] N_0 \Redlmn M_1 \tsubst[N_1/x] \same P_1 $ because $ M_0 \Redlmn `nx.M_1 $ and $ N_0 \Redlmn N_1 $.
By Proposition~\ref{the shape of things}\sk(\ref{negation application case}), either: 

	\begin{description}

	\item [{$ P_2 \same [M_2] N_2 $ with $ M_0 \Redlmn M_2 $ and $ N_0 \Redlmn N_2 $}] 
Then by induction there exists $M_3$, $N_3$ such that $ \Conv `nx.M_1 => M_3 <= M_2 $, and $ \Conv N_1 => N_3 <= N_2 $; then by Proposition~\ref{the shape of things}\sk(\ref{nu abstraction case}), $M_2 \same `nx.M_2'$ and $M_3 \same `nx.M_3'$ and $ \Conv M_1 => M_3' <= M_2' $.
Since $ M_2 \Redlmn `nx.M_3' $ and $ N_2 \Redlmn N_3 $, by Rule $(11)$, $ [M_2]N_2 \Redlmn M_3'\tsubst[N_3/x] $. 
We have $ M_1\tsubst[N_1/x] \Redlmn M_3'\tsubst[N_3/x] $ by Lemma~\ref{running inside substitution}. 

	\item[$ P_2 = M_2\tsubst N_2/x $ with $ M_0 \Redlmn `nx.M_2 $ and $ N_0 \Redlmn N_2 $]
By induction there exists $M_3$, $N_3$ such that $ \Conv{`nx.M_1}{M_3}{`nx.M_2} $, and $ \Conv N_1 => N_3 <= N_2 $; then by Proposition~\ref{the shape of things}\sk(\ref{nu abstraction case}), $M_3 = `nx.M_3'$, and $\Conv M_1 => M_3' <= M_2 $.
Then $ \Conv M_1\tsubst[N_1/x] => M_3'\tsubst[N_3/x] <= M_2\tsubst[N_2/x] $ follows by Lemma~\ref{running inside substitution}. 

	\item[$ P_2 = `m`g.M_2\tsubst {N_2}.`g/`a $ with $ M_0 \Redlmn `m`a.M_2 $ and $ N_0 \Redlmn N_2 $]
By induction there exists \\ $M_3$ such that $ \Conv `nx.M_1 => M_3 <= `m`a.M_2 $; this is impossible.

	\end{description}

 \item [$(12)$] 
Then $ P_0 \same `@ [M_0] N_0 \Redlmn M_1 \tsubst[N_1/`a] \same P_1 $ because $ M_0 \Redlmn `m`a.M_1 $ and $ N_0 \Redlmn N_1 $.
By Proposition~\ref{the shape of things}\sk(\ref{negation application case}), either: 

	\begin{description}

	\item [{$ P_2 \same [M_2] N_2 $ with $ M_0 \Redlmn M_2 $ and $ N_0 \Redlmn N_2 $}] 
Then by induction there exists $M_3$, $N_3$ such that $ \Conv `m`a.M_1 => M_3 <= M_2 $, and $ \Conv N_1 => N_3 <= N_2 $; then by Proposition~\ref{the shape of things}\sk(\ref{nu abstraction case}), $M_2 \same `m`a.M_2'$ and $M_3 \same `m`a.M_3'$ and $ \Conv M_1 => M_3' <= M_2' $.
Since $ M_2 \Redlmn `m`a.M_3' $ and $ N_2 \Redlmn N_3 $, by Rule $(11)$, $ [M_2]N_2 \Redlmn M_3'\tsubst[N_3/`a] $. 
We have $ M_1\tsubst[N_1/`a] \Redlmn M_3'\tsubst[N_3/`a] $ by Lemma~\ref{running inside substitution}. 

	\item[$ P_2 = M_2\tsubst N_2/x $ with $ M_0 \Redlmn `nx.M_2 $ and $ N_0 \Redlmn N_2 $]
By induction there exists $M_3$ \\ such that $ \Conv `m`a.M_1 => M_3 <= `nx.M_2 $; this is impossible.

	\item[$ P_2 = M_2\tsubst N_2/`a $ with $ M_0 \Redlmn `m`a.M_2 $ and $ N_0 \Redlmn N_2 $]
By induction there exists $M_3$, $N_3$ such that $ \Conv `m`a.M_1 => M_3 <= `m`a.M_2 $, and $ \Conv N_1 => N_3 <= N_2 $; then by Proposition~\ref{the shape of things}\sk(\ref{nu abstraction case}), $M_3 = `m`a.M_3'$, and $\Conv M_1 => M_3' <= M_2 $; then $ \Conv M_1\tsubst[N_1/`a] => M_3'\tsubst[N_3/`a] <= M_2\tsubst[N_2/`a] $ follows by Lemma~\ref{running inside substitution}.%
\qed

	\end{description}

 \end{description}
 \end{proof}

We can now state our main result.

 \begin{theorem}[Confluence] 
Reduction in $\redlmn$ is confluent.
 \end{theorem}
 \begin{proof}
By Theorem~\ref{Bar 3.2.6}, we have that $\Redlmn$ satisfies the diamond property, and by Lemma~\ref{Bar 3.2.7} that $ \rtcredlmn $ is the transitive closure of $\Redlmn$. 
Then by Definition~\ref{CR def}, $\redlmn$ is confluent.
\qed
 \end{proof}

\section{Representing \texorpdfstring{$\vdash_{\textsc{ni}}$}{⊢NI} in \texorpdfstring{$\vdash_{\cal L}$}{⊢L}} 

\label{representation}
In this section we will show that all statements provable in $\TurnNI$ have a witness in $\Turnlmn$.
We achieve this result by first defining a mapping for terms from $\nlm$ to $\lmn$; this will deal with a necessary transformation of derivations when establishing a relation between typeability in $\nlm$ and $\lmn$.
What we use here is the transformation from $\Turnlmn$ to $\Turnnlm$:
 \[ \def\Turnlmn{\Turn} \def\Turnnlm{\Turn}
\Inf	[`m]
	{\Inf	{\Inf	[\negE]
			{\InfBox{ \dernlm `G |- M : \neg\neg C }
			 \quad
			 \Inf	[\Ax]{ \dernlm y`:\neg C |- y : \neg C }
			}{\dernlm `G |- [M]y : `B } 
		}{\InfBox<84>{ \dernlm `G |- \Cont[{`@ [M] y }] : `B } }
	}{\dernlm `G |- `My.\Cont[{[M]y}] : C } 
\dquad \textrm{\normalsize into} \dquad
\Inf	[`m]
	{\Inf	{\Inf	[\negE]
			{\InfBox{ \derlmn `G |- M : \neg\neg C }
			 \quad
			 \Inf	[\negI]
				{\Inf	[\Name]
					{\Inf	[\Ax]{ \derlmn y`:C |- y : C }
					}{ \derlmn y`:C,`a`:\neg C |- { [`a] y } : `B }
				}{ \derlmn `a`:\neg C |- { `Ny.[`a]y} : \neg C }
			}{ \derlmn `G |- {`@ [M] (`Ny.[`a]y) } : `B } 
		}{\InfBox<128>{ \derlmn `G |- {\Cont[{`@ [M] (`Ny.[`a]y) }]} : `B } }
	}{\derlmn `G |- `M`a.{\Cont[{`@ [M] (`Ny.[`a]y) }]} : C }
 \]
Remark that in the first, there is no subterm that has type $C$, whereas in the second, there is.
So we can deal with $(\PbC)$ towards an assumption that is not on the left.

 \begin{definition}
We define a mapping $\Sem{`.}_V : \nlm\arrow \cal L $ inductively over terms.
 \[ \begin{array}{rcl@{}l}
\Sem{x}_V &=& `N x.[`a] x & (`a\For x \ele V) \\
\Sem{x}_V &=& x & (x \notele V) \\
\Sem{`M x . M}_V &=& `M `a . \Sem{M}_{V,`a\For x} 
 \end{array} \dquad \begin{array}{rcl}
\Sem{`Lx.M}_V &=& `Lx.\Sem{M}_V \\
\Sem{`@ M N }_V &=& `@ {\Sem{M}_V} \Sem{N}_V \\
\Sem{`N x.M}_V &=& `N x.\Sem{M}_V \\
\Sem{`@ [M] N }_V &=& `@ [\Sem{M}_V] \Sem{N}_V \\
 \end{array} \]
and define
 \[ \begin{array}{rcl@{\quad}l}
\UL{x} &=& `{} \\
\UL{`lx.M} = \UL{`nx.M} &=& \UL{M} \\
\UL{MN} = \UL{[M]N} &=& \UL{M} \Union \UL{N} \\
\UL{`mx.M} &=& \UL{M} \Union \Set { `a\For x } & (`a \textit{ fresh}) \\
 \end{array} 
 \dquad
 \begin{array}{rcl@{\quad}l}
\Sem{`G,x`:A}_V &=& \Sem{`G}_V,`a`:A & (`a\For x \ele V) \\
\Sem{`G,x`:A}_V &=& \Sem{`G}_V,x`:A & (`a\For x \notele V) \\
\Sem{`{}}_V &=& `{} 
 \end{array} \]
 \end{definition}
Remark that, if $`a\For x \ele V$ then $x$ was bound under $`m$, so in a derivation will have a negated type.
Notice that \emph{all} occurrences of term variables that occur in $V$ are replaced by $\Sem{`.}_V$, even if they appear on the left in $(\negE)$: that this is not problematic, can be illustrated by the following:
 \[ \begin{array}[t]{lclclcl}
\Sem{`my. `@ [y] M }_{`e} &=&
`M `a . {`@ [\Sem{y}_{`a/y}] \Sem{M}_{`a/y} } &=& 
`M `a . {`@ [{`N y.[`a] y }] \Sem{M}_{`a/y} } &\redlmn& 
`M `a . {`@ [`a] \Sem{M}_{`a/y} }
 \end{array} \]
so the substitutions on the left-hand side do not affect the result, but just create a slightly more complicated proof than would be necessary.

We can now show a representation result, which essentially shows that, although the inference rules of $\Turnnlm$ and $\Turnlmn$ differ significantly in their applicability of the rule $(`m)$, which represents the proof rule $(\PbC)$, they can witness the same results in $\TurnNI$.
This result does not establish a rule-to-rule mapping of the correspondence between the systems, but states that logical judgements that are provable in one system are also provable in the other.
We know that every provable judgement in $\Turnnlm$ corresponds directly to a provable statement in $\TurnNI$, and vice versa, and with the correspondence we show here, we get that this also holds between $\Turnlmn$ and $\TurnNI$.

We first establish a relation between typeability in $\nlm$ and $\lmn$.

 \begin{theorem}
 \begin{enumerate}
 \item
If $ \derlmn `G |- M : A $, then $ \dernlm `G |- M : A $.
 \item
If $ \dernlm `G |- M : A $, and $V = \UL{M}$, then $ \derlmn \Sem{`G}_V |- \Sem{M}_V : A $.
 \end{enumerate}
 \end{theorem}
 \begin{proof}
 \begin{enumerate}

 \item
Since, once allowing Greek characters for variables as well, rule $(\Name)$ can be omitted and $\Turnlmn$ is a sub-inference system of $\Turnnlm$.

 \item
By induction on the definition of $\Turnnlm$.

 \begin{description}

 \item[$(\Ax)$]
Then $`G = `G',x`:A$; we have two cases:

 \begin{description}
 
 \item[$`a\For x \ele V$]
Then $A = \neg B$, $`G = `G',x`:\neg B$, so $\Sem{`G}_V = \Sem{`G'}_V,`a`:\neg B $ and $ \Sem{x}_V = `Nx.[`a]x $.
We can derive:
 \[ \def\Turnnlm{\Turn}
\Inf	[\negI]
	{\Inf	[\Name]
		{\Inf[\Ax]{ \dernlm \Sem{`G'}_V,x`:B |- x : B }
		}{ \dernlm \Sem{`G'}_V,`a`:\neg B,x`:B |- [`a]x : `B }
	}{ \dernlm \Sem{`G'}_V,`a`:\neg B |- `Nx.[`a]x : \neg B }
 \]

 \item[$x \notele V$]
Then $`G = `G',x`:A$, so $ 
x`:A \ele \Sem{`G}_V $ and $ \Sem{x}_V = x $; the result follows by rule $(\Ax)$.

 \end{description}


 \item[$(`m)$]
Then $ M = `mx.N $, $ \Sem{M}_V = `M`a.\Sem{N}_{V,`a/x} $ and the derivation for $ \dernlm `G |- M : A $ is shaped like:
 \[ \def \Turnnlm{\Turn}
\Inf	[`m]
	{\InfBox{ \dernlm `G,x`:\neg A |- N : `B }
	}{ \dernlm `G |- `Mx.N : A }
 \]
By induction we have $ \derlmn \Sem{`G}_V,`a`:\neg A |- \Sem{N}_{V,`a/ x} : `B $; the result follows by rule $(`m)$.


\item[$(\arrI)$, $(\arrE)$, $(\negI)$, and $(\negE)$]
Straightforward by induction.
\qed

 \end{description}
 \end{enumerate}
 \end{proof}
Moreover, we now have have that every provable judgement in $\TurnNI$ can be inhabited in $\Turnlmn$.

 \begin{corollary} \label{from NI to lmn}
If $ \derNI `G |- A $, if and only if there exist $`G'$ and $M \ele \lmn$ such that $\OL{`G'} = `G$ and $ \derlmn `G' |- M : A $.
 \end{corollary}

We will illustrate the expressiveness of $\Turnlmn$.

 \begin{example}
We can witness $ (\neg B\arrow \neg A)\arrow A\arrow B $ in $\Turnlmn$ (let $`G = x`:\neg B\arrow \neg A, y`:A, `a`:\neg B $):
 \[ \def\Turnlmn{\Turn}
\Inf	[\arrI]
	{\Inf	[\arrI]
		{\Inf	[`m]
			{\Inf	[\negE]
				{\Inf	[\arrE]
					{\Inf	[\Ax]{\derlmn `G |- x : \neg B\arrow \neg A }
					 \Inf	[\negI]
						{\Inf	[\Name]
							{\Inf	[\Ax]{ \derlmn `G,z`:B |- z : B }
							}{ \derlmn `G,z`:B |- { [`a] z } : `B }
						}{ \derlmn `G |- `Nz.[`a]z : \neg B }
				 	}{\derlmn `G |- {`@ x (`Nz.[`a]z) } : \neg A }
				 \Inf	[\Ax]{\derlmn `G |- y : A }
				}{\derlmn `G |- {`@ [{`@ x (`Nz.[`a]z) }] y } : `B }
			}{\derlmn x`:\neg B\arrow \neg A,y`:A |- {`M`a . {`@ [{`@ x (`Nz.[`a]z) }] y }} : B }
		}{\derlmn x`:\neg B\arrow \neg A |- {`Ly.`M`a . {`@ [{`@ x (`Nz.[`a]z) }] y }} : A\arrow B }
	}{\derlmn {} |- {`Lxy.`M`a . {`@ [{`@ x (`Nz.[`a]z) }] y }} : (\neg B\arrow \neg A)\arrow A\arrow B }
 \]

We can show Mendelson's Axiom 3 in $\TurnNI$.
Let $`G = \neg B\arr \neg A, \neg B\arr A $
 \[
\Inf	[\arrI]
	{\Inf	[\arrI]
		{\Inf	[\PbC]
			{\Inf	[\negE]
				{\Inf	[\arrE]
					{\Inf	[\Ax]{ \derlog `G,\neg B |- \neg B\arr \neg A }
					 \Inf	[\Ax]{ \derlog `G,\neg B |- \neg B }	
					}{ \derlog `G,\neg B |- \neg A }
				 \Inf	[\arrE]
					{\Inf	[\Ax]{ \derlog `G,\neg B |- \neg B\arr A }
					 \Inf	[\Ax]{ \derlog `G,\neg B |- \neg B }	
					}{ \derlog `G,\neg B |- A }
				}{ \derlog `G,\neg B |- `B }
			}{ \derlog `G |- B }
 		}{ \derlog \neg B\arr \neg A |- ( \neg B\arr A )\arr B }
	}{ \derlog {} |- (\neg B\arr \neg A )\arr ( \neg B\arr A )\arr B }
 \]
This proof gets represented in $\nlm$ by the term $`Lxy.`Mz.{`@ [`@ x z ] (`@ y z ) }$.

Interpreting this into $\lmn$ gives (where $`G = x`:\neg B\arr \neg A , y`:\neg B\arr A , `a`:\neg B $):
 \[ \def \Turnlmn{\Turn}
\Inf	[\arrI]
	{\Inf	[\arrI]
		{\Inf	[`m]
			{\Inf	[\negE]
				{\Inf	[\arrE]
					{\Inf	[\Ax]{ \derlmn `G |- x : \neg B\arr \neg A }
					 \Inf	[\negI]
						{\Inf	[\Name]
							{\Inf	[\Ax]{ \derlmn `G,z`:B |- z : B }
					 		}{ \derlmn `G,z`:B |- [`a]z : `B }
					 	}{ \derlmn `G |- `Nz.[`a]z : \neg B }	
					}{ \derlmn `G |- {`@ x (`Nz.[`a]z) } : \neg A }
				 \Inf	[\arrE]
					{\Inf	[\Ax]{ \derlmn `G |- y : \neg B\arr A }
					 \Inf	[\negI]
						{\Inf	[\Name]
							{\Inf	[\Ax]{ \derlmn `G,z`:B |- z : B }
					 		}{ \derlmn `G,z`:B |- [`a]z : `B }
					 	}{ \derlmn `G |- `Nz.[`a]z : \neg B }	
					}{ \derlmn `G |- {`@ y (`Nz.[`a]z) } : A }
				}{ \derlmn `G |- {`@ [{`@ x (`Nz.[`a]z) }] ({`@ y (`Nz.[`a]z) }) } : `B }
			}{ \derlmn `G{\setminus} `a |- `M`a.{`@ [{`@ x (`Nz.[`a]z) }] ({`@ y (`Nz.[`a]z) }) } : B }
 		}{ \derlmn x`:\neg B\arr \neg A |- `Ly.`M`a.{`@ [{`@ x (`Nz.[`a]z) }] ({`@ y (`Nz.[`a]z) }) } : ( \neg B\arr A )\arr B }
	}{ \derlmn {} |- `Lxy.`M`a.{`@ [{`@ x (`Nz.[`a]z) }] ({`@ y (`Nz.[`a]z) }) } : (\neg B\arr \neg A )\arr ( \neg B\arr A )\arr B }
 \]

 \end{example}

\def \dlceil{{\lceil}\kern-3\point{\lceil}\kern-1\point}
\def \drfloor{\kern-1\point{\rfloor}\kern-3\point{\rfloor}}
 \def \Redparam(#1){\textit{Red} \, (#1)}
 \def \Red{\ifnextchar(
 	{\Redparam}{\textit{Red}}}
 \def \Inter {\cap}
 \def \Comp#1{\llceil #1 \rrfloor}
 \def \Compp<#1;#2;#3>{\textit{Comp}`<#1`;{#2 }`;#3`>}%
 \def \Compr(#1){\textit{Comp}`(#1`)}%
 \def \Comp(#1){\textit{Comp}\,(#1)}%
 \def \Comp{\ifnextchar(
 	{\Compr}{\Compp}}
 \def \derComp(#1 |- #2 : #3 ){ \derlmn #1 |- { #2 } : #3 \And \Comp(#2) }
 \def \derSN(#1 |- #2 : #3 ){
\SN(#2) }%
 \def \Vtsubst[#1/#2]{\tqsk\{\nhsk\Vect{#1\For #2}\nhsk\}}
 \def \mref#1{\,\textit{\ref{#1}}}

 \section{Strong normalisation} \label{strong normalisation}
\def \SNsymb{\mathcal{S\kern-1\point N}}
\def \SNpar(#1){\SNsymb\psk({#1})}
\def \SN{\ifnextchar(
	{\SNpar}{\SNsymb}}

In this section we shall prove that every term typeable in $\Turnlmn$ is strongly normalisable; this will be done using the technique of \emph{reducibility candidates}, as first defined by Girard \cite{Girard'71}, based on work by Tait \cite{Tait'67}.
We will follow Parigot's application \cite{Parigot'97} of the reducibility method, but adapted to negation and applications like $ [M] N $.

 \begin{definition}
 \begin{enumerate}
 \item
 We use $\SN(M)$ to express that $M$ is strongly normalisable (all reduction paths starting from $M$ are of finite length), and $\SN = \Set{M \ele \lmn \mid \SN(M) }$.
 \item
As in \cite{Parigot'97}, we write $V^{`F}$ for the set of all finite sequences of elements of $V$, with $`e$ representing the empty sequence, and use the notation $\Vect{v}$ for elements of $V^{`F}$.
 \end{enumerate}
 \end{definition}

 \begin{proposition} \label {SN props}
The following properties hold of $\SN$:

 \begin{enumerate} 
 
 \item \label {SN head application}
$\SN(xM_1\dots M_n)$ (with $n\geq 0$) and $\SN(M')$ if and only if $\SN(xM_1\dots M_nM')$.

 \item \label {SN L abstraction}
If $\SN(M)$, then $\SN(`lx.M)$.

 \item \label {SN N abstraction}
If $\SN(M)$, then $\SN(`nx.M)$.

 \item \label {SN M abstraction}
If $\SN(M)$ then $\SN(`m`a.M)$.

 \item \label{SN Mx}
If $\SN(Mx)$ then $\SN(M)$.

 \item \label {SN lambda redex}
If $\SN(M\tsubst[N/x]\Vect{P})$ and $\SN(N)$, then 

 \item \label {SN nu redex}
If $\SN(M\tsubst[N/x]\Vect{P})$ and $\SN(N)$, then $\SN({`@ [`nx.M] N \Vect{P}})$.

 \item \label {SN mu redex}
If $\SN(M\tsubst[N.`g/`a]\Vect{P})$ and $\SN(N)$, then $\SN({`@ (`m`a.M) N \Vect{P}})$.

 \item \label {SN delta redex}
If $\SN(M\tsubst[N/`a]\Vect{P})$ and $\SN(N)$, then $\SN({`@ [`m`a.M] N \Vect{P}})$.

 \item \label {SN renaming redex}
If $\SN(M[`b/`a])$, then $\SN({`@ [`b] (`m`a.M) })$.

 \end{enumerate} 
 \end{proposition}

\def \Arrow {\mathrel{\mbox{\large $\rightarrowtail$}}}
 \def \rightrightharpoons{\setbox11=\hbox{$\rightharpoonup$}%
{\rightharpoondown}\kern-\wd11\raise1.5\point\box11}
 \def \Arrow {\mathrel{\rightrightharpoons}}

The idea is to assign to each type $A$ a set of strongly normalisable terms $\Red(A)$, and to show that every term typeable with $A$ is an element of $\Red(A)$, and thereby strongly normalisable.
Central to the predicate is the notion of functional construction, which states that a set is in $\Red(A\arr B)$, if it is contained of terms that return terms in $\Red(B)$ when applied to terms in $\Red(A)$. 

The latter is expressed through the following:

 \begin{defiC}[\cite{Parigot'97}]
 \begin{enumerate}
 \item 
The \emph{functional construction} $\Arrow : \pow \lmn \Prod \pow \lmn \arrow \pow \lmn $ is defined through:
 \[ \begin{array}{rcl}
A \Arrow B &\ByDef& \Set{ M \ele \lmn \mid \Forall N\ele A [ `@ M N \ele B ] }
 \end{array} \]

 \item
$\Arrow$ is generalised to $\Arrow^{`F}$ through:
 \[ \begin{array}{rcl}
A \Arrow^{`F} B &\ByDef& \Set{ M \ele \lmn \mid \Forall \Vect{N}\ele A^{`F} [ `@ M \Vect{N} \ele B ] }
 \end{array} \]
 \end{enumerate}
 \end{defiC}

Using functional construction, reducibility candidates can easily be defined.

 \begin{defiC}[\cite{Parigot'97}] \label{Red definition}
The set $\Red \subseteq \pow \lmn $ of \emph{reducibility candidates} is inductively defined through: 
 \begin{enumerate}
 \item $\SN \ele \Red$;
 \item If $ A \ele \Red $ and $ B \ele \Red $, then $ A \Arrow B \ele \Red $.
 \end{enumerate}
We write $\Red(A)$ for $A \ele \Red$.
 \end{defiC}

The next result states that all terms that are reducible in $A$ are also strongly normalisable, and that all variables are reducible in any type.

 \begin{lemC}[\cite{Parigot'97}] \label{Red implies SN} \label{vars are Red}
If $\Red(A)$, then $ A \subseteq \SN $ and $A$ contains the $`l$-variables.
 \end{lemC}
 \begin{proof}
We prove 
 \begin{flatenumerate}
 \item
$ A \subseteq \SN $ and 
 \item 
for all $\Vect{N} \ele \SN^{`F}$, $x \Vect{N} \ele A $
 \end{flatenumerate}
by induction on the definition of \Red.
 \begin{enumerate}

 \item 
 \begin{description}
 \item[$ A = \SN$] Immediate.
 \item[$ A = B \Arrow C $] 
Take $M \ele A$, then 
$ \Forall N\ele A [ `@ M N \ele B ] $, so $ \Forall N\ele A [ \SN(`@ M N )] $ by induction~(1). 
Take the $`l$-variable $x$, then by induction~(2), $x \ele B$ and therefore $M x \ele C $, so by induction $ \SN(M x)$, and therefore by Proposition~\ref{SN props}\sk(\ref{SN Mx}) $\SN(M)$. 
 \end{description}

 \item 
 \begin{description}
 \item[$ A = \SN$] By Proposition~\ref{SN props}\sk(\ref{SN head application}).
 \item[$ A = B \Arrow C $] 
Let $ N' \ele B $, then $\SN(N')$ by induction~(1).
Take $\Vect{\SN(N_i)}$, then by Proposition~\ref{SN props}\sk(\ref{SN head application}) $ \SN(x \Vect{N} N') $, and $ x \Vect{N}N' \ele C $ by induction~(2).
So $ x \Vect{N} \ele A $.
\qed 

 \end{description}
 \end{enumerate}
 \end{proof}

 \begin{lemC}[\cite{Parigot'97}]
If $\Red(A)$, then there exists $B \subseteq \SN^{`F}$ such that $ A = B \Arrow^{`F} \SN $.
 \end{lemC}
 \begin{proof}
By induction on the definition of \Red.

 \begin{description}
 \item [$ A = \SN $] Notice that $ \SN \ByDef \Set {`e} \Arrow^{`F} \SN $.

 \item [$ A = C \Arrow D $] 
By induction we have $ D = E \Arrow^{`F} \SN $ for some $E$, and therefore $ A = F \Arrow^{`F} \SN $ where  
 $ \begin{array}{rcl}
F &=& \Set {M \Vect{N} \mid M \ele C, \Vect{N} \ele E }.
 \end{array} $
\qed

 \end{description}
 \end{proof}

 \begin{definition} \label{Red to SN}
For every $A \ele \Red$, $A^{`B}$ is defined as the greatest $ B \subseteq \SN^{`F} $ such that $ A = B \Arrow^{`F} \SN $.
 \end{definition}

 \noindent
Notice that, since $A \subseteq \SN$, if $M \ele A$, then $M \ele \SN$ and $M `e \ele \SN$, so $`e \ele A^{`B}$.

Parigot remarks that membership of $`e$ is essential. 
He says ``\emph{It allows to go from an arbitrary reducibility candidate $A$ to $\SN$ and back, without knowing anything about $A$. This property is used for the case of the rule where one switches from one formula to another. Contrary to the case of a $`l$ where one knows that the type is an arrow, in the one of a $`m$ one has an arbitrary type, but it can be considered as some kind of arrow $ A^{`B} \Arrow^{`F} \SN $ whose number of arguments is unknown (possibly zero).}'' \cite{Parigot'97} (notation adapted).

It is perhaps worthwhile to point out that, so far, there is no relation between typeability and reducibility, in that $\Red(A)$ does not just contain terms that are typeable with $A$; for example, the term $xx$ is in $\Red(A)$, for any $A$, since $ xx \ele \SN$, but is not typeable in $\Turnlmn$.
However, assume that term $M$ is typeable with $A$, then we have a derivation that shows $ {\cal D} : \derlmn `G |- M : A $ and for every free variable $x$ in $M$ there will a type $B$ such that $x`:B \ele `G$, and $\cal D$ contains occurrences of rule $(\Ax)$ showing $ \derlmn `G' |- x : B $.
We will be capable of proving that then replacing $x$ in $M$ by elements of $\Red(B)$, for every free variable of $M$, creates an element of $\Red(A)$ and this will be sufficient for our purposes.

In \cite{Parigot'97}, Parigot shows termination for typeable terms in $\lmu$ enriched with quantification rules; in order to deal with the binding of type variables, he defines a notion of type interpretation that maps type variables onto reducible sets, extended naturally to quantification and (using functional construction) to arrow types.
Here we do not need to deal with quantification, but find it convenient to continue on his path.

 \begin{definition} \label{interpretation}
An \emph{interpretation} $`x$ is a function from type variables to \Red.
Interpretations are extended to arbitrary formulas by:
 \[ \begin{array}{rcl}
`x (A \arrow B) &=& `x (A) \Arrow `x (B) \\
`x (\neg A) &=& `x (A) \Arrow `x (`B) \\
`x (`B) &=& \SN 
 \end{array} \]
 \end{definition}
 
We shall now prove our strong normalisation result by showing that every term typeable with $A$ is reducible in that type.
For this, we need to prove a stronger property:
we will now show that if we replace term-variables by reducible terms in a typeable term, then we obtain a reducible term.

 \begin {lemma} [Replacement Lemma] \label {Replacement property}
 \begin{itemize}\itemsep 2\point 
 \item Let $`G = \Set{ x_1{:}B_1 ,\ldots, x_n{:}B_n, `a_1`:\neg C_1, \ldots, `a_m`:\neg C_m }$.
 \item Let for $\iotn$, $N_i \ele `x(B_i)$, and 
for all $\jotm$, $\Vect{L}_j \ele `x(C_j)^{`B}$ if $C_j = D_j \arr E_j$, and
$L_j \ele `x(D_j)$ if $C_j = \neg D_j$.
 \item Let $ \tsubst[Q_j?/`a] $ stand for $ \tsubst[{\Vect{Q}_j}.`g_j/`a_j] $ if $C_j = (D_j\arr E_j)$, or $ \tsubst[Q_j/`a] $ if $C_j = \neg D_j$.
 \end{itemize}
\vspace*{2mm}

If $\derlmn `G |- M : A $, then $ M \Vtsubst[N_i/x_i] \Vtsubst[L_i?/`a_i] \ele `x(A) $. 
 \end {lemma}

 \begin{proof}
By induction on the structure of derivations. We will use $`S$ for $ \Vtsubst[N_i/x_i] \Vtsubst[L_i?/`a_i] $.

 \def \Turnlmn{\Turn}

 \begin{description} \itemsep 5\point

 \item [$(\Ax)$]
Then $M \same x_i$, for some $\jotn$, $ B_i = A$, and $M`S \same x_i`S \same N_i$.
From the second assumption we have that $N_i \ele `x(A)$.

 \item[$(\arrI)$]
Then $M = `lx.N$, $ A = F\arr G$ and $\derlmn `G,x`:F |- N : G $.
Let $ P \ele `x(F) $,  then by Lemma~\ref{Red implies SN} $\SN(P)$, and by induction $ N\tsubst[P/x]`S \ele `x(G)$.
Let $ \Vect{Q} \ele `x(G)^{`B}$, then $ N\tsubst[P/x]`S\Vect{Q} \ele \SN$ by Definition~\ref{Red to SN}.
Then by Proposition~\ref{SN props}\sk(\ref{SN lambda redex}) also $ (`lx.N)P `S\Vect{Q} \ele \SN$.
Therefore, by Definition~\ref{Red to SN}, $ (`lx.N)P `S \ele `x(G) $ and since $`S$ does not affect $P$, also $ `@ (`lx.N)`S P \ele `x(G) $; then by Definition~\ref{Red definition} and \ref{interpretation}, $ (`lx.N) `S \ele `x(F\arr G) $.

 \item[$(\arrE)$]
Then $M = PQ$ and there exists $F$ such that $\derlmn `G |- P : F\arr A $ and $\derlmn `G |- Q : F $.
By induction, $ P `S \ele `x(F\arr A)$ and $ Q `S \ele `x(F)$; by Definition~\ref{Red definition} and \ref{interpretation} we have $ `@ P`S\, Q`S \ele `x(A)$, and $ `@ P`S\, Q`S \same (`@ P Q )`S $.

 \item [$(\negI)$] 
Then $M\same `ny.P$, $ A = \neg F $, and $\derlmn `G,y`:F |- P : `B $.
Assume $Q \ele  `x(F)$, then by induction, $ P`S\tsubst[Q/y] \ele `x(`B)$, so by Definition~\ref{interpretation}, $\SN(P`S\tsubst[Q/y])$.
Then by Proposition~\ref{SN props}\sk(\ref{SN nu redex}), we have $ \SN(`@ [`Ny.P `S] Q)$, so by Definition~\ref{interpretation}, $ `@ [`Ny.P `S] Q \ele `x(`B)$, so by Definition~\ref{interpretation} $`Ny.P`S \ele `x(\neg F)$, and $ `Ny.P`S \same (`Ny.P)`S $.

 \item [$(\negE)$] 
Then $A = `B$, $M\same `@ [P] Q $, and there exists $F$ such that $\derlmn `G |- P : \neg F $ and $\derlmn `G |- Q : F $.
Then, by induction, we have $P`S \ele `x(\neg F)$ and $Q`S \ele `x(F)$.
Then by Definition~\ref{Red definition} and \ref {interpretation}, $`@ [{P`S}] Q`S \ele `x(`B)$, and $ `@ [{P`S}] Q`S \same (`@ [P] Q )`S $.

 \item[$(\Name)$] 
Then $M = [`a_j]N$ with $\jotm$, $A = `B$, and $\derlmn `G |- N : C_j $ with $`a`:\neg C_j \ele `G$.
By induction, $ N`S \ele `x(C_j) $.
Now either: 
 
 \begin{description}
 \item[$C_j = D_j\arr E_j$]
Notice that $ \tsubst[{\Vect{L}_j}.`g_j/`a_j] \ele `S $ and therefore $ ([`a_j]N)`S = [`g_j]N`S \Vect{L}_j $.
We have $ \Vect{L}_j \ele `x(C_j)^{`B} $ by assumption, and therefore by Definition~\ref{Red to SN}, $N`S \Vect{L}_j \ele \SN$, so also $[`g_j](N`S)\Vect{L}_j \ele \SN$; then by Definition~\ref{interpretation}, $ [`g_j](N`S)\Vect{L}_j \ele \Red(`B) $.

 \item[$C_j = \neg F$]
Now $V_j \ele `x(F)$ by assumption, and therefore by Definition~\ref{interpretation}, $ `@ [N`S] V_j \ele \Red(`B)$, and since $ \{Q/`a\} \ele `S$, also $ `@ [N`S] V_j = ([`a]N) `S $.

 \end{description}

 \item[$(`m)$]
Then $M = `m`b.N$ and $\derlmn `G,`b`:\neg A |- N : `B $.
Now either:
 
 \begin{description}
 \item[$A = D\arr E$] 
Let $ \Vect{Q} \ele `x(A)^{`B}$, then by induction $ N \tsubst[{\Vect{Q}}.`g/`a]`S \ele `x(`B) $ and by Definition~\ref{Red definition} and \ref {interpretation}, $ N \tsubst[{\Vect{Q}}.`g/`a]`S \ele \SN $.
Then by Proposition~\ref{SN props}\sk(\ref{SN mu redex}), $ (`m`a.N)`S\Vect{Q} \ele \SN $, so $ (`m`a.N)`S \ele `x(A) $.

 \item[$A = \neg D$] 
Assume $ Q \ele `x(D) $, then by induction $ P`S\tsubst[Q/`a] \ele `x(`B) $ and by Definition~\ref{Red definition} and \ref {interpretation}, $ P`S\tsubst[Q/`a] \ele \SN $.
Then by Proposition~\ref{SN props}, we have $ `@ [`M`a.P`S] Q \ele \SN $, so by Definition~\ref {interpretation} $ `@ [`M`a.P`S] Q \ele `x(`B) $, so by Definition~\ref{Red definition} $ (`M`a.P)`S \ele `x(\neg D) $.
\qed

 \end{description}
 \end{description}
 \end{proof}
 
We can now prove the main result.

 \begin{theorem}[Strong Normalisation]
 \label{sn result}
Any term typeable in `~$\Turnlmn$' is strongly normalisable.
 \end{theorem}
 \begin{proof}
Let $`G = x_1{:}B_1 ,\ldots, x_n{:}B_n, `a_1`:\neg C_1, \ldots, `a_m`:\neg C_m $ such that $\derlmn `G |- M : A $.
Then by Lemma \ref{vars are Red}, for all $\iotn$, $x_i \ele `x(B_i)$ and $`e \ele `x(C_j)^{`B}$.
Then, by Lemma~\ref {Replacement property},\linebreak$M\,\{\Vect{x_i/x_i}\}\{\Vect{`e?/`a_j }\} \ele \Red(A) $; strong normalisation for $M$ then follows from Lemma~\ref {Red implies SN}.\qed
 \end{proof}

 \def \ptL {\textit{pt}\kern-1\point_{\Lsubscr}\hspace{.75mm}}
 \def \LPair <#1;#2>{`<#1`;#2`>}
 \newqsymbol{`S}{\typesubst}

\setbox21=\hbox{\large $\cal L$}
 \section{Principal typing for \,\texorpdfstring{\copy21}{L}} \label{pp}
In this section, we will show that we can define a notion of principal typing for $\Turnlmn$.
This is achieved in the standard way: we define notions of type substitutions and unification, that are used for the definition of the algorithm $\ptL$ that calculates the principal typing for each term typeable in $\Turnlmn$. 

Substitution is shown to be sound, \emph{i.e.}~maps inferable judgements to inferable judgements, and the algorithm is shown to be complete in that all inferable judgements for a term can be constructed from its principal typing.
\begin{samepage}
 \begin {definition} [Substitution and unification] \label{Substitution and unification}
\label {Robinson unification definition}
 \begin {enumerate}

 \item
 \begin{enumerate}

 \item
The {\em substitution} $\subs$, where $`v$ is a type variable and $C$ a type, is inductively defined\footnote{All algorithmic definitions in this section are presented in `functional style', where calls are matched against the alternatives `top-down', the first match is taken, and the result is undefined in case there is no match.} by:
 \[ \begin{array}{lcl}
\subs `B &=& `B \\
\subs `v &=& C \\
\subs `v' &=& `v' \hfill (`v' \not= `v) \\
\subs A\arrow B &=& (\subs A) \rightarrow (\subs B) \\
\subs \neg A &=& \neg (\subs A) 
 \end{array} \]

 \item
If $`S_1$, $`S_2 $ are substitutions, then so is $`S_1 \after `S_2 $, where $`S_1 \after `S_2 {A} = `S_1 ({`S_2 {A} }) $.

 \item
$`S{`G} = \Set{ x`:`S{B}\mid x`:B \ele `G } \Union \Set{ `a`:`S{B}\mid `a`:B \ele `G } $.

 \item
$`S{ \Pair<`G,A> } $ = $ \Pair<`S{`G },`S{A}> $.

 \item
If there exists a substitution $`S$ such that $`S{A} = B$, then $B$ is a {\em (substitution) instance} of $A$.

 \item
$\idS$ is the identity substitution that replaces all type variables by themselves.
 
 \end{enumerate}

 \item 
Unification of types is defined by:
 \[ \begin{array}{lllcl}
\unifyw &`v & `v	& = & (`v \mapsto `v) \\
\unifyw &`v & B	& = & (`v \mapsto B) \quad (`v \textit{ does not occur in } B) \\
\unifyw & A & `v	& = & \unify {`v} {A} \\
\unifyw & (A\arrow B) & (C\arrow D)	& = & `S_2 \after `S_1	\\
\multicolumn{4}{r}{\WHERE} & \begin{array}[t]{rcl}
		`S_1	& = & \unify {A} {C} \\
		`S_2	& = & \unify (`S_1 {B}) (`S_1{D})
	 \end {array} \\
\unifyw & (\neg A) & (\neg C) & = & \unify {A} {C} 
 \end{array} \]

 \item
The operation $\unifyContexts$ generalises $\unify$ to contexts:
 \[ \begin{array}{lllcl}
\unifyContexts & (`G_1,x`:A) & (`G_2,x`:B) &=& `S_2 \after `S_1, \\
\multicolumn{4}{r}{\WHERE} & \begin{array}[t]{rcl}
		`S_1 &=& \unify {A} {B} \\
		`S_2 &=& \unifyContexts\ (`S_1 {`G_1})\ (`S_1 {`G_2})
 \end{array} \\
\unifyContexts & (`G_1,x`:A) & `G_2 &= &
		\unifyContexts\ `G_1\ `G_2 \quad (x \notele `G_2) \\
\unifyContexts & (`G_1,`a`:A) & (`G_2,`a`:B) &=& `S_2 \after `S_1, \\
\multicolumn{4}{r}{\WHERE} & \begin{array}[t]{rcl}
		`S_1 &=& \unify {A} {B} \\
		`S_2 &=& \unifyContexts\ (`S_1 {`G_1})\ (`S_1 {`G_2})
 \end{array} \\
\unifyContexts & (`G_1,`a`:A) & `G_2 &= &
		\unifyContexts\ `G_1\ `G_2 \quad (`a \notele `G_2) \\
\unifyContexts & \emptyset & `G_2 &=& \idS
 \end{array} \]

 \end {enumerate}
 \end {definition}
 \end{samepage}
This definition specifies $\unify$ as a partial function; if the side condition `$`v$ \textit{does not occur in} $B$' fails, no result is returned.
So, for example, `$ \unify {`v} {`v \arrow `v} $' or `$ \unify (A\arr B) {\,\neg (C\arr D)} $' does not return a substitution.

If successful, unification returns the most general unifier, as stated by:

 \begin {propC} [\cite {Robinson'65}] \label {Robinson unification properties}
For all $A$, $B$: if $`S_1 $ is a substitution such that $`S_1 {A} = `S_1 {B} $ (so then $`S_1$ is a \emph{unifier} of $A$ and $B$), then there exist substitutions $`S_2 $ and $`S_3 $ such that
$ \begin{array}{rcl}
`S_2 &=& \unify{A\ B} 
 \end{array} $ and 
$ \begin{array}{rcl@{\,}}
`S_1 &=& `S_3 \after `S_2
 \end{array} $.
 \end {propC}

 \begin {lemma}[Soundness of substitution] \label {derfailN Substitution sound}
If $\derlmn `G |- M : A $, then $\derlmn `S{`G } |- M : `S{A} $.
 \end {lemma}
 \begin{proof} 
By straightforward induction on the structure of derivations.\QED
 \end{proof}

We now define a notion of principal typing for terms of $\cal L$.

 \begin{definition}
The principal typing algorithm for $\Turnlmn$ is given in Figure~\ref{The algorithm ptL}.
\begin{figure} \small ~ \begin{minipage}{.49\textwidth}
 $ \ptL {x} ~=~ \LPair<x`:`v; `v> $ \\
\WHERE $ \begin{array}[t]{rcl} 
		`v & \textit{is} & \FreshVariable {} 
	 \end{array} $ \vspace*{1mm}

$ \ptL { `lx . M } ~=~ \LPair<`P;P> $ \\ 
\WHERE $ \begin{array}[t]{rcl}
		\LPair<`P'; P'> &=& \ptL {M} \\
		`P `; P &=& 
			\begin{cases}{@{}l@{~}l}
			`P' \Except x `; A\arrow P' & (x`:A \ele `P') \\
			 `P' `; `v\arrow P' & (x \notele `P') 
			\end{cases} \\
		`v &\textit{is}& \FreshVariable 
	 \end{array} $ \\ 

 $ \ptL { M N } ~=~ `S_2 \after `S_1 {\LPair<`P_1 \Union `P_2; `v>} $ \\
\WHERE $ \begin{array}[t]{rcl@{\quad}l}
		\LPair<`P_1; P_1> &=& \ptL {M}	\\
		\LPair<`P_2; P_2> &=& \ptL {N} \\
		`S_1 &=& \unify {P_1} {P_2\arrow `v} \\
		`S_2 &=& \unifyContexts\ (`S_1 {`P_1})\ (`S_1 {`P_2}) \\
		`v &\textit{is}& \FreshVariable
	 \end{array} $ \\ 

$ \ptL { `nx . M } ~=~ \LPair<`P;P> $ \\ 
\WHERE $ \begin{array}[t]{rcl}
		\LPair<`P'; `B> &=& \ptL {M} \\
		`P `; P &=& 
			\begin{cases}{@{}l@{~}l}
			`P' \Except x `; \neg A & (x`:A \ele `P') \\
			 `P' `; \neg `v & (x \notele `P') 
			\end{cases} \\
		`v &\textit{is}& \FreshVariable 
	 \end{array} $ 
\end{minipage}
\begin{minipage}{.49\textwidth}
 $ \ptL { [M] N } ~=~ `S_2 \after `S_1 {\LPair<`P_1 \Union `P_2; `B>} $ \\
\WHERE $ \begin{array}[t]{rcl@{\quad}l}
		\LPair<`P_1; P_1> &=& \ptL {M}	\\
		\LPair<`P_2; P_2> &=& \ptL {N} \\
		`S_1 &=& \unify {P_1} {\neg P_2} \\
		`S_2 &=& \unifyContexts\ (`S_1 {`P_1})\ (`S_1 {`P_2}) 
	 \end{array} $ \\ [3mm]

$ \ptL { `m`a . M } ~=~ \LPair<`P;P> $ \\ 
\WHERE $ \begin{array}[t]{rcl}
		\LPair<`P'; `B> &=& \ptL {M} \\
		`P `; P &=& 
			\begin{cases}{@{}l@{~}l}
			`P' \Except `a `; A & (`a`:\neg A \ele `P') \\
			 `P' `; `v & (`a \notele `P') 
			\end{cases} \\
		`v &\textit{is}& \FreshVariable 
	 \end{array} $ \\ [3mm]

 $ \ptL { [`a] N } ~=~ \LPair<`P ; `B> $ \\
\WHERE $ \begin{array}[t]{rcl@{\quad}l}
		\LPair<`P'; P'> &=& \ptL {N} \\
		`P &=& 
			\begin{cases}{@{}l@{~}l}
			`S {`P'} & (`a`:\neg A \ele `P') \\
			`P',`a`:\neg P' & (`a \notele `P') \\
			\end{cases} \\
		`S  &=& \unify {A} {P'} \\
	 \end{array} $ 

\end{minipage}
\caption{The algorithm $\ptL$} \label{The algorithm ptL}
 \end{figure}
 \end{definition}

We can show that the algorithm creates valid judgements:

 \begin{lemma}[Soundness of $\ptL$] \label{Soundness of ptFN}
If $\ptL {M} = \LPair<`P; P> $, then $\derlmn `P |- M : P $.
 \end {lemma}
 \begin{proof}
By induction on the structure of terms, using \Lmm\,\ref{derfailN Substitution sound}.
 \end{proof}

We will now show the main result for $\ptL$, which states that it calculates the most general typeing with respect to type substitution for all terms typeable in $\Turnlmn$.

 \begin{theorem} [Completeness of substitution.]
If $\derlmn `G |- M : A $, then there exists context $`P$, type $P$, and substitution $`S$ such that: $\ptL {M} = \LPair<`P;P> $, $`S{`P}\subseteq `G $, and $`S{P} = A $.
 \end{theorem}
 \begin{proof}
By induction on the structure of terms in $\cal L$.

 \begin{description} \itemsep 4\point

 \item [$M \same x$]
Then, by rule $(\Ax)$, $x`:A \ele `G $, and $\ptL {x} = \LPair<\Set{x`:`v};`v> $.
Take $`S = (`v \mapsto A) $.

 \item [$M \same `Lx . N $]
Then, by rule $(\arrI)$, there are $C, D$ such that $A = C\arrow D$, and $\derlmn `G,x`:C |- N : D $.
Then, by induction, there are $`P',P'$ and $`S'$ such that $\ptL {N} = \LPair<`P';P'> $, $`S'{`P'}\subseteq `G,x`:C $, and $`S{P'} = D$.
Then either:

 \begin{description}

 \item [$x \ele \fv(N)$]
Then $x`:C' \ele `P'$, and $\ptL {`Lx . N } = \LPair<`P' \Except x ; C'\arr P'> $.
Since $`S'{`P'}\subseteq `G,x`:C$, in particular $`S'{C'} = C$, $`S'{(`P' \Except x)}\subseteq `G $, and $`S'{(C'\arr P')} = C\arrow D$.
Take $`P = `P'\Except x$, $P = C'\arr P'$, and $`S = `S'$.

 \item [$x \notele \fv(N)$]
Then $\ptL {`Lx . N } = \LPair<`P'; `v\arr P'> $, $x$ does not occur in $`P'$, and let $`v$ not occur in $\LPair<`P';P'> $.
Since $`S'{`P'}\subseteq `G,x`:C$, in particular $`S'{`P'}\subseteq `G $.
Take $`S = `S' \after (`v \mapsto C)$, then, since $`v$ does not occur in $`P'$, also $`S{`P'}\subseteq `G $.
Notice that $`S (`v\arr P') = C\arrow D$; take $`P = `P'$ and $P = `v\arr P'$.

 \end{description}

 \item [$M \same Q R $] 
Then, by rule $(\arrE)$, there exists a $B$ such that $\derlmn `G |- Q : {B\arr A} $ and $\derlmn `G |- R : B $.
By induction, there are $`S_1, `S_2$, $\LPair<`P_1 ; P_1> = \ptL {Q} $ and $\LPair<`P_2 ; P_2> = \ptL {R} $ (no type variables shared) such that $`S_1 {`P_1} \subseteq `G $, $`S_2 {`P_2} \subseteq `G $, $`S_1 {P_1} = B\arr A$ and $`S_2 {P_2} = B$.
Notice that $`S_1, `S_2$ do not interfere.
Let $`v$ be a fresh type variable and
 \[ \begin{array}[t]{rcl}	
`S_u &=& \unify{P_1\ (P_2\arr `v)} \\
`S_{C} &=& \unifyContexts\ (`S_u {`P_1})\ (`S_u {`P_2}) \\
\ptL { Q R } &=& `S_{C} \after `S_u {\LPair<`P_1\Union `P_2; `v'_1\Union `D'_2>} 
\end{array} \]

We need to argue that $\ptL { QR } $ is successful: since this can only fail on calls to unification (of $P_1$ and $P_2\arr `v$, or in the unification of the contexts), we need to argue that these are successful.
Take $`S_3 = `S_2 \after `S_1 \after (`v \mapsto A) $, then 
 \[ \begin{array}{rcl}
`S_3 {P_1} &=& B\arr A, \textrm{ and} \\ 
`S_3 (P_2\arr `v) &=& B\arr A. 
 \end{array} \]
so $P_1$ and $P_2\arr `v$ have a common instance $B\arr A$, and by \Prop\,\ref{Robinson unification properties}, $`S_u$ exists.

Notice that we have
 \[ \begin{array}{rcl}
`S_3 {`P_1} &\subseteq& `G, \textrm{ and} \\ 
`S_3 {`P_2} &\subseteq& `G
 \end{array} \]
since $`P_1$ and $`P_2$ share no type-variables.
Since $`G$ is a context, each term variable has only one type, and therefore $`S_3$ is a unifier for $`P_1$ and $`P_2$, so we know that an $`S_4$ exists which extends the substitution that unifies the contexts, even after being changed with $`S_u$, so such that
 \[ \begin{array}{rcl}
`S_4 (`S_u {`P_1}) &\subseteq& `G, \textrm{ and} \\ 
`S_4 (`S_u {`P_2}) &\subseteq& `G. 
 \end{array} \]
So $`S_4$ also unifies $`S_u {`P_1}$ and $`S_u {`P_2}$, so by \Prop\,\ref{Robinson unification properties} there exists a substitution $`S_5$ such that $`S_4 = `S_5 \after `S_{`G} \after `S_u $.
Take $`S = `S_5$.

 \item [$M \same `Nx . N $]
Then, by rule $(\negI)$, there exists $C$ such that $A = \neg C$, and $\derlmn `G,x`:C |- N : `B $.
Then, by induction, there are $`P'$ and $`S'$ such that $\ptL {N} = \LPair<`P';`B> $, and $`S'{`P'}\subseteq `G,x`:C $.
Then either:

 \begin{description}

 \item [$x \ele \fv(N)$]
Then $x`:C' \ele `P'$, and $\ptL {`Nx . N } = \LPair<`P' \Except x ; \neg C'> $.
Since $`S'{`P'}\subseteq `G,x`:C$, in particular $`S'{C'} = C$, $`S'{(`P' \Except x)}\subseteq `G $, and $`S'{(\neg C')} = \neg C$.
Take $`P = `P'\Except x$, $P = \neg C'$, and $`S = `S'$.

 \item [$x \notele \fv(N)$]
Then $\ptL {`Lx . N } = \LPair<`P'; \neg `v> $, $x$ does not occur in $`P'$ where $`v$ does not occur in $\LPair<`P';P'> $.
Since $`S'{`P'}\subseteq `G,x`:C$, in particular $`S'{`P'}\subseteq `G $.
Take $`S = `S' \after (`v \mapsto C)$, then, since $`v$ does not occur in $`P'$, also $`S{`P'}\subseteq `G $.
Notice that $`S (\neg `v) = \neg C$; take $`P = `P'$ and $P = \neg `v$.

 \end{description}

 \item [{$M \same [Q] R $}] 
Then $A=`B$ and by rule $(\negE)$ there exists a $B$ such that $\derlmn `G |- Q : \neg B $ and $\derlmn `G |- R : B $.
By induction, there are $`S_1, `S_2$, $\LPair<`P_1 ; P_1> = \ptL {Q} $ and $\LPair<`P_2 ; P_2> = \ptL {R} $ (no type variables shared) such that $`S_1 {`P_1} \subseteq `G $, $`S_2 {`P_2} \subseteq `G $, $`S_1 {P_1} = \neg B$ and $`S_2 {P_2} = B$.
Notice that $`S_1, `S_2$ do not interfere.
Let $`v$ be a fresh type variable and
 \[ \begin{array}[t]{rcl}	
`S_u &=& \unify{P_1} {\neg P_2} \\
`S_{C} &=& \unifyContexts\ (`S_u {`P_1})\ (`S_u {`P_2}) \\
\ptL { Q R } &=& `S_{C} \after `S_u {\LPair<`P_1\Union `P_2; `B>} 
\end{array} \]

As for the case $M = QR$, take $`S_3 = `S_2 \after `S_1 \after (`v \mapsto A) $, then $ `S_3 {P_1} = \neg B, $ and $`S_3 {P_2} = B$, so $P_1$ and $\neg P_2$ have a common instance $\neg B$ and $`S_u$ exists.
Since also $ `S_3 {`P_1} \subseteq `G$, and $`S_3 {`P_2} \subseteq `G$, as above an $`S_4$ exists such that
$`S_4 (`S_u {`P_1}) \subseteq `G$, and $`S_4 (`S_u {`P_2}) \subseteq `G$ and by Proposition~\ref{Robinson unification properties} there exists a substitution $`S_5$ such that $`S_4 = `S_5 \after `S_{`G} \after `S_u $.
Take $`S = `S_5$.

 \item [$M \same `M`a . N $]
Then, by rule $(`m)$, $\derlmn `G,`a`:\neg A |- N : `B $.
Then, by induction, there are $`P'$ and $`S'$ such that $\ptL {N} = \LPair<`P';`B> $, and $`S'{`P'}\subseteq `G,`a`:\neg A $.
Then either:

 \begin{description}

 \item [$`a \ele \fv(N)$]
Then $`a`:\neg C \ele `P'$, and $\ptL {`m`a . N } = \LPair<`P' \Except `a ; C> $.
Since $`S'{`P'}\subseteq `G,`a`:\neg A$, in particular $`S'{C} = A$ and $`S'{(`P' \Except `a)}\subseteq `G $.
Take $`P = `P'\Except `a$, $P = C'$, and $`S = `S'$.

 \item [$`a \notele \fv(N)$]
Then $\ptL {`M`a . N } = \LPair<`P'; `v> $, $`a$ does not occur in $`P'$ where $`v$ does not occur in $\LPair<`P';P'> $.
Since $`S'{`P'}\subseteq `G,`a`:\neg A$, in particular $`S'{`P'}\subseteq `G $.
Take $`S = `S' \after (`v \mapsto A)$, then, since $`v$ does not occur in $`P'$, also $`S{`P'}\subseteq `G $.
Notice that $`S (`v) = A$; take $`P = `P'$ and $P = `v$.

 \end{description}

 \item [{$M \same [`a] N $}] 
Then $A=`B$ and by rule $(\Name)$ there exists a $B$ such that  $`a`:\neg B \ele `G $ and $\derlmn `G |- N : B $.
By induction, there exists $`S_1$, $\LPair<`P' ; P'> = \ptL {N} $ such that $`S_1 {`P'} \subseteq `G $,  and $`S_1 {P'} = B$.
Then either:

 \begin{description}

 \item [$`a \ele \fv(N)$]
Let $`a`:\neg C \ele `P'$; take $`S_2 = \unify{C} {P'} $, then $ \ptL [`a]N = \LPair<`S_2 {`P'}; `B> $, and $`a`:\neg `S_2 {C} \ele `S_2 {`P'} $.
Since $`a`:\neg C \ele `P'$ and $`S_1 {`P'} \subseteq `G $, we have that $`S_1 {\neg C} = \neg B $ and $`S_1 {P'} = B$, so $`S_2$ is successful and there exists $`S_3$ such that $`S_1 = `S_3 \after `S_2$, so $`S_1 {`P'} = `S_3 (`S_2 {`P'}) \subseteq `G $. 
Take $`S = `S_3$.

 \item [$`a \notele \fv(N)$]
Then $ \ptL [`a]N = \LPair<`P',`a`:\neg P'; `B> $; take $`S = `S_1$.
\QED

 \end{description}
 \end{description}
 \end{proof}
This last result shows the practicality of our notion of type assignment.

 \section*{Conclusion and Future Work}
We have presented $\lmn$ as an extension of Parigot's $\lmu$-calculus by adding negation as a type constructor with the aim of representing proofs in $\TurnNI$, Classical Logic with implication, negation, and Proof by Contradiction.
We gained a more expressive calculus, that no longer represents $\neg A$ through $A\arr `B$, but, more importantly, negation elimination is no longer represented by application, and negation introduction not by $`l$-abstraction, but through new syntactic constructs that represent negation introduction and elimination directly, thus getting a more faithful representation of proofs in $\TurnNI$.

We defined a notion of reduction that extends $\lmu$'s logical and structural reduction rules with two new reduction rules, one dealing with a $(\negI){-}(\negE)$-pair, the other when Proof by Contradiction gets applied against an assumption that has a double negated type.
We showed that type assignment is sound, in that assignable types are preserved under reduction.
Using, as suggested by Py, Aczel's generalisation of Tait and Martin-L\"of's notion of parallel reduction, we showed that reduction is confluent.

By its nature, not all proofs in $\TurnNI$ can be represented in $\lmn$, but we have shown a completeness result in that all propositions that can be shown in $\TurnNI$ have a witness in $\lmn$.
Following Parigot, using Girard's approach of reducibility candidates, we have shown that all typeable terms are strongly normalisable, and that type assignment for $\lmn$ enjoys the principal typing property.

In all, $\lmn$ satisfies all the properties that can be demanded for a calculus claiming to represent proofs and proof contraction in $\TurnNI$.
By representing negation directly, it also severely enlarges the size of the set of {\nef} terms (a syntactic notion, not dependent on assigned types) to those \emph{really} not representing negation.
All these are important issues for the creation of theorem provers for Classical Logic.

Our motivation for our work was two-fold: enlarge the set of \nef-terms, and the fact that (implicative) $\lmu$ does not fully represent negation.
It is now fair to ask: ``Is $\lmn$ the finished product?''
In other words, can all logical connectors be expressed in $\lmn$?
Although it is well known that conjunction $A \And B$ can be expressed through $ \neg (A\arrow \neg B) $ and disjunction $A \Or B$ through $ \neg A\arr B $, this is not enough, since it does not answer the question of representability.
In particular, it seems that rule $(\OrE)$ cannot be represented in $\TurnNI$, since that would require a combination of a limited version of $\LEM$ and $(\OrE)$ through the rule
 \[ 
\Inf	{ \derlog `G,A |- C \quad \derlog `G,\neg A |- C 
	}{ \derlog `G |- C }
 \]
We will leave this issue for future work, as well as the study of {\CBV} reduction for $\lmn$.

 \subsection*{Acknowledgements}
I am very grateful to David Davies for asking me the question: ``How does $\lmu$ deal with negation?''

 \bibliographystyle{alphaurl}
 \bibliography{LMCS}

 \end{document}